\documentclass[a4paper,USenglish,cleveref,autoref,thm-restate]{lipics-v2021}

\title{On Polynomial Kernels for Traveling Salesperson Problem and its Generalizations} %

\titlerunning{On Polynomial Kernels for TSP} %

\author{Václav Blažej}{Department of Theoretical Computer Science, Faculty of Information Technology,
\\Czech Technical University in Prague, Prague, Czech~Republic}
{vaclav.blazej@fit.cvut.cz}{https://orcid.org/0000-0001-9165-6280}{}

\author{Pratibha Choudhary}{Department of Theoretical Computer Science, Faculty of Information Technology,
\\Czech Technical University in Prague, Prague, Czech~Republic}
{pratibha.choudhary@fit.cvut.cz}{https://orcid.org/0000-0002-1648-288X}{}

\author{Dušan Knop}{Department of Theoretical Computer Science, Faculty of Information Technology,
\\Czech Technical University in Prague, Prague, Czech~Republic}
{dusan.knop@fit.cvut.cz}{https://orcid.org/0000-0003-2588-5709}{}

\author{Šimon Schierreich}{Department of Theoretical Computer Science, Faculty of Information Technology,
\\Czech Technical University in Prague, Prague, Czech~Republic}
{schiesim@fit.cvut.cz}{https://orcid.org/0000-0001-8901-1942}{}

\author{Ondřej Suchý}{Department of Theoretical Computer Science, Faculty of Information Technology,
\\Czech Technical University in Prague, Prague, Czech~Republic}
{ondrej.suchy@fit.cvut.cz}{https://orcid.org/0000-0002-7236-8336}{}

\author{Tomáš Valla}{Department of Theoretical Computer Science, Faculty of Information Technology,
\\Czech Technical University in Prague, Prague, Czech~Republic}
{tomas.valla@fit.cvut.cz}{https://orcid.org/0000-0003-1228-7160}{}

\authorrunning{V. Blažej, P. Choudhary, D. Knop, Š. Schierreich, O. Suchý, T. Valla} %

\Copyright{Václav Blažej, Pratibha Choudhary, Dušan Knop, Šimon Schierreich, Ondřej Suchý, Tomáš Valla} %



\category{} %

\relatedversion{} %

\funding{The authors acknowledge the support of the OP VVV MEYS funded project \\CZ.02.1.01/0.0/0.0/16\_019/0000765 ``Research Center for Informatics''. This work was additionally supported by the Grant Agency of the Czech Technical University in Prague, grant \mbox{No.~SGS20/208/OHK3/3T/18}.}

\nolinenumbers %

\hideLIPIcs

\EventEditors{John Q. Open and Joan R. Access}
\EventNoEds{2}
\EventLongTitle{42nd Conference on Very Important Topics (CVIT 2016)}
\EventShortTitle{CVIT 2016}
\EventAcronym{CVIT}
\EventYear{2016}
\EventDate{December 24--27, 2016}
\EventLocation{Little Whinging, United Kingdom}
\EventLogo{}
\SeriesVolume{42}
\ArticleNo{23}

\usepackage{bm}
\usepackage{tikz}
\usetikzlibrary{calc,scopes,backgrounds,decorations.pathmorphing,fit,math}

\def\resultbox<#1>(#2)(#3)(#4;#5)(#6;#7)(#8;#9)%
{
	\tikzstyle{cpxity} = [minimum height=0.4cm,inner sep=0pt,align=center,node distance=0];
	\def\totalwidth{3.85cm};
	\node[inner ysep=3pt,inner xsep=0pt,text width=\totalwidth,align=center] (#2) at (#3) {\small #1};
	\tikzmath{\boxwidth = \totalwidth/3;\nboxwidth = -\totalwidth/3;}
	\node[text width=\boxwidth,cpxity,#4,anchor=north west] (#2_TSP) at (#2.south west) {\small\ttfamily #5};
	\node[text width=\boxwidth,cpxity,below=0 of #2,#6] (#2_sTSP) {\small\ttfamily #7};
	\node[text width=\boxwidth,cpxity,#8,anchor=north east] (#2_WRP) at (#2.south east) {\small\ttfamily #9};
	\node[fit=(#2)(#2_TSP)(#2_sTSP)(#2_WRP),draw,inner sep=-0.01] {};
	\draw (#2_TSP.north west) -- (#2_WRP.north east);
	\draw (#2_sTSP.north west) -- (#2_sTSP.south west);
	\draw (#2_sTSP.north east) -- (#2_sTSP.south east);
	\coordinate (#2_top) at (#2.north);
	\coordinate (#2_bot) at (#2_sTSP.south);
}

\usepackage{xspace}

\usepackage{framed}
\usepackage{tabularx}

\newlength{\RoundedBoxWidth}
\newsavebox{\GrayRoundedBox}
\newenvironment{GrayBox}[1]%
{\setlength{\RoundedBoxWidth}{.93\columnwidth}%
	\def\boxheading{#1}%
	\begin{lrbox}{\GrayRoundedBox}%
		\begin{minipage}{\RoundedBoxWidth}}%
		{   \end{minipage}
	\end{lrbox}
	\begin{center}
		\begin{tikzpicture}%
		\node(Text)[draw=black!20,fill=white,rounded corners,inner sep=2ex,text width=\RoundedBoxWidth]
		{\usebox{\GrayRoundedBox}};
		\coordinate(x) at (current bounding box.north west);
		\node [draw=white,rectangle,inner sep=3pt,anchor=north west,fill=white]
		at ($(x)+(6pt,.75em)$) {\boxheading};
		\end{tikzpicture}
\end{center}}

\newenvironment{defproblemx}[1]{\noindent\ignorespaces%
	\FrameSep=6pt%
	\parindent=0pt%
	\begin{GrayBox}{#1}%
		\begin{tabular*}{\columnwidth}{!{\extracolsep{\fill}}@{\hspace{.1em}} >{\itshape} p{1.1cm} p{0.89\columnwidth} @{}}%
		}{
			\vspace*{-1em}
		\end{tabular*}%
	\end{GrayBox}%
	\ignorespacesafterend
}

\newcommand{\defProblemQuestion}[3]{%
	\begin{defproblemx}{#1}
		Input: & #2 \\
		Question: & #3
	\end{defproblemx}
}

\usepackage[textwidth=25mm,obeyFinal,colorinlistoftodos]{todonotes}

\usepackage{pifont}
\newcommand{\N}{\ensuremath{\mathbb{N}}}
\newcommand{\Z}{\ensuremath{\mathbb{Z}}}

\newcommand{\TSP}{\textsc{Traveling Salesperson Problem}\xspace}
\newcommand{\TSPshort}{\textsc{TSP}\xspace}
\newcommand{\sTSP}{\textsc{Subset \TSPshort}\xspace}
\newcommand{\sTSPshort}{{\sc sub\TSPshort}\xspace}
\newcommand{\WRP}{\textsc{Waypoint Routing Problem}\xspace}
\newcommand{\WRPshort}{\textsc{WRP}\xspace}

\newcommand{\fes}{\ensuremath{\operatorname{fes}}\xspace}
\newcommand{\vc}{\ensuremath{\operatorname{vc}}\xspace}

\newcommand{\fn}{\ensuremath{\operatorname{fn}}\xspace}

\newcommand{\tw}{\ensuremath{\operatorname{tw}}\xspace}
\newcommand{\modulator}{\ensuremath{\operatorname{modul}}\xspace}

\renewcommand{\O}{\ensuremath{\mathcal{O}}}
\newcommand{\poly}{\ensuremath{\operatorname{poly}}}

\renewcommand{\P}{\ensuremath{\mathsf{P}}\xspace}
\newcommand{\NP}{\ensuremath{\mathsf{NP}}\xspace}
\newcommand{\NPh}{\NP-hard\xspace}

\newcommand{\NPc}{\NP-complete\xspace}

\newcommand{\coNPpoly}{{\rm co}\NP{\rm /poly}\xspace}
\newcommand{\FPT}{\ensuremath{\mathsf{FPT}}\xspace}
\newcommand{\XP}{\ensuremath{\mathsf{XP}}\xspace}
\newcommand{\W}[1][1]{\ensuremath{\mathsf{W[#1]}}\xspace}
\newcommand{\Wh}[1][1]{\W[#1]-hard\xspace}
\newcommand{\Whness}[1][1]{\W[#1]-hardness\xspace}
\newcommand{\WK}[1][1]{\ensuremath{\mathsf{WK[#1]}}\xspace}
\newcommand{\WKh}[1][1]{\WK[#1]-hard\xspace}
\newcommand{\WKc}[1][1]{\WK[#1]-complete\xspace}

\DeclareMathAlphabet{\mathpzc}{OT1}{pzc}{m}{it}
\newcommand{\budget}{\ensuremath{\mathpzc{b}}} %
\newcommand{\wFn}{\ensuremath{\omega}} %
\newcommand{\WP}{\ensuremath{W}} %
\newcommand{\cFn}{\ensuremath{\kappa}} %

\newcommand{\natBeh}{b^{\operatorname{nat}}}

\newcommand{\allBeh}{\ensuremath{\mathcal{B}}}

\newcommand{\Imp}{\ensuremath{\operatorname{imp}}} %
\newcommand{\allImp}{\ensuremath{\mathcal{I}}} %

\newcommand{\YES}{\emph{yes}\xspace}
\newcommand{\YESi}{\YES-instance\xspace}
\newcommand{\NO}{\emph{no}\xspace}

\newcommand{\vasek}[1]{}

\newtheorem{rrule}{Reduction Rule}
\Crefname{rrule}{Reduction Rule}{Reduction Rules}
\Crefname{rrule}{Reduction Rule}{Reduction Rules}
\Crefname{claim}{Claim}{Claims}

\begin{document}

\maketitle

\begin{abstract}
	For many problems, the important instances from practice possess certain structure that one should reflect in the design of specific algorithms.
	As data reduction is an important and inextricable part of today's computation, we employ one of the most successful models of such precomputation---the kernelization.
	Within this framework, we focus on \TSP (\TSPshort) and some of its generalizations.
	
	We provide a kernel for \TSPshort with size polynomial in either the feedback edge set number or the size of a modulator to constant-sized components.
	For its generalizations, we also consider other structural parameters such as the vertex cover number and the size of a modulator to constant-sized paths.
	We complement our results from the negative side by showing that the existence of a polynomial-sized kernel with respect to the fractioning number, the combined parameter maximum degree and treewidth, and, in the case of \sTSP, modulator to disjoint cycles (i.e., the treewidth two graphs) is unlikely.
\end{abstract}

\section{Introduction}

The \TSP is among the most popular and most intensively studied combinatorial optimization problems in graph theory\footnote{\TSPshort Homepage: \url{http://www.math.uwaterloo.ca/tsp/index.html}.}.
From the practical point of view, the problem was already studied in the 1950s; see, e.g., \cite{Croes58}.
We define the problem formally as follows.

\defProblemQuestion{\TSP (\TSPshort)}
{An undirected graph $G=(V,E)$, edge weights $\wFn \colon E \to \N$, and a budget $\budget \in \N$.}
{Is there a closed walk in $G$ of total weight at most $\budget$ that traverses each vertex of $G$ at least once?}

\TSPshort is known to be \NPh~\cite{GareyJ79}. %
It is worth mentioning that the input to \TSPshort is usually a complete graph and it is required that each vertex is visited exactly once (single-visit version).
In this paper, we aim to study the impact of the structure of the input graph on the complexity of finding a solution to \TSPshort.
Towards this, we employ parameterized analysis~\cite{CyganFKLMPPS15} and therefore we also consider the decision version of \TSPshort.
Many variants of the problem have been studied in detail~\cite{GutinP07}. %
Our formulation is also known as \textsc{Graphical TSP}~\cite{CornuejolsFN85}.

In this paper, we further study two natural generalizations of \TSPshort---\sTSP and \WRP.
In \sTSP (\sTSPshort) the objective is to find a closed walk required to traverse only a given subset $\WP$ of vertices (referred to as \textit{waypoints}) instead of the whole vertex set as is in \TSPshort.
An instance of \TSPshort can be interpreted as an instance of \sTSPshort by letting $\WP=V(G)$.
In \WRP (\WRPshort) the edges of the input graph have capacities (given by a function $\cFn \colon E(G) \to \N$) and the objective is to find a closed walk traversing a given subset of vertices respecting the edge capacities, i.e., there is an upper limit on the number of times each edge can be traversed in a solution.
As we show later, an instance of \sTSPshort can be interpreted as an instance of \WRPshort by letting $\cFn(e)=2$ for every $e \in E(G)$.
Note that both \sTSPshort and \WRPshort are \NPc.

\subparagraph*{Related Work.}
\TSPshort (and its variants) were extensively studied from the viewpoint of approximation algorithms.
The single-visit version of \TSPshort in general cannot be approximated, unless $\P=\NP$~\cite{SahniG76}.
For the metric single-visit version of the problem Christofides~\cite{Christofides76} provided a $\frac{3}{2}$-approximation, while it is known that unless $\P=\NP$, there is no $\frac{117}{116}$-approximation algorithm~\cite{ChlebikC19}.
This is so far the best known approximation algorithm for the general case,
despite a considerable effort, e.g., \cite{ShmoysW90, Goemans95, CarrVM00, GamarnikLS05, BoydC11, HaddadanN19, KarlinKG20}.

The problem remains \textsf{APX}-hard even in the case of weights 1 and 2, however, a $\frac{7}{6}$-approximation algorithm is known for this special case~\cite{PapadimitriouY93}.
A \textsf{PTAS} is known for the special case of Euclidean~\cite{Arora98, Mitchell} and planar~\cite{GrigniKP95, Arora96, Klein05} \TSPshort.
Also, the case of graph metrics received significant attention.
Gharan et al.~\cite{GharanSS11} found a $\frac{3}{2}-\varepsilon_0$ approximation for this variant. M\"omke and Svensson~\cite{MomkeS11} then obtained a combinatorial algorithm for graphic TSP with an approximation factor of $1.461$.
This was later improved by Mucha~\cite{Mucha12} to $\frac{13}{9}$ and then by Sebö and Vygen~\cite{SeboV14} to $1.4$.
See, e.g., the monograph of Applegate et al.~\cite{ApplegateBixbyChvatalCook2011} for further information.

A popular practical approach to the single-visit version of \TSPshort is to gradually improve some given solution using local improvements---the so-called $q$-OPT moves (see, e.g., \cite{Johnson2007,JohnsonMcGeoch2018}).
Marx~\cite{Marx08} proved that \TSPshort is \Wh w.r.t.~$q$.
Later, Bonnet et al.~\cite{BonnetIJK19} studied the $q$-OPT technique on bounded degree graphs.
They also investigated the complexity in the case where $q$ is a fixed constant.
The first significant theoretical improvement in the general case was by de Berg et al.~\cite{deBergBJW21} (announced at ICALP~'16) who gave an $\O(n^{\lfloor 2q/3 \rfloor + 1})$ time dynamic programming based algorithm.
Furthermore, they showed that improving upon the $\O(n^3)$ time for the case of $3$-OPT would yield a breakthrough result by a connection with the \textsc{All Pairs Shortest Paths} (APSP) problem (see, e.g., \cite{ChanW21}).
Lancia and Dalpasso~\cite{LanciaD2019} observed that the DP-based approach is indeed slow in practice and implemented a practical fast algorithm tailored for $4$-OPT.
Later, Cygan et al.~\cite{CyganKS19} improved the result of de Berg et al.~\cite{deBergBJW21} by presenting an $\O(n^{(1/4 + \varepsilon_q)q})$ time algorithm.
Moreover, they showed a connection of $4$-OPT to APSP.
Guo et al.~\cite{GuoHNS13} considered other suitable measures (parameters) for local search (e.g., $r$-swap or $r$-edit) and provided a \Whness result in general graphs (for both parameters).
Furthermore, they gave an \FPT algorithm for planar graphs and showed that the existence of polynomial kernels is unlikely even in this case.

\sTSP can be solved in time $2^{\ell} \cdot n^{\O(1)}$ by first computing the distance between every pair of waypoints, and then solving the resulting $\ell$-waypoint instance using the standard Bellman-Held-Karp dynamic programming algorithm~\cite{Bellman62,HeldK71}.
Klein and Marx~\cite{KleinM14} showed that if $G$ is an undirected planar graph with polynomially bounded edge weights, then the problem can be solved significantly faster, in $2^{\O(\sqrt{\ell}\log \ell)} \cdot n^{\O(1)}$ time.
It is also known that \sTSP can be solved in $2^{\O(\sqrt{\ell}\log \ell)}\cdot n^{\O(1)}$ on edge-weighted directed planar graphs~\cite{MarxPP18}.

Amiri et al.~\cite{AmiriFS20} initiated the study of the variant of \WRPshort discussed in this paper.
They showed that the problem admits an \XP algorithm w.r.t.\ the treewidth of the input graph.
They also gave a randomized algorithm for the problem with running time $2^\ell\cdot n^{\O(1)}$, where $\ell$ is the number of waypoints; the algorithm can be derandomized if $\ell=\O((\log \log n)^{1/10})$.
Schierreich and Suchý~\cite{SchierreichS22} improved the former algorithm by giving a deterministic algorithm with running time $2^{\O(\operatorname{tw})}\cdot n$, which is tight under ETH~\cite{ImpagliazzoP2001}.
Note that \textsc{Hamiltonian Path} and therefore \TSPshort is \Wh w.r.t.\ the cliquewidth of the input graph~\cite{FominGLS10}.

It is somewhat surprising that, despite its popularity in theoretical computer science, the \TSPshort problem did not receive much attention from the perspective of data reduction (kernelization).
A notable exception is the work of Mannens et al.~\cite{MannensNSS21,MannensNSS21arxiv} who studied \textsc{Many Visits TSP} and gave a polynomial kernel when the problem is parameterized by the vertex cover number of the input graph.
In this work, we significantly expand the kernelization approach to \TSPshort and its variants.

\subparagraph{Structural parameters within the class of bounded treewidth.}
Even though all problems in \FPT admit a kernel, not all of them admit a polynomial-sized kernel.
It is not hard to see that \TSPshort is AND-compositional and therefore we can (conditionally) exclude the existence of a polynomial kernel w.r.t.\ treewidth; more precisely, we show the following.

\begin{restatable}{lemma}{lemTSPnoPKwrtFN}
	\label{lem:TSP:noPKwrtFN}
	There is no polynomial kernel for unweighted \TSP with respect to either the fractioning number or the combined parameter treewidth and maximum degree, unless the polynomial hierarchy collapses.
\end{restatable}

It follows that in order to obtain positive results we should investigate parameters within the class of bounded treewidth.
Note that, roughly speaking, fractioning number is a \emph{modulator to parameter-sized components}.
A~modulator is a set of vertices $M$ such that when we remove $M$ from a graph $G$ the components of $G \setminus M$ fall into a specific (preferably simple) graph class (see \Cref{sec:prelim} for formal definitions).
Since almost all graph classes are closed under disjoint union, it is also popular to speak of addition of $k$ vertices to a graph class (i.e., for a graph class $\mathcal{G}$ we write $\mathcal{G} + kv$ if we are ``adding'' $k$ vertices).
Thus, it naturally models the situation when the input is close to a class of graphs where we can solve the problem efficiently (also called ``Distance to $\mathcal{G}$'').
Modulator size is a popular structural parameter (see, e.g.,~\cite{Cai03,LokshtanovRSSZ18,MajumdarRS18}).
Thus, the above result motivates us to study modulator-based parameters and their effect on kernelization of \TSPshort and its variants.
Indeed, many problems admit much more efficient algorithms w.r.t. vertex cover than for treewidth parameterization; see, e.g., \cite{JansenK13,FellowsFLRSST11,FialaGK11}.
Therefore, this is exactly the point where the positive part of our journey begins.

Before we do so, we discuss one more negative result.
For \sTSPshort we can, under a standard complexity-theoretic assumption, exclude the existence of a polynomial (Turing) kernel for the parameter size of a modulator to graphs of treewidth two. %
More specifically, we show the following.

\begin{restatable}{theorem}{thmsTSPWKcomplete}
	\label{thm:sTSP:WKcomplete}%
	Unweighted \sTSPshort with respect to the minimum $k'$ such that there is a size $k'$ modulator to cycles of size at most $k'$ and there are at most $k'$ non-terminals is \WKh.
\end{restatable}

Therefore, \sTSPshort parameterized by the size of a modulator to (disjoint) cycles does not admit a polynomial (Turing) kernel, unless all problems in \WK admit a polynomial Turing kernel.
Moreover, it does not admit (classical) polynomial kernel, unless the polynomial hierarchy collapses.
For more details on \WKc{}ness, see \Cref{ssec:prelim:turingKernelization} and~\cite{HermelinKSWW15}.

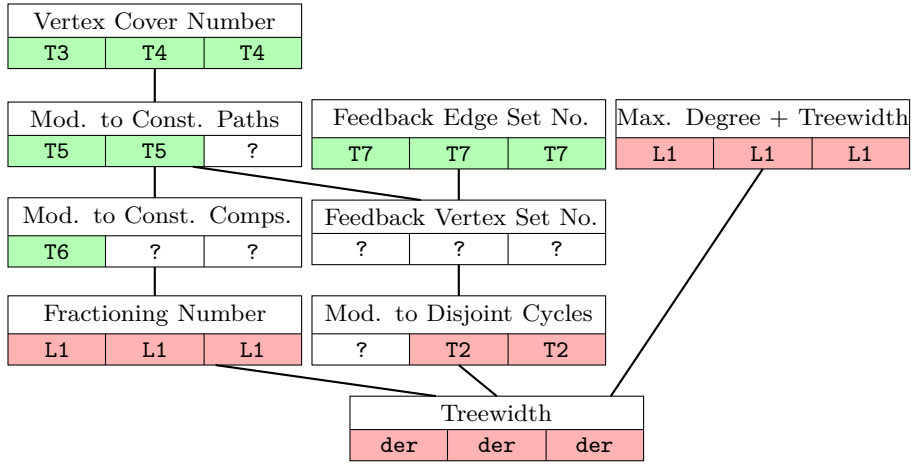
\begin{figure}[bt]
	\centering
	\begin{tikzpicture}

		\tikzstyle{polyk} = [fill=green!30]
		\tikzstyle{npk} = [fill=red!30]
		\tikzstyle{unknown} = []
		\def\ystep{-1.3}

		\resultbox<Vertex Cover Number>(VC)(-3,1*\ystep)%
		(polyk;T\ref{thm:TSPVertexCoverKernel})%
		(polyk;T\ref{thm:WRP:VCkernel})%
		(polyk;T\ref{thm:WRP:VCkernel});

		\resultbox<Mod. to Const. Paths>(MCP)(-3,2*\ystep)%
		(polyk;T\ref{thm:WRP:modulatorToConstPathsKernel})%
		(polyk;T\ref{thm:WRP:modulatorToConstPathsKernel})%
		(unknown;?);

		\resultbox<Feedback Edge Set No.>(FES)(1,2*\ystep)%
		(polyk;T\ref{thm:WRP:FESkernel})%
		(polyk;T\ref{thm:WRP:FESkernel})%
		(polyk;T\ref{thm:WRP:FESkernel});

		\resultbox<Mod. to Const. Comps.>(MCC)(-3,3*\ystep)%
		(polyk;T\ref{thm:WRP:modulatorToConstComponentsKernel})%
		(unknown;?)%
		(unknown;?);

		\resultbox<Feedback Vertex Set No.>(FVS)(1,3*\ystep)%
		(unknown;?)%
		(unknown;?)%
		(unknown;?);

		\resultbox<Max. Degree + Treewidth>(DTW)(5,2*\ystep)%
		(npk;L\ref{lem:TSP:noPKwrtFN})(npk;L\ref{lem:TSP:noPKwrtFN})(npk;L\ref{lem:TSP:noPKwrtFN});

		\resultbox<Fractioning Number>(FN)(-3,4*\ystep)%
		(npk;L\ref{lem:TSP:noPKwrtFN})(npk;L\ref{lem:TSP:noPKwrtFN})(npk;L\ref{lem:TSP:noPKwrtFN});

		\resultbox<Mod. to Disjoint Cycles>(MDC)(1,4*\ystep)%
 		(unknown;?)(npk;T\ref{thm:sTSP:WKcomplete})(npk;T\ref{thm:sTSP:WKcomplete});
		(unknown;?)(npk;)(npk;);

		\resultbox<Treewidth>(TW)(1.5,5*\ystep)(npk;der)(npk;der)(npk;der);

		\draw[thick] ($( VC_bot)+(0.0,0)$) -- ($(MCP_top)+(0.0,0)$);
		\draw[thick] ($(MCP_bot)-(0.0,0)$) -- ($(MCC_top)+(0.0,0)$);
		\draw[thick] ($(MCP_bot)+(0.5,0)$) -- ($(FVS_top)-(0.5,0)$);
		\draw[thick] ($(FES_bot)-(0.0,0)$) -- ($(FVS_top)+(0.0,0)$);
		\draw[thick] ($(MCC_bot)+(0.0,0)$) -- ($( FN_top)+(0.0,0)$);
		\draw[thick] ($(FVS_bot)+(0.0,0)$) -- ($(MDC_top)+(0.0,0)$);
		\draw[thick] ($( FN_bot)+(0.8,0)$) -- ($( TW_top)-(0.8,0)$);
		\draw[thick] ($(MDC_bot)+(0.0,0)$) -- ($( TW_top)+(0.0,0)$);
		\draw[thick] ($(DTW_bot)+(0.0,0)$) -- ($( TW_top)+(1.5,0)$);
	\end{tikzpicture}
	\caption{
		Overview of our results.
		Each node names a graph parameter and at its bottom it contains boxes representing (from left to right) \TSPshort, \sTSPshort, and \WRPshort.
		If a node $A$ is connected to a lower node $B$ with an edge, then $B(G) = \O(A(G))$ for respective parameters.
		\textit{Modulator}, \textit{Constant}, and \textit{Components} are shortened.
		Each box contains either a reference to the theorem which proves existence or non-existence of polynomial kernel (e.g. {\tt T2}), or it contains {\tt ?} if the problem is open, or it contains {\tt der} if the result for such setting is derived from another result.
	}
	\label{fig:results_overview}
\end{figure}

\subparagraph{Our Contribution.}
On a positive note, we begin with a polynomial kernel for \TSPshort and \WRPshort w.r.t.\ vertex cover number to demonstrate the general techniques we employ in this paper.
We study the properties of a nice solution to the instance with respect to a particular structural parameter and its interaction with the modulator.
We begin the study with \TSPshort as a warm-up.
In this case, we show that if there are ``many'' vertices in $G \setminus M$, where $M$ is the modulator (vertex cover), then all but few of them behave in a ``canonical way''.
Based on the properties and cost of this canonical traversal, we identify which of these vertices can be safely discarded from the instance. In particular, we show the following.

\begin{restatable}{theorem}{thmTSPVCkernel}\label{thm:TSPVertexCoverKernel}%
	\TSP admits a kernel with $\O(k^3)$ vertices, $\O(k^4)$ edges, and with a total bit-size $\O(k^{16})$, where $k$ is the vertex cover number of $G$.
\end{restatable}

In the case of \sTSPshort or \WRPshort, we cannot be sure that a solution visits all vertices in the modulator.
Quite naturally, this results in the second kind of rule where we see that if there are ``many'' vertices attached to the modulator in the same way, we can mark some vertices in the modulator as terminals (there is a solution visiting them).

\begin{restatable}{theorem}{thmWRPVCkernel}
	\label{thm:WRP:VCkernel}%
  \WRP admits a kernel with $\O(k^8)$ vertices, $\O(k^9)$ edges, and with a total bit-size $\O(k^{36})$, where $k$ is the vertex cover number of $G$.
\end{restatable}

Using similar, yet more involved results, we are able to find a polynomial kernel for \sTSPshort for a modulator to constant-sized paths and for \TSPshort for a modulator to constant-sized components.

One ingredient for our approach is that we may assume that all vertices outside the modulator are terminals, and therefore every solution visits them.
It is not clear how to use our approach in the presence of capacities already in the case of modulator to paths, since the previous observation does not hold already in this simple case.
Indeed, we show this by contracting some of the edges, which in the presence of capacities might have a side effect on the set of possible traversals of the component.
The other and more important ingredient is the so-called blending of solutions (see \Cref{lem:path_mid_behavior}).

\begin{restatable}{theorem}{thmWRPmodulatorToConstPathsKernel}\label{thm:WRP:modulatorToConstPathsKernel}%
	Let $r$ be a fixed constant.
	\sTSP admits a kernel of size $k^{\O(r)}$, where $k$ is the size of a modulator to paths of size at most~$r$.
\end{restatable}

When we allow general constant-sized components outside the modulator, we have not yet succeeded in dealing with the highly connected non-terminals in the modulator.
At this point, it is not clear whether many components with the same combination of canonical transversal and attachment to the modulator set ensure that we can safely mark a vertex in the modulator as a terminal.

\begin{restatable}{theorem}{thmWRPmodulatorToConstComponentsKernel}\label{thm:WRP:modulatorToConstComponentsKernel}%
	Let $r$ be a fixed constant.
	\TSP admits a kernel of size $k^{\O(r)}$, where $k$ is the size of a modulator to components of size at most~$r$.
\end{restatable}

We conclude our algorithmic study by providing a rather simple polynomial kernel for \WRPshort parametrized by the feedback edge set number.
On the one hand, this result is an application of rather straightforward local reduction rules.
On the other hand, these rules are heavily based on the use of edge capacities.
It should be pointed out that we get a polynomial kernel for \TSPshort as it has a polynomial compression to \WRPshort, however, it can be shown that ``local reduction rules'' do not exist in the case of \TSPshort.

\begin{restatable}{theorem}{thmWRPFESkernel}\label{thm:WRP:FESkernel}%
	\WRP admits a kernel with $\O(k)$ vertices and edges and bit size $\O(k^4)$, where $k$ is the feedback edge set number of $G$.
\end{restatable}

For an overview of our results, please, refer to \Cref{fig:results_overview}.

\subparagraph{Organization of the paper.}
We summarize the notation and technical results we rely on in \Cref{sec:prelim}.
\Cref{sec:generalObservations} contains a few useful technical lemmas and simple reduction rules.
In \Cref{sec:vc}, we begin with the core concepts applied in the case of \TSPshort and vertex cover number. %
The similar approach is then applied in \Cref{sec:modulatorsPolyKernel} to more general types of modulators; this is the most technical part of this manuscript.
Finally, in \Cref{sec:conclusions} we conclude the results and discuss future research directions.

\section{Preliminaries}\label{sec:prelim}

We follow the basic notation of graph theory by Diestel~\cite{Diestel17}. In parameterized complexity theory, we follow the monograph of Cygan et al.~\cite{CyganFKLMPPS15}.

A \emph{walk} in a graph $G$ is a non-empty alternating sequence of vertices and edges $S=v_1,e_1,\dots,e_{\ell-1},v_{\ell}$ such that $v_i\in V(G)$, $e_i\in E(G)$, and $e_i=\{v_i,v_{i+1}\}$, $\forall i\in\{1, \ldots, \ell-1\}$.
It is \emph{closed} if $v_\ell=v_1$.
The weight of walk~$S$ is $\wFn(S)=\sum_{i=1}^{\ell-1} \wFn(e_i)$.
If all vertices in a walk are distinct, it is called a \emph{path}.

A \emph{solution} to our problems is a closed walk $S$ visiting every vertex in $\WP$ of total weight at most $\budget$ (traversing each edge $e$ at most $\cFn(e)$ times).
The least cost such walk is called \emph{an optimal solution}.
Given such a walk, we can construct the \emph{corresponding multigraph} $G_S$ which is a multigraph with vertex set being the set of vertices visited by $S$ and each edge occuring as many times as traversed by $S$.
We naturally extend $\wFn$ to this multigraph, which yeilds $\wFn(G_S)= \sum_{e \in E(G_S)} \wFn(e)=\wFn(S)$.
The degree of a vertex in a multigraph is the number of the edges incident with it.
Conversely, if a multigraph $G_S$ is Eulerian (connected with all degrees even), then it admits a walk visiting every vertex of the graph and traversing each edge exactly as many times as occuring in $G_S$.

\subparagraph{Structural Graph Parameters.}

Let $G=(V,E)$ be a graph.%
\begin{definition}
    A set of edges $F \subseteq E$ is a \emph{feedback edge set} of the graph $G$ if $G \setminus F$ is an acyclic graph.
    \emph{Feedback edge set number} $\fes(G)$ of a graph $G$ is the size of a smallest feedback edge set $F$ of $G$.
\end{definition}
\begin{definition}
    A set of vertices $C\subseteq V$ is called \emph{vertex cover} of the graph $G$ if it holds that $\forall e \in E$ we have $e \cap C \neq \emptyset$.
    The \emph{vertex cover number} $\vc(G)$ of a graph $G$ is the least size of a vertex cover of $G$.%
\end{definition}

\begin{definition}
    Let $\mathcal{G}$ be a graph family, let $G$ be a graph, and $M\subseteq V(G)$.
    We say that $M$ is a \emph{modulator} of the graph $G$ to the class $\mathcal{G}$ if each connected component of $G\setminus M$ is in $\mathcal{G}$.
	The \emph{distance of~$G$ to $\mathcal{G}$}, denoted $\modulator(G,\mathcal{G})$, is the minimum size of a modulator of $G$ to $\mathcal{G}$.
\end{definition}

Let $r$ be a fixed constant.
Let $\mathcal{G}_{\operatorname{rc}}$ be the class of graphs with every connected component having at most $r$ vertices.
The \emph{distance of $G$ to $r$-components} is $\modulator(G,\mathcal{G}_{\operatorname{rc}})$. Similarly, if $\mathcal{G}_{\operatorname{rp}}$ is a class of graphs where every connected component is a path with at most $r$ vertices, then the \emph{distance of $G$ to $r$-paths} is $\modulator(G,\mathcal{G}_{\operatorname{rp}})$.

\begin{definition}[Fractioning Number]
	A set of vertices $C_r\subseteq V$ is called \emph{$r$-fractioning set} of the graph $G$ if $|C_r| \le r$ and every connected component of $G \setminus C_r$ has at most $r$ vertices. The \emph{fractioning number} $\fn(G)$ is a minimum $r\in\N$ such that there is an $r$-fractioning set $C_r$ in graph $G$.
\end{definition}

\subsection{(Turing) Kernelization}\label{ssec:prelim:turingKernelization}
\subparagraph{Kernelization.} A language~$L\subseteq \Sigma^* \times \N$, where~$\Sigma$ is a finite alphabet, is called a \emph{parameterized problem}. For an instance~$(x,k)\in \Sigma^* \times \N$ of $L$ we call~$k$ its \emph{parameter}. We say that function~$\psi\colon \Sigma^* \times \N \to \Sigma^* \times \N$ is a \emph{data reduction rule} if it maps an instance~$(x,k)$ to an instance~$(x',k')$ of the same parameterized problem~$L$ such that~$\psi$ is computable in time polynomial in~$|x|$ and~$k$, and instances~$(x,k)$ and~$(x',k')$ are \emph{equivalent}, i.e., $(x,k)\in L$ if and only if~$(x',k')\in L$. A data reduction rule, also called \emph{reduction rule}, is termed \emph{safe} (or \emph{correct}) if indeed the resulting instance after its application is equivalent to the original one.

A \emph{kernelization algorithm}~$\mathcal{A}$, or \emph{kernel} for short, for parameterized problem $L$ is an algorithm that works in polynomial time and, given an instance~$(x,k)\in L$, returns an equivalent instance~$(x',k')\in L$ such that~$|x'|+k \leq f(k)$ for some computable function~$f\colon\N\to\N$. If~$f$ is a polynomial function of the parameter, then we say that~$L$ admits a \emph{polynomial kernel}.

\subparagraph{Refuting Polynomial Kernels.}
When refuting existence of a polynomial kernel for a parameterized problem~$L\subseteq \Sigma^* \times \N$ we follow the framework of Bodlaender et al.~\cite{BodlaenderDFH09}; see also~\cite[Chapter~15]{CyganFKLMPPS15}.

\begin{definition}[Polynomial equivalence relation]\label{def:poly_equiv_relation}
  An equivalence relation $\mathcal{R}$ on the set $\Sigma^*$ is a \emph{polynomial equivalence relation} if:
  \begin{enumerate}
    \item
      There is an algorithm that, given $x,y \in \Sigma^*$, decides $\mathcal{R}(x,y)$ in $\poly(|x|+|y|)$ time.
    \item
      $\mathcal{R}$ restricted to strings of length at most $n$ has at most $\poly(n)$ equivalence classes.
  \end{enumerate}
\end{definition}

\begin{definition}[OR-Cross-composition]
  Let $L\subseteq \Sigma^* \times \N$ be a parameterized problem and let $Q \subseteq \Sigma^*$ be a language.
  We say that $Q$ \emph{OR-cross-composes} into~$L$ if there is an algorithm $\mathcal{A}$ that takes on input strings $x_1, \ldots, x_t \in \Sigma^*$ that are equivalent with respect to some polynomial equivalence relation, runs in time $\poly(\sum_{i=1}^t |x_i|)$ and outputs one instance $(y,k) \in \Sigma^* \times \N$ such that %
    $k \le \poly(\max_{i=1}^t, \log t)$, and
    $(y,k) \in L$ if there exists $i$ such that $x_i \in Q$.
\end{definition}

If we replace the last condition with $(y,k) \in L$ if $x_i \in Q$ for all $i$, then we say that $Q$ \emph{AND-cross-composes} into~$L$.

\begin{proposition}[{\cite{Drucker15}}]\label{pro:nph_or_cross_composition}
    Assume that an \NPh language $Q$ OR-cross-composes to a parameterized language~$L$ (or $Q$ AND-cross-composes to $L$).
    Then, $L$ does not admit a polynomial kernel, unless \NP{} $\subseteq$ \coNPpoly{} and polynomial-hierarchy collapses.
\end{proposition}

\begin{definition}[{\cite[Definition~15.14]{CyganFKLMPPS15}}]
	Let $P$ and $Q$ be two parametrized problems.
	An algorithm~$\cal A$ is called a \emph{polynomial parameter transformation} from $P$ to $Q$ if, given an instance $(x,k)$ of $P$, $\cal A$ works in polynomial time and outputs an equivalent instance $(x',k')$ of $Q$, such that $k' \le \poly(k)$.
\end{definition}

\subparagraph{Turing Kernels.}

\begin{definition}[{\cite[Definition~9.28]{CyganFKLMPPS15}}]
	Let $Q$ be a parametrized problem and let $f \colon \N \to \N$ be a computable function.
	A \emph{Turing kernelization} for~$Q$ of size~$f$ is an algorithm that decides whether a given instance $(x,k)$ is contained in~$Q$ in time polynomial in ${|x|+k}$, when given access to an oracle that decides membership in~$Q$ for any instance~$(x',k')$ with $|x'|,k' \le f(k)$ in a single step.
\end{definition}

If the function $f$ is a polynomial, we say that the problem admits a \emph{polynomial Turing kernel}.

Hermelin et al.~\cite{HermelinKSWW15} introduced a class \WK and conjectured that \WK-hard problems (under polynomial parameter transformations) do not admit polynomial Turing kernels.
Moreover, they showed that no \WK-hard problem has a (classical) polynomial kernel, unless \NP{} $\subseteq$ \coNPpoly{} and polynomial-hierarchy collapses.

\subsection{Reducing Numbers in the Input}

We say that two functions $f,g \colon \Z^d \to \Z$ are \emph{equivalent} on a polyhedron $P \subseteq \Z^d$ if $f(\bm{x}) \leq f(\bm{y})$ if and only if $g(\bm{x}) \leq g(\bm{y})$ for all $\bm{x}, \bm{y} \in P$.
\begin{proposition}[{Frank and Tardos~\cite{FrankT87}}]\label{thm:FT}
	Given a rational vector $\bm{w} \in \mathbb{Q}^{d}$ and an integer $M$, there is a polynomial algorithm which finds a $\widetilde{\bm{w}} \in \Z^d$ such that the linear functions $\bm{w} \bm{x}$ and $\widetilde{\bm{w}}\bm{x}$ are equivalent on $[-M,M]^d$, and $\|\widetilde{\bm{w}}\|_\infty \leq 2^{\O(d^3)} M^{\O(d^2)}$.
\end{proposition}
The standard approach, also used in kernelization of weighted problems~\cite{BentertvBFNN19,BevernFT19,BevernFT20,ChaplickFGK019,EtscheidKMR17,GoebbelsGRY17,GurskiRR19} is to use the above proposition which ``kernelizes'' a linear objective function if the dimension is bounded by a parameter.

\subsection{Problem Definitions}

\defProblemQuestion{\sTSP (\sTSPshort)}
	{An undirected graph $G=(V,E)$, set of waypoints $\WP \subseteq V$, edge weights $\wFn \colon E \to \N$, budget $\budget \in \N$.}
	{Is there a closed walk in $G$ of total weight at most $\budget$, that traverses each vertex in $\WP$ (at least once)?}
\defProblemQuestion{\WRP (\WRPshort)}
	{An undirected graph $G=(V,E)$, set of waypoints $\WP \subseteq V$, edge weights $\wFn \colon E \to \N$, edge capacities $\cFn \colon E \to \N$, budget $\budget \in \N$.}
	{Is there a closed walk $C$ in $G$ of total weight at most $\budget$, that traverses each vertex in $\WP$ (at least once) and such that for each edge $e$ the number of times $C$ traverses $e$ is at most $\cFn(e)$?}

Although \WRP allows for arbitrary edge capacities, it can be seen that in any solution, traversing an edge more than twice can lead to redundancy in the solution (see \Cref{lem:savy_optimal_solution}). Hence, in this paper, we assume the edge capacities to be either $1$ or $2$ when analyzing \WRPshort.

\section{The Toolbox}\label{sec:generalObservations}

In this section, we give formal proofs to a few technical statements we use throughout the rest of the paper.
This yields a useful set of assumptions that allow us to present less technical proofs in the subsequent sections.
We begin with a technical lemma we were not able to find in literature.

\begin{lemma}\label{lem:many_edges_cycle}
    Let $G$ be a graph with more than $2|V(G)|-2$ edges.
    Then there is a cycle $C$ in $G$ such that the graph $G'=G \setminus E(C)$ has the same connected components as $G$.
\end{lemma}
\begin{proof}
    Let $T$ be a maximum acyclic subgraph of $G$.
    Note that $T$ has the same connected components as $G$.
    If $R=G \setminus E(T)$ is acyclic, then we have $|E(G)| = |E(T)|+|E(R)| \le |V(G)|-1 +|V(G)|-1$ contradicting the assumptions of the lemma.
    Hence, $R$ contains a cycle $C$.
    Since $C$ is disjoint with $T$, it satisfies the conditions of the lemma.
\end{proof}

When proving the safeness of our reduction rules, it is easier to work with a solution that does not use many edges and traverses each edge at most twice.
We show that we can always assume to work with such a solution.

\begin{definition}[Nice Solution]
 Let $(G,\WP,\wFn,\cFn,\budget)$ be an instance of \TSPshort, \sTSP, or \WRPshort.
 We call a solution \emph{nice} if it uses every edge at most twice and contains at most $2|V|$ edges (edge traversals) in total.
\end{definition}

Indeed, we may always assume that we work with a nice solution.

\begin{lemma}\label{lem:savy_optimal_solution}
    Let $(G,\WP,\wFn,\cFn,\budget)$ be an instance of \TSPshort, \sTSP, or \WRPshort.
    There is a nice optimal solution.
\end{lemma}
\begin{proof}
    Let $S$ be an optimal solution with a number of edges.
    If $S$ contains at most $2|V|$ edges and each edge at most twice, then we are done.
    Otherwise, denote by $G_S$ the graph $G[E(S)]$, i.e., the multigraph induced by edges of the walk $S$.
    Since $S$ is a solution for the \WRPshort instance $(G,\WP,\wFn,\cFn,\budget)$, $G_S$ is an Eulerian graph.
    If $G_S$ contains more than  $2|V|$ edges, by \Cref{lem:many_edges_cycle} there is a cycle $C$ in $G_S$, such that $G'_S =G_S \setminus C$ is also connected.
    Since the degree of each vertex in $C$ is $2$, $G'_S$ is also Eulerian, yielding an optimal solution with less edges, contradicting the choice of $S$.
    Similarly, if $G_S$ contains three edges between the same pair of vertices, then two of them form such a cycle.
    This completes the proof.
\end{proof}

We continue with a remark about our instances.
It is important to note that item (b) is applicable in general, however, one should be careful when doing so, since it increases the weights in the instance.

\begin{remark}\label{rem:connectedInstancesAndPositiveWeights}
	Let $(G,\WP,\wFn,\cFn,\budget)$ be an instance of \TSPshort, \sTSP, or \WRPshort.
  We assume that
			a) $G$ is a connected graph and
			b) $\wFn(e) > 0$ for all $e \in E(G)$.
\end{remark}
\begin{proof}
~
\begin{enumerate}
	\item %
		If $G$ is disconnected and there are two waypoints $w_i,w_j\in \WP$ such that $w_i$ and $w_j$ belong to different components, then the answer is trivially \NO.
		Otherwise, we can continue with the connected component containing all waypoints, which is a connected graph.
		Note that in the case of \TSPshort all vertices are considered waypoints.
	\item %
	  Let $G = (V,E)$.
		We first multiply the weight of every edge by a factor $Q = \left( \sum_{e \in E} \wFn(e) \right) + 2|V| + 1$.
		Then, we increase the weight of all edges $e$ with weight~0 to~1 and adjust the budget.
		Formally, we create a new instance $\mathcal{\widehat{I}} = (G,\WP,\widehat{\wFn},\cFn,\widehat{\budget})$, where
		\[
			\widehat{\wFn}(e) =
			\begin{cases}
				Q \cdot \wFn(e) & \text{if } \wFn(e) > 0 \\
				1								& \text{if } \wFn(e) = 0
			\end{cases}
		\]
		and $\widehat{\budget} = Q \cdot \budget + 2|V|$.

		Assume $\mathcal{I} = (G,\WP,\wFn,\cFn,\budget)$ has a solution.
		Since we can always assume we have a nice solution, there are no more than $2|V|$ edges of weight~0.
		It is not hard to verify that we can use the same solution in~$\mathcal{I}$ and that this solution indeed obeys the constraint of $\widehat{\budget}$.

		In the opposite direction assume~$\mathcal{\widehat{I}}$ admits a solution~$S$.
		Let $S_Q$ denote the set of edges in $S$ with weight strictly more than 1.
		Now, we get that
		\[
          Q \cdot \budget + 2|V|
		  =
		  \widehat{\budget}
			\ge
			\sum_{e \in S} \widehat{\wFn}(e)
			\ge
			\sum_{e \in S_Q} \widehat{\wFn}(e)
			=
			Q \cdot \sum_{e \in S_Q} \wFn(e)
			\,,
		\]
		which implies that
		\[
		 \sum_{e \in S_Q} \wFn(e) \le \frac{ Q \cdot \budget + 2|V|}{Q} < \budget +1
		\]
		as $Q > 2|V|$.
		We conclude that $\wFn(S)=\sum_{e \in S_Q} \wFn(e) \le \budget$.
		Thus, $S$ is a solution to $\mathcal{I}$.
		\qedhere
\end{enumerate}
\end{proof}

We apply the following lemma to reduce the weights of the kernelized instances.
It follows in a rather straightforward way from \Cref{thm:FT}.

\begin{lemma}\label{lem:magic_kernel}
 There is a polynomial time algorithm, which, given an instance $(G,\WP,\wFn,\cFn,\budget)$ of \TSPshort, \sTSPshort, or \WRPshort with at most $d$ edges, produces $\wFn'$ and $\budget'$ such that the instances  $(G,\WP,\wFn,\cFn,\budget)$ and $(G,\WP,\wFn',\cFn,\budget')$ are equivalent and $(G,\WP,\wFn',\cFn,\budget')$ is of total bit-size $\O(d^4)$.
\end{lemma}
\begin{proof}
We consider a rational vector $(\bm{w},\budget)\in \mathcal{Q}^{d'+1}$ where $d' \le d$ is the total number of edges.
Since each edge can be present in the solution zero, one or two times, let $M=2$ so that $\{0,1,2\} \subseteq [-M,M]$.
Then due to \Cref{thm:FT} there exists a polynomial algorithm that returns a $(\widetilde{\bm{w}},\widetilde{\budget}) \in \Z^{d'+1}$ such that the linear functions $(\bm{w},\budget) \bm{x}$ and $(\widetilde{\bm{w}},\widetilde{\budget})\bm{x}$ are equivalent on $[-2,2]^{d'+1}$, and $\|(\widetilde{\bm{w}},\widetilde{\budget})\|_\infty \leq 2^{\O(d'^3)}M^{\O(d'^2)}   =2^{\O(d^3)}$.
Observe that since $(\bm{w},\budget) \bm{x}$ and $(\widetilde{\bm{w}},\widetilde{\budget})\bm{x}$ are equivalent, we have that $\bm{x}$ is a solution w.r.t.\ weights $\bm{w}$ and a budget $\budget$ if and only if $\bm{x}$ is a solution w.r.t.\ weights $\widetilde{\bm{w}}$ and a budget $\widetilde{\budget}$.
To see this, let $\bm{x}_E$ denote the first $d'$ coordinates of~$\bm{x}$.
Note that we have $\bm{w}\cdot\bm{x}_E \le \budget$ if and only if $\bm{w}\cdot\bm{x}_E - \budget \le 0$ which is if and only if $\widetilde{\bm{w}}\cdot\bm{x}_E - \widetilde{\budget} \le 0$.
As $\|(\widetilde{\bm{w}},\widetilde{\budget})\|_\infty =2^{\O(d^3)}$, each new weight can be described by $\O(d^3)$ bits, and, thus, the whole instance can be described using $\O(d^4)$ bits.
\end{proof}

Finally, we present two simple reduction rules; we always assume that \Cref{rrule:stop_condition} is not applicable.

\begin{rrule}\label{rrule:stop_condition}
		Let an instance of \TSPshort, \sTSPshort, or \WRPshort be given.
    If $\budget < 0$, then answer \NO.
    Otherwise, if $|\WP| \leq 1$, then answer \YES.
\end{rrule}
\begin{proof}[Safeness]
    If the budget $\budget$ is negative, we have already exceeded allowed operations, so the instance is trivially \NO.
    In the other case, if $|\WP| \leq 1$, then there is only one waypoint and the solution is always an empty path.
    So the answer is \YES.
\end{proof}

\begin{rrule}\label{rrule:short_circuit}
 Let $I$ be an instance of \sTSP{} and $v \notin \WP$.
 For each pair of vertices $u,w \in N(v)$ we introduce a new edge $\{u,w\}$ into the graph with $\wFn(\{u,w\})=\wFn(\{u,v\})+\wFn(\{v,w\})$ (if this creates parallel edges, then we only keep the one with the lower weight).
 Finally, we remove $v$ together with all its incident edges.
\end{rrule}
\begin{proof}[Safeness]
	Let $I'$ be the resulting instance.
	If $S$ is a solution for $I$, then by omitting all occurrences of vertex $v$ and replacing them with the direct edge between its neighbors, we obtain a solution for $I'$ with at most the same weight.
	If $S'$ is a solution for $I'$, then we can replace each usage of an edge $\{u,w\}$ not present in $I$ by a traversal trough $v$ using the edges $\{u,v\}$ and $\{v,w\}$, obtaining a solution $S$ for $I$ of at most the same weight.
\end{proof}

While \Cref{rrule:short_circuit} is easy to apply, its application may destroy the structure of the input graph.
Therefore, we only apply it to some specific vertices as explained in the subsequent sections.

\section{Polynomial Kernel with Respect to Feedback Edge Set Number}\label{sec:FESkernel}

In this section we give an algorithm to derive a polynomial kernel for \WRP when the parameter is feedback edge set number.
This implies that \TSP and \sTSP also admit a polynomial kernel when parameterized by the feedback edge set number as \WRP is in \NP and the two problems are \NPh.

First, we observe that we can assume that there are no degree-one vertices in the input instance.

\begin{rrule}\label{rrule:noLeafWithCapacityOneEdgeInWRP}
    Let $(G,\WP,\wFn,\cFn,\budget)$ be an instance of \WRP, reduced with respect to \autoref{rrule:stop_condition}.
    If there is a vertex $v\in \WP$ with $\deg(v) = 1$, $\{u\}=N(v)$, and $\cFn(\{u,v\}) = 1$, then answer \NO.
\end{rrule}
\begin{proof}[Safeness]
    Our goal in the \WRP is to find the least cost closed walk~$S$ such that it respects capacities of the edges, and visits every waypoint $w\in \WP$ at least once.
    In a closed walk we must be able to enter and leave any vertex in $S$ at least once.
    But it is not possible for $v$ since the only incident edge is $\{u,v\}$ which can be traversed at most once.
\end{proof}

\begin{rrule}\label{rrule:noLeafnonTerminalInWRP}
    Let $(G,\WP,\wFn,\cFn,\budget)$ be an instance of \sTSP or \WRP and $v\in V(G)\setminus\WP$ be a non-terminal vertex such that $\deg(v) = 1$.
    We remove $v$ from $G$ and continue with the modified instance.
\end{rrule}
\begin{proof}[Safeness]
    Let $S$ be an optimal solution for $(G,\WP,\wFn,\cFn,\budget)$.
    For a contradiction, we assume that $v$ is visited by $S$ at least once.
    Let $u$ be a neighbor of $v$ in $G$ and let $e_i$ be the first occurrence of the edge $\{u,v\}$ in $S$.
    If we remove edges $e_i$ and $e_{i+1}$ from $S$, then we obtain solution $S'$ such that $\wFn(S') \leq \wFn(S)$, since all weights are non-negative.
    This contradicts the optimality of $S$, $v$ is not part of any optimal solution $S$, and it can be safely removed from $G$.
\end{proof}

\begin{rrule}\label{rrule:noLeafWithCapacityTwo+EdgeInWRP}
    Let $(G,\WP,\wFn,\cFn,\budget)$ be an instance of \sTSP or \WRP (reduced with respect to \Cref{rrule:noLeafWithCapacityOneEdgeInWRP} in case of \WRPshort), $v\in \WP$ be a terminal vertex such that $\deg(v) = 1$ and $u\in V(G)$ be the neighbor of $v$.
    We remove $v$ from $G$, decrease the budget $\budget$ by $2\cdot\wFn(\{u,v\})$, and add $u$ to $\WP$.
\end{rrule}
\begin{proof}[Safeness]
    Let $S$ be an optimal solution for $(G,\WP,\wFn,\cFn,\budget)$.
    If $S$ does not traverse $v$, then $S$ cannot be a solution by a problem definition.
    So $S$ visits $v$ at least once.
    For a contradiction, we assume that $S$ traverses $v$ at least twice.
    Let $e_i$ be a second appearance of the edge $\{u,v\}$, where $u$ is the neighbor of $v$, in $S$.
    If we create the solution $S'$ by removing edges $e_i$ and $e_{i+1}$ from $S$, then we obtain another solution such that $\wFn(S') \leq \wFn(S)$ which contradicts the optimality of $S$.
    Thus, $v$ is visited by $S$ exactly once and the predecessor of $v$ in $S$ is always $u$.
    So, we can safely reduce the budget, remove $v$ and make $u$ a terminal.

    Note that in case of \WRPshort we assume that the instance is reduced with respect to \Cref{rrule:noLeafWithCapacityOneEdgeInWRP}.
    Otherwise, we cannot afford two traversals of the edge $\{u,v\}$.
\end{proof}

The following rule is a variant of \Cref{rrule:short_circuit} applicable to \WRPshort{}.

\begin{rrule}\label{rrule:short_circuitWPRFES}
    Let $P=(p_0, p_1, \ldots, p_\ell)$ be a path in $G$ such that $\ell \ge 2$ and $\deg_G p_i=2$ and $p_i \notin \WP$ for every $i \in \{1, \ldots, \ell-1\}$.
    Then delete the inner vertices of $P$, and introduce an edge between vertices $p_0$ and $p_\ell$, with capacity set to $\min_{e \in E(P)} \cFn(e)$ and weight equal to $\sum_{e \in E(P)} \wFn(e)$.
    Do not change $\WP$ and $\budget$.
\end{rrule}
\begin{proof}[Safeness]
    Since none of the vertices in the path $P'=(p_1,\dots,p_{\ell-1})$ belong to the set $\WP$, if a minimum weighted path traversing all the waypoints in the graph traverses some vertices in $P'$, it shall only be the case that the path $P'$ serves as a sub-path in a larger path.
    Let $S$ be an optimal solution of weight $\budget$ for $G$.
    If the vertices in $P'$ are not contained in $S$, then $S$ is a solution for $G'$ as well.
    If only some vertices of $P'$ belong to $S$, then the edges incident with these vertices are traversed twice, and these vertices can be deleted without affecting the solution.
    If all vertices of $P'$ belong to $S$, then we can replace the path $P'$ by the newly created edge, say $e'$, in $G'$ and have a solution of the same weight.
    Observe that due to the weight of $e'$, the weight of the solution is not affected.
    In the other direction, let $S'$ be a solution of weight $k$ for $G'$.
    If $e'$ does not belong to $S'$, then $S$ is also a solution for $G$.
    If $e'$ does belong to $S'$, then we can replace $e'$ by the path $P'$ and have a solution of same weight in $G$.
\end{proof}

The following rule deals with remaining long paths.
\begin{rrule}\label{rrule:longPathFES}
    Let $P=(p_0, p_1, \ldots, p_\ell)$ be a path in $G$ such that $\ell \ge 3$ and $\deg_G p_i=2$ and $p_i \in \WP$ for every $i \in \{1, \ldots, \ell-1\}$.
    We replace the inner vertices of $P$ by a vertex $x, x \notin V(G)$ and two edges $e'_1=\{p_0,x\}$ and $e'_2=\{x,p_\ell\}$, putting $x$ in $\WP$.
    We set the properties of the edges based on the number of edges of capacity $1$ on $P$ as follows:
    \begin{enumerate}[a)]
        \item If there are distinct $e_1, e_2 \in E(P)$ with $\cFn(e_1)=\cFn(e_2)=1$, then let $\cFn(e'_1)=\cFn(e'_2)=1$, $\wFn(e'_1)=\wFn(e_1)$, and $\wFn(e'_2)=\sum_{e \in (E(P)\setminus \{e_1\})} \wFn(e)$.
        \item If there is exactly one $e_1 \in E(P)$ with $\cFn(e_1)=1$, then let $\cFn(e'_1)=1$, $\cFn(e'_2)=2$, $\wFn(e'_1)=\wFn(e_1)$, and $\wFn(e'_2)=\sum_{e \in (E(P)\setminus \{e_1\})} \wFn(e)$.
        \item If $\cFn(e) \ge 2$ for every edge $e \in E(P)$, then let $e_1$ be an edge with the largest weight among edges of $E(P)$.
            We let $\cFn(e'_1)=1$, $\cFn(e'_2)=2$, $\wFn(e'_1)=\wFn(e_1)$, and $\wFn(e'_2)=\sum_{e \in (E(P)\setminus \{e_1\})} \wFn(e)$ as in the previous case.
            Furthermore, in this case, we add to $G$ edge $e'_3=\{p_0,p_\ell\}$ and let $\cFn(e'_3)=1$ and $\wFn(e'_3)=\sum_{e \in E(P)} \wFn(e)$.
    \end{enumerate}
    In all cases, we leave the budget $\budget$ unchanged.
\end{rrule}
\begin{proof}[Safeness]
Let $P'=(p_1,\dots,p_{\ell-1})$ be the subpath of $P$ that does not include the vertices $p_0,p_\ell$.
\begin{enumerate}[a)]
    \item If there are distinct $e_1, e_2 \in E(P)$ with $\cFn(e_1)=\cFn(e_2)=1$.

        Observe that if two edges $e_1,e_2$ in $P$ have capacity only one, then at least one traversal of the vertices in $V(P)$ involves traversing whole of $P$ from $p_0$ to $p_\ell$ or vice-versa.
        Traversing some edges in $P$ more than once is not efficient since all the vertices in $V(P)$ have already been traversed and $P$ cannot be used as a bridge between two subgraphs, more than once.
        Thus no vertex/edge in $P$ is traversed twice.
        In such a case, we arbitrarily assign the weight of  edge $e_1$ to $e_1'$ and the sum of remaining edge weights to the other edge $e_2'$.
        And we assign capacities one to both $e_1'$ and $e_2'$.
        Note it is inevitable that $p_0,p_\ell$ are traversed if $V(P')\subseteq \WP$.
        If $S$ is a solution for $G$ then $S$ contains $P$.
        If $G'$ is the reduced graph, then a solution for $G'$ can be obtained from $S$ by replacing $P'$ by the vertex $x$ and including the edge weights of $e_1'$ and $e_2'$.
        Similarly, a solution $S'$ for the reduced graph $G'$ can be used to generate a solution of the same weight for $G$, by simply replacing $x$ with the path $P'$ with its original edge weights.

    \item If there is exactly one $e_1 \in E(P)$ with $\cFn(e_1)=1$, then let $\cFn(e'_1)=1$, $\cFn(e'_2)=2$, $\wFn(e'_1)=\wFn(e_1)$, and $\wFn(e'_2)=\sum_{e \in (E(P)\setminus \{e_1\})} \wFn(e)$.

        Let $e_1=(a,b)$.
        Observe that here $P$ can be traversed in two ways, the first being traversed just once from one end to another, and the second way is to traverse from one end of $P$, say $p_0$, to the vertex $a$, and then traverse back to $p_0$, and then later traverse from another end of $P$, i.e.
        $p_\ell$, to the vertex $b$, and then traverse back to $p_\ell$.
        Thus we assign the weight of $e_1$ to $e_1'$ while setting its capacity one, and the sum total of remaining edge weights to $e_2'$ with capacity two.

        Let $S$ be a solution for $G$.
        If $S$ traverses $P$ just once, we can obtain a solution with same weight $S'$ for $G'$ by traversing the edges $e_1',e_2'$ once.
        If $S$ traverses all edges in $P$ twice except for the edge $e_1$, then we can obtain $S'$ by ignoring the edge $e_1'$ and including the edge $e_2'$ twice in the solution.

        In the other way, let $S'$ be a solution for the reduced graph $G'$.
        If $S'$ contains $e_1',e_2'$ being traversed once, they can be replaced by $P$ to obtain a solution $S$ of the same weight for the original graph $G$.
        If $S'$ does not contain $e_1'$ but contains $e_2'$ being traversed twice, then we can obtain $S$ in the following way.
        Let $e_1=\{p_i,p_{i+1}\}$.
        First remove the two traversals of $e'_2$ from $S'$.
        Then find a place in $S'$ where $p_0$ is traversed and add a traversal of $p_0, \ldots, p_i, \ldots, p_0$ to that place.
        Finally find a place in $S'$ where $p_\ell$ is traversed and add a traversal of $p_\ell \ldots, p_{i+1}, \ldots, p_\ell$ to that place.
        This way each edge of $P$ except for $e_1$ is traversed exactly twice and the weights of $S$ and $S'$ are the same.

    \item If $\cFn(e) \ge 2$ for every edge $e \in E(P)$.

        Observe that if all edges in $P$ can be traversed twice, there are three ways to traverse the path $P$: (i) start from $p_0$ and traverse until $p_\ell$, (ii) start from $p_\ell$ and traverse until $p_0$, (iii) start from either $p_0$ and traverse until some vertex $p_i$ ($i<\ell$) and then traverse back to $p_0$, and later traverse from $p_\ell$ till $p_{i+1}$ (note that here traversing $p_i$ again is redundant and can be avoided) and then traverse back to $p_\ell$.
        Note that the motivation for traversing in the (iii) way, can be a heavy edge, say $e_h$, between $p_i$ and $p_{i+1}$.
        In such a case $e_h$ is naturally a/the heaviest edge in $P$.
        If there are more than one edges with heaviest weight, we arbitrarily pick one of them, and refer to it as $e_1$.
        Hence, we assign $e_1'$ with the weight of $e_1$ and capacity one, while the sum of weight of remaining edges gets assigned to $e_2'$.
        In order to cover the possibility of whole of $P$ being first traversed by (i) and then (ii), we introduce another edge $e_3'$ between $p_0$ and $p_\ell$ with the sum of all edges weights in $P$ and capacity two.

        On the one hand, if there exists a solution $S$ in $G$, then $V(P)$ belongs to $S$.
        If $S$ traverses whole of path $P$ once or twice, then we can replace the first traversal of $P$ in $S$ by the traversal of $e'_1, x, e'_2$ (or $e'_2, x, e'_1$, as necessary) and the eventual second traversal of $P$ in $S$ by the edge $e_3'$. This way we obtain a solution $S'$ of the same weight for the reduced graph $G'$.
        If $S$ excludes a heaviest edge in $P$, then $S'$ can be obtained by replacing $P$ in $S$ by $x$ and double traversal of $e_2'$ at $p_\ell$ and ignoring the edge $e_1'$ altogether.

        On the other hand, if $S'$ is a solution for $G'$ then $x$ belongs to $S'$ and, hence, $S'$ traverses $e_2'$ at least once.
        Then a solution $S$ can be obtained for $G$ as follows.
        If $S'$ also traverses $e_1'$ once, then we include in $S$ the path $P$ once between $p_0,p_\ell$.
        If $S'$ involves traversing $e_3'$ then we replace it by the path $P$ in $S'$.
        If $S'$ involves traversing only $e_2'$ twice, and not traversing $e_1'$ and $e_3'$, then we obtain $S$ in the following way.
        Let $e_1=\{p_i,p_{i+1}\}$.
        First remove the two traversals of $e'_2$ from $S'$.
        Then find a place in $S'$ where $p_0$ is traversed and add a traversal of $p_0, \ldots, p_i, \ldots, p_0$ to that place.
        Finally find a place in $S'$ where $p_\ell$ is traversed and add a traversal of $p_\ell \ldots, p_{i+1}, \ldots, p_\ell$ to that place.
        This way each edge of $P$ except for $e_1$ is traversed exactly twice and the weights of $S$ and $S'$ are the same.

        In all cases the obtained walk traverses all terminals and is of the same cost as $S'$.
\end{enumerate}
\end{proof}

We are now ready to prove \Cref{thm:WRP:FESkernel} which we repeat here for reader's convenience.
\thmWRPFESkernel*
\begin{proof}
After exhaustive application of \Cref{rrule:stop_condition,rrule:noLeafWithCapacityOneEdgeInWRP,rrule:noLeafnonTerminalInWRP,rrule:noLeafWithCapacityTwo+EdgeInWRP}, a reduced graph does not have vertices of degree at most $1$.
As the input graph has feedback edge set of size~$k$, we have $\sum_{v \in V} \deg(v) = 2|E| \le 2 (|V| + k)$ and thus
$
	\sum_{v \in V} \left( \deg(v) - 2 \right) \le 2k
$.
Therefore, it has at most $2k$ vertices of degree at least $3$. %
Now,
\[
	\sum_{\substack{v \in V\\ \deg(v) \ge 3}} \deg(v) = \sum_{v \in V} \left( \deg(v) - 2 \right) + 2 \cdot \sum_{\substack{v \in V\\ \deg(v) \ge 3}} 1 \le 2k + 2\cdot 2k = 6k
\]
and thus it has at most $6k$ edges incident with vertices of degree at least $3$.
Therefore, to bound the size of the reduced instance it remains to bound the number of degree-two vertices.
As each induced path is of length at most $2$ (by exhaustive application of  \Cref{rrule:short_circuitWPRFES,rrule:longPathFES}), it is implied that there are most $6k$ vertices of degree $2$ and, hence, $8k$ vertices and $9k$ edges in total.
Now, we use \Cref{lem:magic_kernel} to reduce the edge weights.
\end{proof}

\section{Polynomial Kernel with Respect to Vertex Cover Number}\label{sec:vc}

\subsection{\TSP}

In this (warm-up) section, we argue that \TSPshort admits a polynomial kernel with respect to the vertex cover number.
That is, we are going to present the most simple use-case of our reduction rules and therefore we can focus on the introduction of the core concept---the natural behavior.
We begin with the definition of a (natural) behavior, which is a formal description of how a vertex ``can behave'' in a solution.
Let $M$ be a vertex cover of $G$ of size $k = \vc(G)$ and let $R= V \setminus M$.

\begin{definition}\label{def:vc_behavior}
    For a vertex $r \in R$ \emph{a behavior of $r$} is a multiset $F \subseteq \{\{r,m\} \mid \{r,m\} \in E, m \in M\}$ containing exactly two edges (edge occurrences).
    We let $B(r)$ be the set of all behaviors of $r$.
    We naturally extend the weight function such that for a behavior $F \in B(r)$ we set its weight to $\wFn(F)=\sum_{e \in F} \wFn(e)$.
    For a vertex $r \in R$ its \emph{natural behavior $\natBeh(r)$} is a fixed minimizer of\/ \( \min \left\{ \wFn(F) \mid F \in B(r) \right\} \) that takes two copies of a minimum weight edge incident with $r$.
\end{definition}

The following lemma shows that in an optimal solution, most of the vertices of $R$ actually use some behavior.

\begin{lemma}\label{lem:vc_mostly_behavior}
Let $S$ be an optimal solution. Then the number of vertices $r \in R$ that are traversed more than once by $S$ is at most $k$.
\end{lemma}
\begin{proof}
    Let $R_1$ be the set of vertices $r \in R$ that are traversed more than once by $S$.
    Assume $|R_1| > k$.
    Let $G_S$ be the corresponding multigraph formed by edges of $S$.
    Let $H$ be a submultigraph of $G_S$ induced by $M \cup R_1$.
    Since each vertex $r \in R_1$ is incident with at least $4$ edges of $G_S$, $H$ has at least $4|R_1|$ edges.
    As $4|R_1| = 2|R_1|+ 2|R_1|  > 2|R_1|+2k =2(|R_1|+k)$ and $H$ has $|R_1|+k$ vertices, by \Cref{lem:many_edges_cycle}, there is a cycle $K$ in $H$, such that $H'=H \setminus E(K)$ has the same connected components as $H$.
    Let $G'_S = G_S \setminus E(K)$.
    The removal of edges according to $K$ changes the degree of each vertex of $K$ by exactly $2$.
    Hence, all degrees in $G'_S$ are even.

    We claim that $G'_S$ is also connected.
    Indeed, if there are two vertices $x$ and $y$ in $V(G_S)$ which are in different connected components of $G'_S$, then each path from $x$ to $y$ in $G_S$ must contain an edge $\{u,v\}$ of $K$.
    However, by the choice of $K$, there is a path from $u$ to $v$ in $H'$ and hence also in $G'_S$, which can be used to obtain a walk between $x$ and $y$ in $G'_S$.

    Therefore $G'_S$ is Eulerian, and an Eulerian trail in $G'_S$ contradicts the optimality of $S$.
\end{proof}

Now, that we know what a behavior is, we can observe that each vertex can have one of two roles in the sought solution---they are either ``just attached'' using the natural behavior or they provide some connectivity between (two) vertices in the vertex cover.
The role of a vertex is formalized as follows.

\begin{definition}\label{def:vc_impact}
    Let $r \in R$ and $F \in B(r)$ be a behavior of $r$.
    We say that $F$ \emph{touches} a vertex $m$ of $M$ if $m$ has nonzero degree in $(\{r\} \cup M,  F)$.
    We call the set $\Imp(F)$ of touched vertices \emph{the impact} of behavior $F$.
\end{definition}

Let $\allImp_r= \{\Imp(F) \mid F \in B(r)\}$ be the set of all possible impacts of behaviors of $r$.
Let $\allImp = \bigcup_{r \in R} \allImp_r$.

Note that we have $1 \le |\Imp(F)| \le 2$ for each behavior $F$.
Therefore, we get $|\allImp| \le \binom{k}{2} +k \le k^2$.
Note also that for each $r \in R$ we have $|\Imp(\natBeh(r))| = 1$.

Now, we show that, in large instances, most of the vertices of~$R$ fall back to their natural behavior in an optimal solution.
Towards this, we take a solution that differs from this in the fewest possible vertices $r \in R$.
Then, we observe that if $r$ is not in the natural behavior, then it is attached to (at least) two vertices in~$M$.
If there are many such ``extra'' edges in a solution, we can apply \Cref{lem:many_edges_cycle}---this would yield a contradiction.
Thus, we get the following.

\begin{lemma}\label{lem:vc_mostly_natural}
    There is an optimal solution $S$ such that for all but at most $3k$ vertices $r \in R$ the solution contains exactly the edges of $\natBeh(r)$ among the edges incident with $r$.
\end{lemma}
\begin{proof}
    Let $S$ be an optimal solution and $G_S$ be the multigraph formed by the edges used by~$S$.
    We assume that $S$ is chosen such that the number of vertices $r \in R$ such that the multiset of edges of $G_S$ incident with $r$ differs from $\natBeh(r)$ is minimized.
    Let $R_1$ be the set of vertices $r \in R$ that are traversed more than once by $S$.
    By \Cref{lem:vc_mostly_behavior} we have $|R_1| \le k$.
    Let $R_2$ be the set of vertices  $r \in R$ that are traversed exactly once by $S$, but for which the set of incident edges differs from $\natBeh$.

    If there is a vertex $r \in R_2$ with exactly one neighbor $u$ in $G_S$, then we can remove from $G_S$ all edges incident with $r$ and add the edges of $\natBeh(r)$ to obtain a graph $G_S'$.
    Let $v \in \Imp(\natBeh(r))$.
    Since $G_S$ is Eulerian, $\deg_{G_S} r$ is even.
    Hence the removal of edges changes the degree of $u$ by an even number, and as $u$ was the only neighbor of $r$, $G_S \setminus \{r\}$ is an Eulerian graph. As the $\natBeh(r)$ consist from two copies of the same edge, its addition changes the degree of $v$ and $r$ by exactly $2$ and makes the graph again connected.
    I.e., $G'_S$ is an Eulerian graph, let $S'$ be an Eulerian trail in $G'_S$.
    Observe that $S'$ is a solution to the instance.
    Since $S$ contained at least $2$ edges incident with $r$ and $S'$ only contains $\natBeh(r)$, the cost of $S'$ is at most that of $S$. Moreover, for $S'$ now there are lesser vertices $r' \in R$ such that the multiset of edges of $G'_S$ incident with $r'$ differs from $\natBeh(r')$, contradicting the choice of $S$.

    Hence every $r\in R_2$ has exactly two neighbors.
    Let $H$ be a multigraph on the vertex set $M$ and initially no edges formed by the following process.
    For each $r \in R_2$, let $u,v$ be its two neighbors. Then add to $A$ the edge $\{u,v\}$ labeled with $r$.
    If $H$ has at most $2k$ edges, then there are at most $2k$ vertices $r \in R_2$ and we are done.

    If $|E(H)| > 2k$, then by \Cref{lem:many_edges_cycle}, there is a cycle $K$ in $H$, such that $H'=H \setminus E(K)$ has the same connected components as $H$.
    For each $r$ such that an edge $\{u,v\}$ labeled $r$ belongs to $K$, we remove the edges $\{r,u\}$ and $\{r,v\}$ from $G_S$ to form $G'_S$.
    For each such $r$ we add $\natBeh(r)$ to $G'_S$ to obtain $G''_S$.
    The removal of edges according to $K$ changes the degree of each vertex in $R$ by an even number and the degree of each vertex of $K$ by exactly $2$.
    The addition of $\natBeh(r)$ increases the degree of $r$ and the vertex in $\Imp(\natBeh(r))$ by exactly $2$.
    Hence, all degrees in $G''_S$ are even.

    We claim that $G''_S$ is also connected.
    Indeed, if there are two vertices $x$ and $y$ in $M$ which are in different connected components of $G''_S$, then there must be $u$ and $v$ in $M$ such that $\{u,v\} \in E(K)$ and $u$ and $v$ are in different connected components of $G''_S$.
    However, then in $H'$ there is a path from $u$ to $v$, which can be translated to a path from $u$ to $v$ in $G_S$.
    Moreover, this path is not touched by the removals, contradicting $u$ and $v$ being in different components.
    Also we ensured that no $r \in R$ is isolated in $G''_S$, hence, each $r \in R$ is in the same connected component as some $x \in M$.

    Therefore $G''_S$ is Eulerian. Let $S''$ be an Eulerian trail in $G''_S$.
    Observe that $S''$ is a solution to the instance.
    Since each addition of $\natBeh(r)$ is preceded by a removal of at least two edges incident with $r$, the cost of $S''$ is at most that of $S$.
    If the cost is strictly lower, then this contradicts the optimality of $S$.
    Otherwise, in $S''$ there are lesser vertices $r \in R$ such that the multiset of edges of $G''_S$ incident with $r$ differs from $\natBeh(r)$, contradicting the choice of $S$.
\end{proof}

Now, we know that there are only a few vertices that behave unnaturally in an optimal solution, we would like to keep these in the kernel.
We arrive at the question of which vertices to keep.
To resolve this question, we observe that if a vertex deviates from its natural behavior, we (possibly) have to pay for this some extra price (which comes in the exchange of the provided connectivity).
Thus, we would like to keep sufficiently many vertices for which this deviation is cheap.
To this end, we first formalize the price.

\begin{definition}
	For $r \in R$ and an impact $I \in \allImp$ let \emph{the price $P(r,I)$ of change} from $\natBeh(r)$ to $I$ at $r$ be $\wFn(F_I) - \wFn(\natBeh(r))$ if there is a behavior $F_I \in B(r)$ with $\Imp(F_I)=I$ and we let it be $\infty$ otherwise.
\end{definition}

Note that if there is a behavior $F_I \in B(r)$ with $\Imp(F_I)=I$, then it is unique.

Now, based on this, we can provide the following reduction rule we employ.

\begin{rrule}\label{rul:vc_rule}
    For each $I \in \allImp$ if there are at most $3k$ vertices $r \in R$ with finite $P(r,I)$, then mark all of them.
    Otherwise mark $3k$ vertices $r \in R$ with the least $P(r,I)$.

    For each unmarked $r \in R$, remove $r$ and decrease $\budget$ by $\wFn(\natBeh(r))$.
\end{rrule}
\begin{proof}[Safeness]
    Let $(G,\wFn,\budget)$ be the original instance and $(\widehat{G}, \widehat{\wFn}, \widehat{\budget})$ be the new instance resulting from the application of the rule.
    Note that $\widehat{\wFn}$ is just the restriction of $\wFn$ to $\widehat{G}$ which is a subgraph of $G$.
    Let $R^-$ be the set of vertices of $R$ removed by the rule.
    Note that $\widehat{\budget} = \budget - \sum_{r \in R^-} \wFn(\natBeh(r))$.
    We first show that if the new instance is a \YESi, then so is the original one.

    Let $\widehat{S}$ be a solution walk in the new instance and let $\widehat{G}_S$ be the corresponding multigraph formed by the edges of $\widehat{S}$.
    Note that the total weight of $\widehat{G}_S$ is at most $\widehat{\budget}$.
    We construct a multigraph $G_S$ by adding to $\widehat{G}_S$ for each $r \in R^-$ the vertex $r$ together with the edge set $\natBeh(r)$.
    Since $\widehat{G}_S$ is connected and each $\natBeh(r)$ is incident with a vertex of $M \subseteq V(\widehat{G}_S)$, $G_S$ is also connected.
    As each addition increases the degree of the involved vertex by exactly $2$ and the degree of each vertex is even in $\widehat{G}_S$, it follows that the degree of each vertex is even in $G_S$.
    Hence $G_S$ is Eulerian and contains all vertices of $G$.
    Since the weight of the added edges is exactly $\sum_{r \in \mathcal{R}^-} \wFn(\natBeh(r))$ it follows that $(G,\wFn,\budget)$ is a \YESi.

    Now suppose that the original instance $(G,\wFn,\budget)$ is a \YESi.
    \begin{claim}\label{cla:RR5_exactly_nat_beh}
    There is an optimal solution $S$ for $(G,\wFn,\budget)$ such that each $r \in R^-$ is incident exactly to the edges of $\natBeh(r)$ in $S$.
    \end{claim}
    \begin{claimproof}
    Let $S$ be a solution guaranteed by \Cref{lem:vc_mostly_natural}, that is, for all but at most $3k$ vertices $r \in R$ the solution contains exactly the edges of $\natBeh(r)$ among the edges incident with $r$.
    Moreover, assume that it further minimizes the number of edges of $G_S$ incident with $R^-$ which are not part of $\bigcup_{r \in R^-} \natBeh(r)$.
    If it contains no such edges, then we are done.
    Suppose that it does contain such edges.

    Let $r_1$ be a vertex incident with an edge of $G_S$ not part of $\natBeh(r_1)$ and let $e_1$ be an arbitrary such edge.
    Consider the occurrence of $r_1$ in $S$ directly preceding or succeeding the traversal of edge $e_1$ and  let $e_2$ be the other edge directly surrounding the occurrence.
    Let $I=\Imp(\{e_1, e_2\})$.

    Since $r_1 \in R^-$, there are at least $3k$ marked vertices $r' \in R$ with $P(r',I) \le  P(r_1,I)$.
    Since at most $3k$ vertices $r'$ are incident with edges different from $\natBeh(r')$ in $G_S$, and $r_1$ is one of them, there is a vertex $r_2 \in R \setminus R^-$ such that $P(r_2,I) \le  P(r_1,I)$ and $r_2$ is only incident with edges of $\natBeh(r_2)$ in $G_S$.
    Let $F_{I} \in B(r_2)$ be the behavior with $\Imp(F_{I})=I$.
    In $S$ we replace the occurrence of edges $e_1$ and $e_2$ and vertex $r_1$ with the edges of $F_{I}$ in appropriate order and an occurrence of vertex $r_2$.
    If this was the only occurrence of $r_1$ in $S$, then we find the occurrence of $v \in \Imp(\natBeh(r_1))$, split it into two and add the edges $\natBeh(r_1)$ and an occurrence of $r_1$ in between.
    Furthermore, find the original occurrence of $r_2$ in $S$, remove it together with the surrounding edges of $\natBeh(r_2)$ and merge the surrounding occurrence of the vertex $u \in \Imp(\natBeh(r_2))$.
    The resulting trail $\widehat{S}$ visits all the vertices, let $\widehat{G}_S$ be the corresponding  multigraph.

    We removed the edges $e_1$, $e_2$, and $\natBeh(r_2)$ from $G_S$ and added the edges of $F_{I}$ and possibly $\natBeh(r_1)$. However, since $\wFn(F_I) - \wFn(\natBeh(r_2))=P(r_2,I) \le  P(r_1,I)=\wFn(\{e_1, e_2\}) - \wFn(\natBeh(r_1))$, we have $\wFn(F_I) + \wFn(\natBeh(r_1)) \le \wFn(\{e_1, e_2\}) + \wFn(\natBeh(r_2))$ and the cost of $\widehat{S}$ is at most the cost of $S$.
    Hence $\widehat{S}$ is also optimal and there are less edges in $\widehat{G}_S$ incident with $R^-$ which are not part of $\bigcup_{r \in R^-} \natBeh(r)$ (namely, we removed edges $e_1$ and $e_2$ and did not introduce any), contradicting our choice of $S$.
    \end{claimproof}

    Let $S$ be a solution for $(G,\wFn,\budget)$ as in the claim and $G_S$ the corresponding multigraph.
    Let $\widehat{G}_S$ be obtained from $G_S$ by removing the edges of $\natBeh(r)$ and vertex $r$ for each $r \in R^-$.
    This reduces the total weight by exactly $\sum_{r \in R^-} \wFn(\natBeh(r))$, hence, as $G_S$ is of weight at most $\budget$, $\widehat{G}_S$ is of weight at most $\widehat{\budget}$.
    Furthermore, as each $r \in R^-$ has only one neighbor in $G_S$, $\widehat{G}_S$ is connected.
    Moreover, each removal decreases the degree of exactly one remaining vertex by exactly $2$, hence all degrees in $\widehat{G}_S$ are even.
    Thus $\widehat{G}_S$ is Eulerian, yielding a solution $\widehat{S}$ for $(\widehat{G}, \widehat{w}, \widehat{\budget})$.
    This completes the proof.
\end{proof}

Since $|\allImp| \le k^2$, the following observation is immediate.

\begin{observation}\label{obs:TSP:vc:kernelBound}
    After \Cref{rul:vc_rule} has been applied, the number of vertices in $R$ is bounded by $3k^3$ and, hence, the total number of edges is $\O(k^4)$.
\end{observation}
\begin{proof}
    For each $I \in \allImp$ we mark at most $3k$ vertices.
    As $|\allImp| \le k^2$, we mark at most $3k^3$ vertices in total.
    Since the unmarked vertices removed, there are at most $3k^3$ vertices in $R$.
    Thus, we have at most $\binom{k}{2}$ edges within $M$ and at most $k \cdot 3k^3$ edges between $M$ and $R$, which are all the edges of the graph and the bound follows.
\end{proof}

We conclude this section in the following theorem (which follows using \Cref{lem:magic_kernel}).

\thmTSPVCkernel*
\begin{proof}
	By \Cref{obs:TSP:vc:kernelBound}, after the application of \Cref{rul:vc_rule}, there are $\O(k^4)$ edges.
	To see that the rule can be applied in polynomial time, observe, that $\natBeh(r)$ can be computed in $\O(k)$ time for each $r$. Then $P(r,I)$ can be computed in constant time for each $r$ and $I$, which allows us to apply the rule in polynomial time.
	Finally, we apply \Cref{lem:magic_kernel} to obtain a kernel of size $\O(k^{16})$ in polynomial time.
\end{proof}

\subsection{\WRP}\label{sec:vc_WRP}

In this section, we use similar approach to \WRPshort as we did for \TSPshort.
We use this opportunity to focus on the differences between these problems.
Again, we begin with the definition of (natural) behavior.
This time we also consider slightly more complicated behaviors.

\begin{definition}
  For a vertex $r \in R$ \emph{a behavior of $r$} is a multiset $F \subseteq \{\{r,c\} \mid \{r,c\} \in E, c \in M\}$ containing two or four edges (edge occurrences) such that each edge is used at most twice and, if $F$ contains twice the same edge~$e$, then $\cFn(e) > 1$.
  Furthermore, if $r \notin \WP$, then $\emptyset$ is also a behavior of $r$.
  We let $B(r)$ be the set of all behaviors of $r$ and $B_2(r)$ the set of behaviors consisting of at most $2$ edge occurrences.
  For a vertex $r \in R$ \emph{its natural behavior $\natBeh(r)$} is a fixed minimizer of\/ \( \min \left\{ \wFn(F) \mid F \in B(r) \right\} \).
\end{definition}

We observe that $\natBeh(r) \in B_2(r)$ for every vertex $r \in R$ and $w(\natBeh(r)) = 0$ for any non-terminal $r$. %

Note that, in this case, we can have vertices $r \in R$ whose natural behavior is attached to two vertices in the vertex cover.
This will later yield one difference in the presented reduction rule.
As we consider more complicated behaviors, their impact is also more complicated.

\begin{definition}
    Let $r \in R$ and $F \in B(r)$ be a behavior of $r$.
    We say that $F$ \emph{touches} a vertex $m$ of $M$ if $m$ has nonzero degree in $(\{r\} \cup M,  F)$.
    Let $T(F)$ be the set of touched vertices.
    Let $D(F)$ be a function $T(F) \to \{1,2\}$, which assigns to each vertex of $T(F)$ its degree in $(\{r\} \cup M,  F)$.
    We call the pair $\Imp(F)=(T(F),D(F))$ \emph{the impact} of behavior $F$.
\end{definition}

Let $\allImp_r= \{\Imp(F) \mid F \in B(r)\}$ be the set of all possible impacts of behaviors of $r$ and $\allImp_{2,r}= \{\Imp(F) \mid F \in B_2(r)\}$.
Let $\allImp = \bigcup_{r \in R} \allImp_r$ and $\allImp_2 = \bigcup_{r \in R} \allImp_{2,r}$.
Since $0 \le |\Imp(F)| \le 4$ for each behavior $F \in B(r)$ and every $r$ and $0 \le |\Imp(F)| \le 2$ for each behavior $F \in B(r)$ and every $r$, we have $|\allImp| = \O(k^4)$ and $|\allImp_2| = \O(k^2)$.

Again, our task will be to keep sufficiently many vertices to provide some extra connectivity needed among the vertices in the vertex cover.
To this end, we would like to keep those that do so at a low cost.
The definition of the price has to be slightly adjusted.

\begin{definition}
    For $r \in R$ and a pair $I, I' \in \allImp$ let \emph{the price $P(r,I,I')$ of change} from $I$ to $I'$ at $r$ be $\wFn(F_{I'}) - \wFn(\natBeh(r))$ if there is a behavior $F_{I'} \in B(r)$ with $\Imp(F_{I'})=I'$ and $I= \Imp(\natBeh(r))$ and we let it be $\infty$ if some of the conditions is not met.
\end{definition}

Note that again if there is a behavior $F_{I'} \in B(r)$ with $\Imp(F_{I'})=I'$, then it is unique. Also note that, as $\Imp(\natBeh(r)) \in \allImp_2$ for every $r$, the price is infinite whenever $I \notin \allImp_2$.

Now, we are ready to present the reduction rule.
Note that this time, we face two difficulties---we are not sure that a solution will visit all the vertices in the vertex cover and there are vertices in~$R$ for which the natural behavior is attached to two vertices in~$M$.
To keep a record of the vertices that are marked towards achieving a different goal, we mark them in different colors.
Thus, the marked vertices of \Cref{rul:vc_rule} are now marked in \emph{red}.

Suppose first, that every vertex in $M$ that is incident with some natural behavior is in $W$.
Then, to follow the lines of safeness for \Cref{rul:vc_rule}, it is sufficient that we ensure that the number of removed vertices from $R$ is even.
Then, when we return these vertices back in their natural behavior, we change the degree of any vertex in~$M$ in the solution by an even number and thus we will have a solution.
We also ensure that at least one vertex with the same impact of natural behavior remains in the instance to provide the connectivity.
This is the task of the vertices marked in \emph{green}.

It remains to secure that every vertex in $M$ that is incident with some natural behavior for many vertices in~$R$ is in $W$.
Towards this, we prove that if a vertex $m \in M$ is needed for a natural behavior for many vertices in~$R$, it is safe to mark $m$ as a terminal.
If this is not the case, we mark the vertices $r \in R$ in \emph{yellow}.
Those we have to keep, however, we will show there are only few of these.

\begin{rrule}\label{rule:wrp_vc}
    For each pair $I \in \allImp_2$, $I' \in \allImp$ of impacts if there are at most $2|\allImp_2|+k$ vertices $r \in R$ with finite $P(r,I,I')$, then mark all of them in red.
    Otherwise mark in red $2|\allImp_2|+k$ vertices $r \in R$ with the least $P(r,I,I')$.

    For each $I = (T,D) \in \allImp_2$ with $T \nsubseteq \WP$, if there are at most $2|\allImp_2|$ unmarked vertices $r \in R$ with $\Imp( \natBeh(r)) = I$, then mark all of them in yellow.
    Otherwise add $T$ to $\WP$.

    For each $I \in \allImp_2$, if there are unmarked vertices $r \in R$ with $\Imp(\natBeh(r))=I$, then if the number of such vertices is odd, then mark one arbitrary such vertex in green.
    If the number of such vertices is even, then mark two arbitrary such vertices in green.

    If $\WP$ was not changed, then for each unmarked $r \in R$, remove $r$ and decrease $\budget$ by $\wFn(\natBeh(r))$.
\end{rrule}
\begin{proof}[Safeness]
    Let $(G,\WP,\wFn,\cFn,\budget)$ be the original instance and $(\widehat{G},\widehat{\WP},\widehat{\wFn},\widehat{\cFn},\widehat{\budget})$ be the new instance resulting from the application of the rule.
    Note that $\widehat{\wFn}$ and $\widehat{\cFn}$ are just the restrictions of $\wFn$ and $\cFn$ to $\widehat{G}$ which is a subgraph of $G$, respectively.
    Let $R^-$ be the set of vertices of $R$ removed by the rule.
    Note that for each $I \in \allImp$ there is an even number of vertices $r \in R^-$ with $\Imp(\natBeh(r))=I$ and that $\widehat{\budget} = \budget - \sum_{r \in R^-} \wFn(\natBeh(r))$.
    We first show that if the new instance is a \YESi, then so is the original one.

    Let $\widehat{S}$ be a solution walk in the new instance and let $\widehat{G}_S$ be the corresponding multigraph formed by edges of $\widehat{S}$.
    Note that the total weight of $\widehat{G}_S$ is at most $\widehat{\budget}$.
    If $\widehat{\WP} \neq \WP$, then $R^-= \emptyset$, $G=\widehat{G}$, $\widehat{\budget}=\budget$, and $S$ is also a solution for the original instance.
    Otherwise we have for each $r \in R^-$ that $T(\natBeh(r)) \subseteq \WP$.
    Then we construct a multigraph $G_S$ by adding to $\widehat{G}_S$ for each $r \in R^-\cap \WP$ the vertex $r$ together with the edge set $\natBeh(r)$.
    Since $\widehat{G}_S$ is connected and each nonempty $\natBeh(r)$ is incident with a vertex of $\WP \cap M \subseteq V(\widehat{G}_S)$, $G_S$ is also connected.
    As we are adding an even number of vertices with the same $\Imp(\natBeh(r))$, the degree of each vertex in $M$ is changed by an even number due to the additions.
    As the degree of each vertex is even in $\widehat{G}_S$, it follows that the degree of each vertex is even in $G_S$.
    Hence $G_S$ is Eulerian and contains all vertices of $\WP$.
    Since the weight of the added edges is exactly $\sum_{r \in \mathcal{R}^-} \wFn(\natBeh(r))$ it follows that $(G,\WP,\wFn,\cFn,\budget)$ is a \YESi.

    Now suppose that the original instance $(G,\WP,\wFn,\cFn,\budget)$ is a \YESi.

    Suppose first that $\widehat{\WP} \neq \WP$.
    Let $S$ be a solution that traverses all vertices of $\WP$ and the most vertices of $\widehat{\WP} \setminus \WP$ among all solutions and $G_S$ be the corresponding multigraph formed by edges of $S$.
    If $S$ traverses all the vertices of $\widehat{\WP}$, then it is also a solution for $(\widehat{G},\widehat{\WP},\widehat{\wFn},\widehat{\cFn},\widehat{\budget})$.
    Suppose $S$ does not traverse a vertex $u \in \widehat{\WP} \setminus \WP$.
    Since $u$ was added to $\widehat{\WP}$, there is an impact $I = (T,D) \in \allImp_2$ with $u \in T$ and at least $2|\allImp_2|+1$ vertices $r\in R$ such that $I= \Imp( \natBeh(r))$.
    Let $R_I$ be the set of all of them.
    Note that $R_I \subseteq \WP$ and, thus, $S$ traverses each vertex in~$R_I$.
    For each $r \in R_I$ let $A_r$ be two arbitrary edges (edges occurrences) of $S$ incident with~$r$.
    That is, $A_r \in B_2(r)$.
    Since $|R_I| \ge 2|\allImp_2|+1$, there are three vertices $r_1, r_2, r_3 \in R_I$ such that $\Imp(A_{r_1})=\Imp(A_{r_2})=\Imp(A_{r_3})$.
    We now want to change the behavior of $r_1$ and $r_2$ to a ``blend'' between the natural behavior and the behavior $A_{r_1}$ or $A_{r_2}$. To this end we first distinguish two cases.

    If $T(A_{r_1})=T(A_{r_2})=\{a\}$, then we have $u\neq a$ as $S$ does not traverse $u$, but does traverse $a$.
    This implies that the edge $\{a,r_1\}$ is contained in $\natBeh(r_1)$ at most once, as it also contains the edge $\{u,r_1\}$.
    Since $\{u,r_1\} \in \natBeh(r_1)$, we have $\wFn(\{u,r_1\}) \le \wFn(\{a,r_1\})$, otherwise we would have $\wFn(\natBeh(r_1))>\wFn(A_{r_1})$ contradicting the choice of $\natBeh(r_1)$.
    Similarly, $\wFn(\{u,r_2\}) \le \wFn(\{a,r_2\})$.

    If $T(A_{r_1})=T(A_{r_1})=\{a,a'\}$, $a \neq a'$, then we have $u \notin \{a,a'\}$ as $S$ does not traverse $u$, but does traverse both $a$ and $a'$.
    Without loss of generality assume that $a \notin T$ (possibly $a' \in T$), i.e., $\{a,r_1\} \notin \natBeh(r_1)$.
    Since $\{u,r_1\} \in \natBeh(r_1)$ and $\{a,r_1\} \notin \natBeh(r_1)$, we have $\wFn(\{u,r_1\}) \le \wFn(\{a,r_1\})$.
    Similarly, $\wFn(\{u,r_2\}) \le \wFn(\{a,r_2\})$.

    In both cases, we remove from $G_S$ (one occurrence of) the edges $\{a,r_1\}$ and $\{a,r_2\}$ and add vertex $u$ together with edges $\{u,r_1\}$ and $\{u,r_2\}$ to obtain $\widehat{G}_S$.
    By the above argument, this does not increase the cost.
    Also the degree of $a$ decreases by $2$ and the degree of $u$ is~$2$, while all the other degrees remain unchanged, thus $\widehat{G}_S$ has all degrees even.
    For every pair of vertices $x$ and $y$ in $V(G_S)$ there was a path between them in $G_S$.
    If it contained neither $\{a,r_1\}$ nor $\{a,r_2\}$, then it is still present in $\widehat{G}_S$.
    Otherwise we can replace $\{a,r_1\}$ with the subpath $a, r_3, a', r_1$ and $\{a,r_2\}$ with the subpath $a, r_3, a', r_2$ to obtain a walk from $x$ to $y$ in $\widehat{G}_S$.
    Hence $\widehat{G}_S$ is also connected.
    Therefore it is Eulerian and the corresponding trail constitutes a solution $\widehat{S}$ that traverses all vertices of $\WP$ and more vertices of $\widehat{\WP} \setminus \WP$ than $S$, contradicting the choice of $S$.
    Hence there is a solution $S$ traversing all vertices of $\widehat{\WP}$ proving the safeness of the rule in the case $\widehat{\WP} \neq \WP$.

    Now suppose that $\widehat{\WP} = \WP$.
    Let $S$ be a solution and $R_1$ the set of vertices $r \in R$ that are traversed more than once by $S$.
    By \Cref{lem:vc_mostly_behavior} we have $|R_1| \le k$. We prove the following stronger condition.

    \begin{claim}
    There is an optimal solution $S$ for $(G,\WP,\wFn,\cFn,\budget)$ such that each $r \in R^-$ and each $r \in R$ marked in green is incident exactly to the edges of $\natBeh(r)$.
    \end{claim}
    \begin{claimproof}
    Let $R^G$ be the set of vertices marked in green.
    Let $S$ be a solution that minimizes the number of edges of $G_S$ incident with $R^- \cup R^G$ which are not part of $\bigcup_{r \in R^- \cup R^G} \natBeh(r)$.
    If it contains no such edges, then we are done.
    So suppose that it contains such edges.

    Let $r_1 \in R^- \cup R^G$ be a vertex incident with an edge of $G_S$ not part of $\natBeh(r_1)$ and let $I = (T,D)= \Imp(\natBeh(r_1))$.
    Note that $T \subseteq \WP$, otherwise $r_1$ would be marked in yellow.
    Let $I'= (T',D') \in \allImp(r_1)$ be an impact of $r_1$ that we identify later.
    We want to show that there is a marked vertex $r_2$ currently in a natural behavior which can achieve impact $I'$ for at most the same cost.

    Since $r_1 \in R^- \cup R^G$, there are at least $2|\allImp_2|+k$ vertices $r' \in R$ with $P(r',I, I') \le  P(r_1,I, I')$ marked in red.
    Let $R_{I'}$ be the set of those among them which are not in $R_1$.
    Note that $|R_{I'}| \ge 2|\allImp_2|$.
    Each $r' \in R_{I'}$ is incident with at most $2$ edges of $G_S$, let $A_{r'}$ be the multiset of them.

    Assume first that for each $r' \in R_{I'}$ we have $A_{r'} \neq \natBeh(r')$, i.e., $\Imp(A_{r'}) \neq I$.
    Since $|R_{I'}| > 2(|\allImp_2|-1)$, this implies that there are three vertices $r_2$, $r_3$, and $r_4$ such that $\Imp(A_{r_2})=\Imp(A_{r_3})=\Imp(A_{r_4})=I''=(T'',D'') \in \allImp_2 \setminus \{I\}$.
    Let $G''_S$ be the multigraph obtained from $G_S$ by replacing the edges of $A_{r_2}$ and $A_{r_3}$ with the edges of $\natBeh(r_2)$ and $\natBeh(r_3)$.
    This cannot increase the cost by the definition of a natural behavior.
    This decreases the degree of vertices $m''$ in $T''$ by exactly $2 \cdot D''(m'')$ and increases the degree of each vertex $m$ in $T$ by exactly $2\cdot D(m)$.
    Hence all the degrees remain even.
    Moreover, as any path through $r_2$ or $r_3$ can be rerouted through $r_4$ and if the vertices $r_2$ and $r_3$ are in $\WP$, then they are connected to $T \subseteq \WP$ in $G''_S$, graph $G''_S$ is also connected.
    Therefore, it yields another optimal solution in which there is a vertex $r_2 \in R_{I'}$ with $A_{r_2} = \natBeh(r_2)$.

    To determine the appropriate $I'$, we distinguish several cases.
    \begin{enumerate}[1)]
    \item In case $|E(G_S) \cap \natBeh(r_1)| \le 1$, let $E(G_S) \cap \natBeh(r_1) = \{e_1\}$ or let $e_1$ be an arbitrary edge of $S$ incident on $r_1$ if $E(G_S) \cap \natBeh(r_1)=\emptyset$.
    Consider the occurrence of $r_1$ in $S$ directly preceding or succeeding the traversal of edge $e_1$ and let $e_2$ be the other edge directly surrounding the occurrence.
    Let $I'=(T',D')=\Imp(\{e_1, e_2\})$.

    By the above argument there is $r_2 \in R_{I'}$ such that $A_{r_2} = \natBeh(r_2)$.
    Let $F_{I'} \in B(r_2)$ be the behavior with $\Imp(F_{I'})=I'$.
    In $S$ we replace the occurrence of edges $e_1$ and $e_2$ and vertex $r_1$ with the edges of $F_{I'}$ in an appropriate order and an occurrence of vertex $r_2$ (note that $I'=\Imp(\{e_1, e_2\})$).
    Similarly, we replace the original occurrence of vertex $r_2$ and the edges of $\natBeh(r_2)$ by the occurrence of $r_1$ and edges of $\natBeh(r_1)$ (note that $I = \Imp(\natBeh(r_1)) = \Imp(\natBeh(r_2))$).
    The resulting walk $S'$ visits all the vertices in $\WP$, let $G'_S$ be the corresponding  multigraph.
    As all the introduced edges are used at most as many times as in $F_{I'}$ or $\natBeh(r_1)$, no capacities are exceeded.

    We removed the edges $e_1$, $e_2$, and $\natBeh(r_2)$ from $G_S$ and added the edges of $F_{I'}$ and $\natBeh(r_1)$. However, since $\wFn(F_{I'}) - \wFn(\natBeh(r_2))=P(r_2,I, I') \le  P(r_1,I, I')=\wFn(\{e_1, e_2\}) - \wFn(\natBeh(r_1))$, we have $\wFn(F_{I'}) + \wFn(\natBeh(r_1)) \le \wFn(\{e_1, e_2\}) + \wFn(\natBeh(r_2))$ and the cost of $S'$ is at most the cost of $S$.
    Hence $S'$ is also optimal and there are less edges in $G'_S$ incident with $R^- \cup R^G$ which are not part of $\bigcup_{r \in R^- \cup R^G} \natBeh(r)$ (namely, we removed at least edge $e_2$ and did not introduce any), contradicting our choice of $S$.

    \item If the case $|E(G_S) \cap \natBeh(r_1)| \le 1$ does not apply, then we have $\natBeh(r_1) \subseteq E(G_S)$, $r_1 \in \WP$ and $r_1$ is visited at least twice by $S$ (since it is also incident with an edge of $G_S$ not part of $\natBeh(r_1)$).
    We again distinguish two cases.
    \begin{enumerate}[i)]
    \item If $|T|=1$, i.e., $T=\{m\}$ and $D\colon m\mapsto 2$, then let $e_1$ be an arbitrary edge of $G_S$ incident on $r_1$ not part of $\natBeh(r_1)$.
    Consider the occurrence of $r_1$ in $S$ directly preceding or succeeding the traversal of edge $e_1$ and let $e_2$ be the other edge directly surrounding the occurrence.
    Let again $I'=(T',D')=\Imp(\{e_1, e_2\})$.

    By the above argument there is $r_2 \in R_{I'}$ such that $A_{r_2} = \natBeh(r_2)$.
    Let $F_{I'} \in B(r_2)$ be the behavior with $\Imp(F_{I'})=I'$.
    In $S$ we replace the occurrence of edges $e_1$ and $e_2$ and vertex $r_1$ with the edges of $F_{I'}$ in an appropriate order and an occurrence of vertex $r_2$ (note that $I'=\Imp(\{e_1, e_2\})$).
    We remove the original occurrence of vertex $r_2$ and the edges of $\natBeh(r_2)$ (note that $T(\natBeh(r_2)) = T(\natBeh(r_1)) = \{m\}$).
    The resulting walk $S'$ visits all the vertices in $\WP$, let $G'_S$ be the corresponding  multigraph.
    As all the introduced edges are used at most as many times as in $F_{I'}$, no capacities are exceeded.

    We removed the edges $e_1$, $e_2$, and $\natBeh(r_2)$ from $G_S$ and added the edges of $F_{I'}$. However, since $\wFn(F_{I'}) - \wFn(\natBeh(r_2))=P(r_2,I, I') \le  P(r_1,I, I')=\wFn(\{e_1, e_2\}) - \wFn(\natBeh(r_1)) \le \wFn(\{e_1, e_2\})$, we have $\wFn(F_{I'}) \le \wFn(\{e_1, e_2\}) + \wFn(\natBeh(r_2))$ and the cost of $S'$ is at most the cost of $S$.
    Hence $S'$ is also optimal and there are less edges in $G'_S$ incident with $R^- \cup R^G$ which are not part of $\bigcup_{r \in R^- \cup R^G} \natBeh(r)$ (namely, we removed at least the edge $e_1$ and did not introduce any), contradicting our choice of $S$.

    \item If $|T|=2$, we distinguish two final cases, namely, whether the two edges of $\natBeh(r_1)$ appear consecutively in $S$ or not.
    \begin{enumerate}[a)]
    \item If they are consecutive in $S$, then let $e_1$ be an arbitrary edge of $G_S$ incident on $r_1$ not part of $\natBeh(r_1)$.
    Consider the occurrence of $r_1$ in $S$ directly preceding or succeeding the traversal of edge $e_1$ and let $e_2$ be the other edge directly surrounding the occurrence.
    Note that, by assumption, $e_2 \notin \natBeh(r_1)$.
    This time, let $I'=(T',D')=\Imp(\natBeh(r_1) \cup \{e_1, e_2\})$.

    By the above argument there is $r_2 \in R_{I'}$ such that $A_{r_2} = \natBeh(r_2)$.
    Let $F_{I'} \in B(r_2)$ be the behavior with $\Imp(F_{I'})=I'$.
    Let $e_1=\{x,r_1\}$ and $e_2=\{y,r_1\}$.
    Since $\Imp(F_{I'})=I'=\Imp(\natBeh(r_1) \cup \{e_1, e_2\})$ and $\Imp(\natBeh(r_2)) = I = \Imp(\natBeh(r_1))$ we have $F_{I'}=\natBeh(r_2)\cup \{\{x,r_2\}, \{y,r_2\}\}$.
    In $S$ we replace the occurrence of edges $e_1$ and $e_2$ and vertex $r_1$ with the edges $\{x,r_2\}$ and $\{y,r_2\}$ of $F_{I'}$ in an appropriate order and an occurrence of vertex $r_2$.
    The resulting walk $S'$ visits all the vertices in $\WP$, let $G'_S$ be the corresponding  multigraph.
    As all the introduced edges are used at most as many times as in $F_{I'}$, no capacities are exceeded.

    We removed the edges $e_1$ and $e_2$ from $G_S$ and added the edges $\{x,r_2\}$ and $\{y,r_2\}$ of $F_{I'}$. However, $\wFn(\{\{x,r_2\}, \{y,r_2\}\})= \wFn(F_{I'}) - \wFn(\natBeh(r_2))=P(r_2,I, I') \le  P(r_1,I, I')=\wFn(\{e_1, e_2\})$ and the cost of $S'$ is at most the cost of $S$.
    Hence $S'$ is also optimal and there are less edges in $G'_S$ incident with $R^- \cup R^G$ which are not part of $\bigcup_{r \in R^- \cup R^G} \natBeh(r)$ (namely, we removed edges $e_1$ and $e_2$ and did not introduce any), contradicting our choice of $S$.

    \item Finally, if the edges of $\natBeh(r_1)$ are not consecutive in $S$, then let $\natBeh(r_1)= \{e_1,e_2\}$.
    Consider the occurrence of $r_1$ in $S$ directly preceding or succeeding the traversal of edge $e_1$ and let $e_3$ be the other edge directly surrounding the occurrence.
    Similarly, let $e_4$ be the edge consecutive to $e_2$ in $S$.
    Note that, by assumption, $e_3,e_4 \notin \natBeh(r_1)$.
    Let $I'=(T',D')=\Imp(\{e_1, e_2,e_3,e_4\})=\Imp(\natBeh(r_1) \cup \{e_3, e_4\})$.

    By the above argument there is $r_2 \in R_{I'}$ such that $A_{r_2} = \natBeh(r_2)$.
    Let $F_{I'} \in B(r_2)$ be the behavior with $\Imp(F_{I'})=I'$.
    Let $e_i=\{x_i,r_1\}$ for every $i \in \{1, \ldots, 4\}$.
    Since $\Imp(F_{I'})=I'=\Imp(\{e_1, e_2, e_3, e_4\})$ and $\Imp(\natBeh(r_2)) = I = \Imp(\natBeh(r_1))$ we have $F_{I'}=\natBeh(r_2)\cup \{\{x_3,r_2\}, \{x_4,r_2\}\}$ and $\natBeh(r_2)= \{\{x_1,r_2\}, \{x_2,r_2\}\}$.
    In $S$ we replace the occurrence of edges $e_1$ and $e_3$ and vertex $r_1$ with the edges $\{x_1,r_2\}$ and $\{x_3,r_2\}$ of $F_{I'}$ in an appropriate order and an occurrence of vertex~$r_2$.
    Similarly, we replace the occurrence of edges $e_2$ and $e_4$ and vertex $r_1$ with the edges $\{x_2,r_2\}$ and $\{x_4,r_2\}$ of $F_{I'}$ in an appropriate order and an occurrence of vertex $r_2$.
    Finally, we replace the original occurrence of edges $\{x_1,r_2\}$ and $\{x_2,r_2\}$ and the original occurrence of $r_2$ with edges $e_1$ and $e_2$ in an appropriate order and an occurrence of vertex $r_1$.
    The resulting walk $S'$ visits all the vertices in $\WP$, let $G'_S$ be the corresponding  multigraph.
    As all the introduced edges are used at most as many times as in $F_{I'}$, no capacities are exceeded.

    We removed the edges $e_3$ and $e_4$ from $G_S$ and added the edges $\{x_3,r_2\}$ and $\{x_4,r_2\}$ of $F_{I'}$. However, $\wFn(\{\{x_3,r_2\}, \{x_4,r_2\}\})= \wFn(F_{I'}) - \wFn(\natBeh(r_2))=P(r_2,I, I') \le  P(r_1,I, I')=\wFn(\{e_3, e_4\})$ and the cost of $S'$ is at most the cost of $S$.
    Hence $S'$ is also optimal and there are less edges in $G'_S$ incident with $R^- \cup R^G$ which are not part of $\bigcup_{r \in R^- \cup R^G} \natBeh(r)$ (namely, we removed edges $e_3$ and $e_4$ and did not introduce any), contradicting our choice of $S$.
    \end{enumerate}%
    \end{enumerate}%
    \end{enumerate}%
    \end{claimproof}

    Let $S$ be a solution for $(G,\WP,\wFn,\cFn,\budget)$ as in the claim and $G_S$ the corresponding multigraph.
    Let $\widehat{G}_S$ be obtained from $G_S$ by removing the edges of $\natBeh(r)$ and vertex $r$ for each $r \in R^-$.
    This reduces the total weight by exactly $\sum_{r \in R^-} \wFn(\natBeh(r))$.
    Hence, as $G_S$ is of weight at most~$\budget$, $\widehat{G}_S$ is of weight at most $\widehat{\budget}$.
    Note also that for each $I \in \allImp$ such that there is a vertex $r \in R^-$ with $\Imp(\natBeh(r))=I$, there is an even number of such vertices and at least one vertex $r_I$ with  $\Imp(\natBeh(r_I))=I$ marked in green.
    Therefore, the removals decrease the degree of all remaining vertices by an even number, hence all degrees in $\widehat{G}_S$ are even.
    Moreover, we show that $\widehat{G}_S$ is connected.
    Let $x$ and $y$ be two arbitrary vertices of $\widehat{G}_S$.
    Since $G_S$ is connected, there is a path from $x$ to $y$ in $G_S$.
    We replace each vertex $r \in R^-$ on this path by the green vertex $r_{\Imp(\natBeh(r))}$ to obtain a walk from $x$ to $y$ in $\widehat{G}_S$.
    Since $x$ and $y$ were arbitrary, $\widehat{G}_S$ is indeed connected.
    Thus $\widehat{G}_S$ is Eulerian, yielding a solution $\widehat{S}$ for $(\widehat{G},\widehat{\WP}, \widehat{\wFn}, \widehat{\cFn}, \widehat{\budget})$.
    This completes the proof.
\end{proof}

We apply \Cref{rule:wrp_vc} and we compute the size of the reduced instance.

\begin{observation}\label{obs:WRP:vc:kernelBound}
    After \Cref{rule:wrp_vc} has been exhaustively applied, the number of vertices in $R$ is bounded by $\O(k^8)$ and, hence, the total number of edges is $\O(k^9)$.
\end{observation}

\begin{proof}
    For each pair $I \in \allImp_2$, $I' \in \allImp$ we mark at most $2|\allImp_2|+k$ vertices in red.
    For each $I \in \allImp_2$ we mark at most $2|\allImp_2|$ vertices in yellow.
    Finally, for each $I \in \allImp_2$ we mark at most $2$ vertices in green.
    As $|\allImp_2| = \O(k^2)$ and $|\allImp| = \O(k^4)$, we mark at most $\O(k^8)$ vertices in total.
    Since the unmarked vertices are removed, there are at most $\O(k^8)$ vertices in $R$.
    Thus, we have at most $\binom{k}{2}$ edges within $M$ and at most $\O(k \cdot k^8)$ edges between $M$ and $R$, which are all the edges of the graph and the bound follows.
\end{proof}

Again, the theorem follows by \Cref{obs:WRP:vc:kernelBound} and \Cref{lem:magic_kernel}.

\thmWRPVCkernel*
\begin{proof}
By \Cref{obs:WRP:vc:kernelBound}, after the application of \Cref{rule:wrp_vc}, there are $\O(k^9)$ edges.
To see that the rule can be applied in polynomial time, observe, that $\natBeh(r)$ can be computed in $\O(k)$ time for each $r$. Then $P(r,I,I')$ can be computed in constant time for each $r$, $I$, and $I'$ which allows us to apply the rule in polynomial time.
Finally, we apply \Cref{lem:magic_kernel} to obtain a kernel of size $\O(k^{36})$ in polynomial time.
\end{proof}

\section{Polynomial Kernels with Respect to Modulator Size}\label{sec:modulatorsPolyKernel}

In this section, we move to more general type of modulators.
Recall that if $M$ is a vertex cover, the connected components of $G \setminus M$ are single vertices.
We are going to relax this a bit---we shall investigate general components of constant size and paths of constant size.
For the first, more general one, we obtain a polynomial kernel for \TSPshort.
For the second, we obtain a polynomial kernel even for \sTSPshort.

\subsection{TSP and the Distance to Constant Size Components}\label{sec:components}

For the rest of the section, we assume that we are given an undirected graph $G = (V,E)$, $k$ is its distance to $r$-components, $M$ the corresponding modulator, $E_M= \binom{M}{2}$, and $R = G \setminus M$.
Let $M=\{m_1, \ldots, m_k\}$ be a fixed order of the vertices of the modulator.

As usually, we begin with the definition of a behavior.

\begin{definition}\label{def:behavior}
    For a connected component $C$ of $R$ we consider the graph $G_C = {G[C \cup M]} \setminus E_M$.
    \emph{A behavior of $C$} is a multiset $F \subseteq E[G_C]$ of edges of $G_C$ (each can be used at most twice) such that
	    \begin{itemize}
	        \item in the multigraph $(C \cup M,  F)$,
	            \begin{enumerate}[(i)]
	                \item each vertex $v \in C$ has nonzero even degree;
	                \item each connected component contains a vertex of $M$;
	            \end{enumerate}
	        \item $F$ contains at most $2r$ edges incident with vertices of $M$.
	    \end{itemize}
    We let $B(C)$ be the set of all behaviors of $C$.
    We again naturally extend the weight function such that for a behavior $F \in B(C)$ we set its weight to $\wFn(F)=\sum_{e \in F} \wFn(e)$.
    For a component $C$ \emph{its natural behavior $\natBeh(C)$} is a fixed minimizer of\/ \( \min \left\{ \wFn(F) \mid F \in B(C) \right\} \).
\end{definition}

Now, it is not hard to see, that a solution might ``visit'' a component of $G \setminus M$ more than once.
With this, we arrive at the key notion we introduce in this section---the segments.
It is not hard to see that a solution may even revisit some terminals in a component; we shall later prove that this is the case for a few exceptional components only (see \Cref{lem:mostly_natural}).

\begin{definition}\label{def:solution_parts}
    Let $C$ be a connected component of $R$ and $S$ be a walk starting and ending in $v_0 \in M$ forming a nice optimal solution to the instance.
    We split $S$ into \emph{segments} $S_1, \ldots, S_q$ such that each segment starts and ends in a vertex of $M$, whereas the inner vertices of each segment are from $R$.
    Let $\mathcal{F}(S,C)$ be obtained from the empty set by the following process: For each $i$, add $S_i$ into $\mathcal{F}(S,C)$ if and only if there is a vertex $c \in C$ visited by $S_i$ and not visited by any of the previous segments.
    Let $F(S,C)$ be the union of edges of $S_i$, $S_i \in \mathcal{F}(S,C)$, including multiplicities.
\end{definition}

We observe that the multiset $F(S,C)$ is a behavior of $C$, i.e, $F(S,C) \in B(C)$.%

\begin{observation}\label{obs:solution_parts}
    The multiset $F(S,C)$ is a behavior of $C$, i.e, $F(S,C) \in B(C)$.
\end{observation}
\begin{proof}
    Since the original solution uses each edge at most twice and $F(S,C)$ is a subset of it, it also uses every edge at most twice.
    All edges of $F(S,C)$ are in $E[G_C]$, because any other edges would be in segments of $S$ not included in $\mathcal{F}(S,C)$.
    As each vertex is visited by $S$ and we included the first segment visiting each $c \in C$ in $\mathcal{F}(S,C)$, each vertex of $C$ has nonzero degree in $(C \cup M,  F(S,C))$.
    Furthermore, as each segment starts and ends in a vertex of $M$, each connected component of $(C \cup M,  F(S,C))$ contains a vertex of $M$ and each vertex of $C$ has an even degree with respect to each segment of $S$, hence an even degree in $(C \cup M,  F(S,C))$.
    Finally, each segment of $S$ contains at most $2$ edges incident with $M$.
    Since $C$ has at most $r$ vertices and each segment in $\mathcal{F}(S,C)$ contains a vertex of $C$ not contained in any previous segment, $F(S,C)$ contains at most $2r$ edges incident with $M$.
\end{proof}

Again, for a connected component of $G \setminus M$, we want to introduce the notion of touched vertices and the (connectivity) impact of a behavior of the component.
For a better understanding of \Cref{def:impact}, we refer the reader to an example in \Cref{fig:impact_def}.

\begin{definition}\label{def:impact}
    Let $C$ be a connected component of $R$ and $F \in B(C)$ be a behavior of $C$.
    We say that $F$ \emph{touches} a vertex $m$ of $M$ if $m$ has nonzero degree in $(C \cup M,  F)$.
    Let $T(F)$ be the set of touched vertices.

    Consider the multiset $D(F)$ of edges obtained from the empty set as follows.
    For each connected component $H$ of $(C \cup M,  F)$ containing at least two vertices of $M$, let $m_i$ be the vertex with the least index in $M \cap V(H)$.
    For each $m_j$ in $(M \cap V(H)) \setminus \{m_i\}$ add to $D(F)$ a single edge $\{m_i,m_j\}$ if $m_j$ is incident with an odd number of edges of $F$ and a double edge $\{m_i,m_j\}$ if $m_j$ is incident with an even number of edges of $F$.

    We call the pair $\Imp(F)=(T(F),D(F))$ \emph{the impact} of behavior $F$.

    Let $\allImp_C= \{\Imp(F) \mid F \in B(C)\}$ be the set of all possible impacts of behaviors of $C$.

    Let $\allImp = \bigcup_{C \text{ connected component of } R} \allImp_C$.
\end{definition}

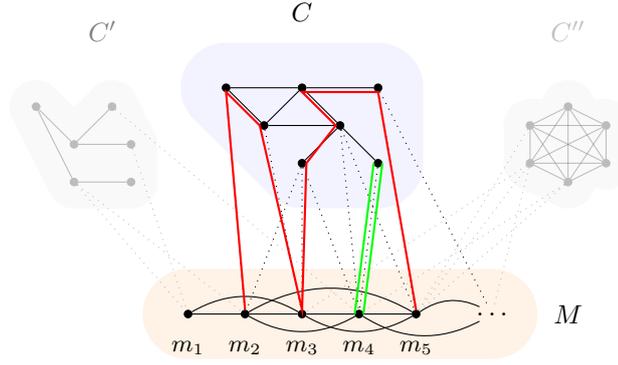
\begin{figure}[tb!]
	\centering
	\newcommand{\convexpath}[2]{
		[
		create hullnodes/.code={
			\global\edef\namelist{#1}
			\foreach [count=\counter] \nodename in \namelist {
				\global\edef\numberofnodes{\counter}
				\node at (\nodename) [draw=none,name=hullnode\counter] {};
			}
			\node at (hullnode\numberofnodes) [name=hullnode0,draw=none] {};
			\pgfmathtruncatemacro\lastnumber{\numberofnodes+1}
			\node at (hullnode1) [name=hullnode\lastnumber,draw=none] {};
		},
		create hullnodes
		]
		($(hullnode1)!#2!-90:(hullnode0)$)
		\foreach [
		evaluate=\currentnode as \previousnode using \currentnode-1,
		evaluate=\currentnode as \nextnode using \currentnode+1
		] \currentnode in {1,...,\numberofnodes} {
			let
			\p1 = ($(hullnode\currentnode)!#2!-90:(hullnode\previousnode)$),
			\p2 = ($(hullnode\currentnode)!#2!90:(hullnode\nextnode)$),
			\p3 = ($(\p1) - (hullnode\currentnode)$),
			\n1 = {atan2(\y3,\x3)},
			\p4 = ($(\p2) - (hullnode\currentnode)$),
			\n2 = {atan2(\y4,\x4)},
			\n{delta} = {-Mod(\n1-\n2,360)}
			in
			{-- (\p1) arc[start angle=\n1, delta angle=\n{delta}, radius=#2] -- (\p2)}
		}
		-- cycle
	}

	\begin{tikzpicture}
		\node[draw,fill=black,circle,inner sep=1pt] (C1v0) at (0,2.5) {};
		\node[draw,fill=black,circle,inner sep=1pt] (C1v1) at (-0.5,3) {};
		\node[draw,fill=black,circle,inner sep=1pt] (C1v2) at (0.5,3) {};
		\node[draw,fill=black,circle,inner sep=1pt] (C1v3) at (1,2.5) {};
		\node[draw,fill=black,circle,inner sep=1pt] (C1v4) at (0,3.5) {};
		\node[draw,fill=black,circle,inner sep=1pt] (C1v5) at (-1,3.5) {};
		\node[draw,fill=black,circle,inner sep=1pt] (C1v6) at (1,3.5) {};
		\draw (C1v0) -- (C1v2) -- (C1v1) -- (C1v4) -- (C1v5) -- (C1v1);
		\draw (C1v4) -- (C1v2) -- (C1v3);
		\draw (C1v4) -- (C1v6);
		\node (C1) at (0,4.5) {$C$};
		\begin{pgfonlayer}{background}
			\fill[blue,opacity=0.05] \convexpath{C1v0,C1v5,C1v6,C1v3}{17pt};
		\end{pgfonlayer}

		\node[draw,fill=black,circle,inner sep=1pt,label={[label distance=5pt]270:\small$m_1$}] (m1) at (-1.5,0.5) {};
		\node[draw,fill=black,circle,inner sep=1pt,label={[label distance=5pt]270:\small$m_2$}] (m2) at (-0.75,0.5) {};
		\node[draw,fill=black,circle,inner sep=1pt,label={[label distance=5pt]270:\small$m_3$}] (m3) at (0,0.5) {};
		\node[draw,fill=black,circle,inner sep=1pt,label={[label distance=5pt]270:\small$m_4$}] (m4) at (0.75,0.5) {};
		\node[draw,fill=black,circle,inner sep=1pt,label={[label distance=5pt]270:\small$m_5$}] (m5) at (1.5,0.5) {};
		\node[inner sep=2pt] (m6) at (2.5,0.5) {$\ldots$};
		\draw (m1) -- (m2) -- (m3) -- (m4) -- (m5);
		\draw (m1) edge[bend left] (m3);
		\draw (m2) edge[bend right] (m4) edge[bend left] (m5);
		\draw (m3) edge[bend right] (m5);
		\draw (m6) edge[bend left] (m4) edge[bend right] (m5);
		\node (M) at (3.5,0.5) {$M$};
		\begin{pgfonlayer}{background}
			\fill[orange,opacity=0.1] \convexpath{m1,m6}{17pt};
		\end{pgfonlayer}

		\draw[dotted] (C1v0) edge (m2) edge (m3) edge (m4);
		\draw[dotted] (C1v2) edge (m4) edge (m5);
		\draw[dotted] (C1v3) edge (m4);
		\draw[dotted] (C1v6) edge (m5) edge (m6);
		\draw[dotted] (C1v1) edge (m3);
		\draw[dotted] (C1v5) edge (m2);

		\draw[red,thick] (m2.north) -- (C1v5.south) -- (C1v1.west) -- (m3.north) -- (C1v0.east) -- (C1v2.west) -- (C1v4.south) -- (C1v6.south) -- (m5.north);
		\draw[green,thick] (m4.west) -- (C1v3.west);
		\draw[green,thick] (m4.east) -- (C1v3.east);

		\node[draw,gray,fill=gray,circle,inner sep=1pt,opacity=0.5] (t0) at (-3.5,3.25) {};
		\node[draw,gray,fill=gray,circle,inner sep=1pt,opacity=0.5] (t1) at (-3,2.75) {};
		\node[draw,gray,fill=gray,circle,inner sep=1pt,opacity=0.5] (t2) at (-3,2.25) {};
		\node[draw,gray,fill=gray,circle,inner sep=1pt,opacity=0.5] (t3) at (-2.5,3.25) {};
		\node[draw,gray,fill=gray,circle,inner sep=1pt,opacity=0.5] (t4) at (-2.25,2.75) {};
		\node[draw,gray,fill=gray,circle,inner sep=1pt,opacity=0.5] (t5) at (-2.25,2.25) {};
		\draw[gray,opacity=0.5] (t0) -- (t1) -- (t3);
		\draw[gray,opacity=0.5] (t1) -- (t4);
		\draw[gray,opacity=0.5] (t0) -- (t2) -- (t5);
		\node[gray] (T) at (-2.625,4.25) {$C'$};
		\draw[dotted,gray,opacity=0.5] (t2) edge (m2) edge (m1);
		\draw[dotted,gray,opacity=0.5] (t4) edge (m1);
		\draw[dotted,gray,opacity=0.5] (t3) edge (m4);
		\begin{pgfonlayer}{background}
			\fill[gray,opacity=0.05] \convexpath{t0,t1,t3,t4,t5,t2}{10pt};
		\end{pgfonlayer}

		\node[draw,gray,fill=gray,circle,inner sep=1pt,opacity=0.5] (cq0) at (3.5,3.25) {};
		\node[draw,gray,fill=gray,circle,inner sep=1pt,opacity=0.5] (cq1) at (3.5,2.25) {};
		\node[draw,gray,fill=gray,circle,inner sep=1pt,opacity=0.5] (cq2) at (3,2.5) {};
		\node[draw,gray,fill=gray,circle,inner sep=1pt,opacity=0.5] (cq3) at (3,3) {};
		\node[draw,gray,fill=gray,circle,inner sep=1pt,opacity=0.5] (cq4) at (4,3) {};
		\node[draw,gray,fill=gray,circle,inner sep=1pt,opacity=0.5] (cq5) at (4,2.5) {};
		\draw[gray,opacity=0.5] (cq0) edge (cq1) edge (cq2) edge (cq3) edge (cq4) edge (cq5);
		\draw[gray,opacity=0.5] (cq1) edge (cq2) edge (cq3) edge (cq4) edge (cq5);
		\draw[gray,opacity=0.5] (cq2) edge (cq3) edge (cq4) edge (cq5);
		\draw[gray,opacity=0.5] (cq3) edge (cq4) edge (cq5);
		\draw[gray,opacity=0.5] (cq4) edge (cq5);
		\node[gray,opacity=0.5] (CQ) at (3.5,4.25) {$C''$};
		\draw[dotted,gray,opacity=0.5] (cq2) edge (m6) edge (m3) edge (m5);
		\draw[dotted,gray,opacity=0.5] (cq1) edge (m5) edge (m4);
		\draw[dotted,gray,opacity=0.5] (cq3) edge (m5);
		\begin{pgfonlayer}{background}
			\fill[gray,opacity=0.05] \convexpath{cq0,cq1,cq2,cq3,cq4,cq5}{10pt};
		\end{pgfonlayer}
	\end{tikzpicture}
	\caption{Illustration of \Cref{def:impact}. A behavior $F$ consists of red and green edges, and vertices $m_2,\ldots,m_5$ made the set $T(F)$ of touched vertices. We built the multiset $D(F)$ for behavior $F$ as follows. We start with an empty set. For the red connected component, the vertex $m_i\in M$ with least index is $m_2$. Thus, for $m_3$ we add to $D(F)$ double edge $\{m_2,m_3\}$, because the vertex $m_3$ is incident with even number of edges in behavior $F$. For the vertex $m_5$, we add to $D(F)$ the edge $\{m_2,m_5\}$ only once, as the vertex $m_5$ is incident with odd number of edges in $F$. For the green connected component, we do not extend $D(F)$ since, in this connected component, there is only one vertex from $M$.}
	\label{fig:impact_def}
\end{figure}

From the Handshaking Lemma we get the following observation.
\begin{observation}\label{obs:impact}
    Let $m_i$ be as in \Cref{def:impact}.
    Then $m_i$ is incident with an even number of edges of $D(F)$ if and only if it is incident with an even number of edges of $F$.
\end{observation}

Next, we bound the number of possible impacts of behaviors.

\begin{observation}\label{obs:num_impacts}
    $|\allImp| = k^{\O(r)}$.
\end{observation}
\begin{proof}
	Since there are at most $2r$ vertices incident on $M$ in any behavior $F$, there are at most $2r$ vertices in $T(F)$.
	Hence, there are $k^{\O(r)}$ possible sets $T(F)$.
	There are at most $(2r)^{2r} = k^{\O(r)}$ possible ways to split a particular set $T(F)$ to connected components of $D(F)$.
	Finally, we need to determine for each vertex in $T(F)$, whether it has odd or even degree in $D(F)$.
	There are at most $2^{2r} = k^{\O(r)}$ options for that.
	Alltogether, there are $k^{\O(r)} \cdot k^{\O(r)} \cdot k^{\O(r)}= k^{\O(r)}$ possible impacts in total.
\end{proof}

Now, we want to give some properties that an optimal solution should have.
Namely, we want to see that most of the connected components of $G \setminus M$ fallback to their natural behavior.
Towards this, we first prove that there are only a few segments that are not part of a behavior of any component; recall that if the solution ``revisited'' some terminals in the connected component, this segment might not be part of any behavior.

\begin{lemma}\label{lem:mostly_natural}
    There is an optimal solution $S$ such that there are at most $2|M|$ segments of $S$ that are not part of any $\mathcal{F}(S,C)$, and, furthermore, for all but at most $2|\allImp|^2+2|M|$ components we have $F(S,C)=\natBeh(C)$ and $F(S,C)$ contains all edges of $S$ incident with $C$.
\end{lemma}
\begin{proof}
    Let $S$ be an optimal solution and $G_S$ be the multigraph formed by the edges used by~$S$.
    Let $\mathcal{S}$ be the set of segments of $S$ that are not part of any $\mathcal{F}(S,C)$.
    \begin{claim}
        $|\mathcal{S}| \le 2|M|$.
    \end{claim}
    \begin{claimproof}
        Suppose this is not the case.
        Let  $H$ be a multigraph on vertex set $M$ and initially no edges formed by the following process.
        For each $S_i \in \mathcal{S}$, if the segment $S_i$ starts in a vertex $u \in M$ and ends in $v \in M$, then we add to $H$ the edge $\{u,v\}$.

        As $|E(H)|=|\mathcal{S}| > 2|M|$, by \Cref{lem:many_edges_cycle}, there is a cycle $K$ in $H$, such that $H'=H \setminus E(K)$ has the same connected components as $H$.
        Let $K'$ be the closed walk formed by the segments of $\mathcal{S}$ corresponding to the edges of $K$.
        We claim that $G'_S=G_S \setminus E[K']$ is Eulerian.

        Each vertex of $R$ has even degree with respect to each segment.
        Each vertex of $M \cap V(K')$ appears in exactly two consecutive segments of $K'$ and, hence, has degree $2$ with respect to~$K'$.
        Therefore, as $G_S$ has even degrees, so does $G'_S$.

        Suppose that $G'_S$ is disconnected, i.e., there are vertices $u$ and $v$ such that there is no path between $u$ and $v$ in $G'_S$.
        If $u$ is in a component $C$ of $R$, then, since the segments of $\mathcal{F}(S,C)$ are unaffected by the change, there is a vertex $u' \in M$ in the same connected component of $G'_S$ as $u$.
        Similarly for $v$.
        Hence there are vertices $u',v' \in M$ such that there is no path between $u'$ and $v'$ in $G'_S$.
        In $G_S$ there was a walk between $u'$ and $v'$ formed by the segments of $S$.
        If none of the segments is in $K'$, then the walk is still present in $G'_S$.
        Otherwise, by the choice of $K$ and \Cref{lem:many_edges_cycle}, we can replace each edge of $K$ in $H$ by a walk in $H \setminus E(K)$ between its endpoints.
        Thus we can replace any missing segment by a walk formed by remaining segments, obtaining a walk between $u'$ and $v'$ in $G'_S$---a contradiction.
        Thus $G'_S$ is indeed Eulerian, yielding a solution $S'$ contradicting the optimality of $S$.
    \end{claimproof}
    Each $S_i \in \mathcal{S}$ is formed either by edges of some $E[G_C]$ or by a single edge on $M$.
    Let $\mathcal{C}_1 = \{ C \mid \exists S_i \in \mathcal{S} \colon S_i \subseteq E[G_C]\}$.
    We have $|\mathcal{C}_1| \le 2|M|$.
    Let $\mathcal{C}_2$ be the set of components $C$ of $R$ not in $\mathcal{C}_1$ with $F(S,C) \neq \natBeh(C)$.
    We may assume that $S$ is chosen so that $|\mathcal{C}_1|+|\mathcal{C}_2|$ is minimum.
    \begin{claim}
        $|\mathcal{C}_2| \le 2|\allImp|^2$.
    \end{claim}
    \begin{claimproof}
        Suppose this is not the case.
        Then let us classify the components $C$ of $\mathcal{C}_2$ according to pair of impacts $(\Imp(\natBeh(C)), \Imp(F(S,C)))$.
        Due to the number of the components, there are at least 3 that share the same pair, let us call them $C_1,C_2,C_3$.
        We obtain $G'_S$ by removing from $G_S$ the edges of $F(S,C_1)$ and $F(S,C_2)$ and replacing them with edges of $\natBeh(C_1)$ and $\natBeh(C_2)$, respectively.
        Note that by the definition of natural behavior, the total weight of $G'_S$ is at most that of $G_S$.
        We claim that $G'_S$ is again Eulerian.

        Let $(T,D)=\Imp(F(S,C_1))=\Imp(F(S,C_2))=\Imp(F(S,C_3))$.
        By \Cref{def:impact} and \Cref{obs:impact} the degree of each vertex in $M$ is odd with respect to $F(S,C_1)$ if and only if it is odd in $D$ which is if and only if it is odd with respect to $F(S,C_2)$.
        Hence the removal of edges of $F(S,C_1)$ and $F(S,C_2)$ changes the degree of each vertex in $M$ by an even number.
        Similarly, letting $(T',D')=\Imp(\natBeh(C_1))=\Imp(\natBeh(C_2))$, the degree of each vertex in $M$ is odd with respect to $\natBeh(C_1)$ if and only if it is odd in $D$ which is if and only if it is odd with respect to $\natBeh(C_2)$.
        Hence the addition of edges of $\natBeh(C_1)$ and $\natBeh(C_2)$ also changes the degree of each vertex of $M$ by an even number.
        As the degree of each vertex of $M$ is even in $G_S$ it follows that it is also even in $G'_S$.
        The degree of each vertex in $V(R) \setminus (V(C_1) \cup V(C_2))$ is the same in $G_S$ and $G'_S$, that is, even.
        Finally, each vertex of $C_1$ is only incident with edges of $\natBeh(C_1)$, since $C_1 \notin \mathcal{C}_1$.
        Hence it has even degree by \Cref{def:behavior}.
        Similarly for vertices of the component $C_2$.

        Now suppose that $G'_S$ is not connected, i.e., there are vertices $u$ and $v$ such that there is no path between $u$ and $v$ in $G'_S$.
        If $u$ is in a component $C$ of $R$, then, since $G'_S$ contains a behavior of $C$ (either $F(S,C)$ or $\natBeh(C)$), there is a vertex $u' \in M$ in the same connected component of $G'_S$ as $u$.
        Similarly for $v$.
        Hence there are vertices $u',v' \in M$ such that there is no path between $u'$ and $v'$ in $G'_S$.
        In $G_S$ there was a walk between $u'$ and $v'$ formed by the segments of $S$.
        If none of the segments is in $\mathcal{F}(S,C_1) \cup \mathcal{F}(S,C_2)$, then the walk is still present in $G'_S$.
        Let $S_i$ be a segment in $\mathcal{F}(S,C_1)$ which starts in $x \in M$ and ends in $y \in M$.
        Then $x$ and $y$ lie in the same component of $F(S,C_1)$, hence in the same component of $D$ and also in the same component of $F(S,C_3)$.
        Similarly for a segment in $\mathcal{F}(S,C_2)$.
        Therefore for each such segment there is a walk from $x$ to $y$ in $G'_S$.
        Thus we can replace any missing segment by a walk in $G'_S$, obtaining a walk between $u'$ and $v'$ in $G'_S$---a contradiction.
        Thus $G'_S$ is indeed Eulerian, yielding a solution $S'$.

        To get the final contradiction it remains to show that we have $|\mathcal{C}'_1| + |\mathcal{C}'_2| < |\mathcal{C}_1| + |\mathcal{C}_2|$.
        Let $\mathcal{S}'$ be the set of segments of $S'$ that are not part of any $\mathcal{F}(S',C)$.
        Let $\mathcal{C}'_1 = \{ C \mid \exists S_i \in \mathcal{S}' \colon S_i \subseteq E[G_C]\}$.
        Let $\mathcal{C}'_2$ be the set of components $C$ of $R$ not in $\mathcal{C}'_1$ with $F(S',C) \neq \natBeh(C)$.
        We claim that $\mathcal{C}'_1 \cup \mathcal{C}'_2 \subseteq (\mathcal{C}_1 \cup \mathcal{C}_2) \setminus \{C_1,C_2\}$.
        To this end, let $C$ be a component of $R$ not in $(\mathcal{C}_1 \cup \mathcal{C}_2) \setminus \{C_1,C_2\}$.
        Then $C$ is only incident with edges of $\natBeh(C)$ in $G'_S$.
        Thus, if $F(S',C) \neq \natBeh(C)$, then we have $F(S',C) \subsetneq \natBeh(C)$ which contradicts the choice of $\natBeh(C)$.
        Hence, $F(S',C) = \natBeh(C)$ and $C$ is incident with no other edges in $G'_S$, i.e., $C$ is neither in $\mathcal{C}'_1$ not in $\mathcal{C}'_2$.
        Therefore $\mathcal{C}'_1 \cup \mathcal{C}'_2 \subseteq (\mathcal{C}_1 \cup \mathcal{C}_2) \setminus \{C_1,C_2\}$ and $|\mathcal{C}'_1| + |\mathcal{C}'_2| < |\mathcal{C}_1| + |\mathcal{C}_2|$, contradicting the choice of $S$.
    \end{claimproof}
    Hence we have $|\mathcal{C}_1|+|\mathcal{C}_2|\le 2|\allImp|^2+2|M|$, completing the proof.
\end{proof}

Next, we want to introduce the notion of price of a transition between two impacts of a connected component.

\begin{definition}\label{def:impact_change_price}
    For a connected component~$C$ of $R$ and an impact $I \in \allImp$ let $\allBeh(C,I)$ be the set of behaviors $F \in B(C)$ satisfying $\Imp(F)=I$.

    For a connected component~$C$ of $R$ and a pair $I, I' \in \allImp$ let \emph{the price $P(C,I,I')$ of change} from $I$ to $I'$ at $C$ be $\min\{\wFn(F) \mid F \in \allBeh(C,I')\} - \wFn(\natBeh(C))$ if $\allBeh(C,I')$ is non-empty and $I= \Imp(\natBeh(C))$ and we let it be $\infty$ if some of the conditions is not met.
\end{definition}

Now, we are ready to introduce the reduction rule used in this section.
A novel notion here are the components marked in blue.
These are supposed to cover the segments which are not part of any behavior.

\begin{rrule}\label{rul:component_rule}
    For each pair $I, I' \in \allImp$ if there are at most $2|\allImp|^{2}+2|M|$ components $C$ with finite $P(C,I,I')$, then mark all of them in red.
    Otherwise mark $2|\allImp|^{2}+2|M|$ components~$C$ with the least $P(C,I,I')$ in red.

    For each pair of vertices $u,v \in M$, if there is a component $C$ in $R$ such that $G_C$ contains a $u$-$v$-path, then mark a component which contains the shortest such path in blue.

    For each $I \in \allImp$, if there are unmarked components $C$ with $\Imp(\natBeh(C))=I$, then do the following.
    If the number of such components is odd, then mark one arbitrary such component in green.
    If the number of such components is even, then mark two arbitrary such components in green.

    For each unmarked component $C$, remove $C$ from $G$ and reduce $\budget$ by $\wFn(\natBeh(C))$.
\end{rrule}
\begin{proof}[Safeness]
    Let $(G,\wFn,\budget)$ be the original instance and $(\widehat{G}, \widehat{w}, \widehat{\budget})$ be the new instance resulting from the application of the rule.
    Note that $\widehat{w}$ is just the restriction of $w$ to $\widehat{G}$ which is a subgraph of $G$.
    Let $\mathcal{C}^-$ be the set of components of $R$ removed by the rule.
    Note that $\widehat{\budget} = \budget - \sum_{C \in \mathcal{C}^-} \wFn(\natBeh(C))$.
    We first show that if the new instance is a \YESi, then so is the original one.

    Let $\widehat{S}$ be a solution walk in the new instance and let $\widehat{G}_S$ be the corresponding multigraph formed by edges of $\widehat{S}$.
    Note that the total weight of $\widehat{S}$ is at most $\widehat{\budget}$.
    We construct a multigraph $G_S$ by adding to $\widehat{G}_S$ for each component $C$ in $\mathcal{C}^-$ the vertex set $V(C)$ together with the edge set $\natBeh(C)$.
    Since each connected component of $(M \cup C, \natBeh(C))$ contains a vertex of $M$ and $\widehat{G}_S$ is connected, $G_S$ is also connected.
    Furthermore, for each $I \in \allImp$ there is an even number of components $C$ in $\mathcal{C}^-$ such that $\Imp(\natBeh(C))=I$.
    If $I=(T,D)$, then we have for each such $C$ that the degree of each vertex in $M$ is odd in $(M \cup C, \natBeh(C))$ if and only if it is odd in~$D$.
    Hence the degree of each vertex in $M$ changes by an even number by the addition of edges of $\natBeh(C)$ for all such components.
    Moreover, the degree of each vertex in $C$ is even in $(M \cup C, \natBeh(C))$.
    As the degree of each vertex is even in $\widehat{G}_S$, it follows that the degree of each vertex is even in $G_S$.
    Hence $G_S$ is Eulerian and contains all vertices of $G$.
    Since the weight of the added edges is exactly $\sum_{C \in \mathcal{C}^-} \wFn(\natBeh(C))$ it follows that $(G,\wFn,\budget)$ is a \YESi.

    Now suppose that the original instance $(G,\wFn,\budget)$ is a \YESi.
    We first show that there is an optimal solution $S$ such that for each $C$ in $\mathcal{C}^-$ we have $F(S,C)=\natBeh(C)$ and there are no other incident edges of $G_S$.
    Let again $\mathcal{S}$ be the set of segments of $S$ that are not part of any $\mathcal{F}(S,C)$.
    Let $\mathcal{C}_1 = \{ C \mid \exists S_i \in \mathcal{S} \colon S_i \subseteq E[G_C]\}$.
    We may assume that $S$ is chosen so that the number of edges of $S$ (including multiplicities) in connected components that are not marked blue is minimized.
    \begin{claim}\label{claim:blue_comps}
        All components in $\mathcal{C}_1$ are marked blue.
    \end{claim}
    \begin{claimproof}
        Suppose that there is a component $C$ in $\mathcal{C}_1$ which is not marked in blue, and $S_i$ is the corresponding segment connecting vertices $x$ and $y$ of $M$.
        There is a path from $x$ to $y$ in $G_C$, hence, the component $C_{xy}$ which contains shortest such path is marked in blue.
        Consider the graph $G'_S$ obtained by removing $S_i$ from $G_S$ and adding the shortest path $P$ from $x$ to $y$ in $G_{C_{xy}}$ (increasing the multiplicity if some of the edges are already present).
        Since $P$ is a shortest such path, the total weight of $G'_S$ is at most that of $G_S$.
        We claim that $G'_S$ is also Eulerian.

        Each vertex of $R$ has even degree with respect to each segment and also with respect to~$P$.
        Vertices $x$ and $y$ were incident with one edge of the segment $S_i$ each in $G_S$ and are incident with one edge of $P$ each in $G'_S$.
        Therefore, as $G_S$ has even degrees, so does $G'_S$.

        Suppose that $G'_S$ is disconnected, i.e., there are vertices $u$ and $v$ such that there is no path between $u$ and $v$ in $G'_S$.
        If $u$ is in a component $C'$ of $R$, then, since the segments of $\mathcal{F}(S,C')$ are unaffected by the change, there is a vertex $u' \in M$ in the same connected component of $G'_S$ as $u$.
        Similarly for $v$.
        Hence there are vertices $u',v' \in M$ such that there is no path between $u'$ and $v'$ in $G'_S$.
        In $G_S$ there was a walk between $u'$ and $v'$ formed by the segments of $S$.
        If it contained the segment $S_i$, then we can replace it by $P$, obtaining a walk between $u'$ and $v'$ in $G'_S$---a contradiction.
        Thus $G'_S$ is indeed Eulerian, yielding an optimal solution~$S'$ with less edges in connected components not marked blue, contradicting the choice of $S$.
    \end{claimproof}

    Let $\mathcal{C}_2$ be the set of components $C$ of $R$ not in $\mathcal{C}_1$ with $F(S,C) \neq \natBeh(C)$.
    By \Cref{lem:mostly_natural} we may assume that $|\mathcal{C}_1|+|\mathcal{C}_2|\le 2|\allImp|^2+2|M|$ (note that this is not affected by the argument above).
    We may further assume that among all such solutions with all components in $\mathcal{C}_1$ marked blue, $S$ is chosen such that the number of components in $\mathcal{C}_2$ that are not marked red is minimized.
    \begin{claim}\label{claim:red_comps}
        All components of $\mathcal{C}_2$ are marked red.
    \end{claim}
    \begin{claimproof}
        Suppose that $C_{\rm out}$ is in $\mathcal{C}_2$ but not marked in red.
        Let $I=\Imp(\natBeh(C_{\rm out}))$ and $I'=\Imp(F(S,C_{\rm out}))$.
        Note that since $\wFn(F(S,C_{\rm out}))\ge \min\{\wFn(F) \mid F \in \allBeh(C_{\rm out},I')\}$, we have $\wFn(F(S,C_{\rm out}))-\wFn(\natBeh(C_{\rm out})) \ge P(C_{\rm out},I,I')$.
        Since $P(C_{\rm out},I,I')$ is finite and $C_{\rm out}$ is not marked in red, there are at least $2|\allImp|^{2}+2|M|$ components $C$ marked in red such that $P(C,I,I') \le P(C_{\rm out},I,I')$.
        Since $|\mathcal{C}_1|+|\mathcal{C}_2|\le 2|\allImp|^2+2|M|$ and $C_{\rm out}$ is in $\mathcal{C}_2$, at least one of these components is not in $\mathcal{C}_1 \cup \mathcal{C}_2$.
        Let us denote one such component~$C_{\rm in}$.
        Note that we have $F(S,C_{\rm in})=\natBeh(C_{\rm in})$.

        Let $F_{I'} \in \allBeh(C_{\rm in},I')$ be such that $\wFn(F_{I'})=\min\{\wFn(F) \mid F \in \allBeh(C_{\rm in},I')\}$, i.e., $P(C_{\rm in},I,I')=\wFn(F_{I'})-\wFn(\natBeh(C_{\rm in}))$.
        We obtain $G'_S$ by removing from $G_S$ the edges of $F(S,C_{\rm out})$ and $\natBeh(C_{\rm in})$ and replace them with edges of $\natBeh(C_{\rm out})$ and $F_{I'}$, respectively.
        Since $\wFn(F(S,C_{\rm out}))-\wFn(\natBeh(C_{\rm out})) \ge P(C_{\rm out},I,I') \ge P(C_{\rm in},I,I')=\wFn(F_{I'})-\wFn(\natBeh(C_{\rm in}))$, the total weight of $G'_S$ is at most that of $G_S$.
        We claim that $G'_S$ is also Eulerian.

        Let $I=(T,D)$ and $I'=(T',D')$.
        The degree of each vertex in $M$ is odd with respect to $F(S,C_{\rm out})$ if and only if it is odd with respect to $D'$ which is if and only if it is odd with respect to $F_{I'}$.
        Similarly, the degree of each vertex in $M$ is odd with respect to $\natBeh(C_{\rm in})$ if and only if it is odd with respect to $D$ which is if and only if it is odd with respect to $\natBeh(C_{\rm out})$.
        As the degrees are even in $G_S$, they are also even in $G'_S$.
        Degree of each vertex in $V(R) \setminus (V(C_{\rm in}) \cup V(C_{\rm out}))$ is unchanged between $G_S$ and $G'_S$.
        The degree of vertices in $V(C_{\rm in})$ and $V(C_{\rm out})$ is even by \Cref{def:behavior}, since they are only incident with edges of $\natBeh(C_{\rm in})$ and $F_{I'}$, respectively.
        Hence, all the degrees in $G'_S$ are even.

        Suppose that $G'_S$ is disconnected, i.e., there are vertices $u$ and $v$ such that there is no path between $u$ and $v$ in $G'_S$.
        If $u$ is in a component $C$ of $R$, then, since $G'_S$ contains a behavior of $C$ (either $F(S,C)$, $\natBeh(C_{\rm out})$, or $F_{I'}$), there is a vertex $u' \in M$ in the same connected component of $G'_S$ as $u$.
        Similarly for $v$.
        Hence there are vertices $u',v' \in M$ such that there is no path between $u'$ and $v'$ in $G'_S$.
        In $G_S$ there was a walk between $u'$ and $v'$ formed by the segments of $S$.
        If none of the segments is in $\mathcal{F}(S,C_{\rm out}) \cup \mathcal{F}(S,C_{\rm in})$, then the walk is still present in $G'_S$.
        Let $S_i$ be a segment in $\mathcal{F}(S,C_{\rm out})$ which starts in $x \in M$ and ends in $y \in M$.
        Then $x$ and $y$ lie in the same component of $F(S,C_{\rm out})$, hence in the same component of $D'$ and also in the same component of $F_{I'}$.
        Similarly let $S_i$ be a segment in $\mathcal{F}(S,C_{\rm in})$ which starts in $x \in M$ and ends in $y \in M$.
        Then $x$ and $y$ lie in the same component of $\natBeh(C_{\rm in})$, hence in the same component of $D$ and also in the same component of $\natBeh(C_{\rm out})$.
        Therefore for each such segment there is a walk from $x$ to $y$ in $G'_S$.
        Thus we can replace any missing segment by a walk in $G'_S$, obtaining a walk between $u'$ and $v'$ in $G'_S$---a contradiction.
        Thus $G'_S$ is indeed Eulerian, yielding a solution $S'$.

        Let $\mathcal{S}'$ be the set of segments of $S'$ that are not part of any $\mathcal{F}(S',C)$.
        Let $\mathcal{C}'_1 = \{ C \mid \exists S_i \in \mathcal{S}' \colon S_i \subseteq E[G_C]\}$.
        Let $\mathcal{C}'_2$ be the set of components $C$ of $R$ not in $\mathcal{C}'_1$ with $F(S',C) \neq \natBeh(C)$.
        We claim that $\mathcal{C}'_1 \cup \mathcal{C}'_2 \subseteq (\mathcal{C}_1 \cup \mathcal{C}_2 \cup \{C_{\rm in}\}) \setminus \{C_{\rm out}\}$.
        To this end, let $C$ be a component of $R$ not in $(\mathcal{C}_1 \cup \mathcal{C}_2 \cup \{C_{\rm in}\}) \setminus \{C_{\rm out}\}$.
        Then $C$ is only incident with edges of $\natBeh(C)$ in $G'_S$.
        Thus, if $F(S',C) \neq \natBeh(C)$, then we have $F(S',C) \subsetneq \natBeh(C)$ which contradicts the choice of $\natBeh(C)$.
        Hence, $F(S',C) = \natBeh(C)$ and $C$ is incident with no other edges in $G'_S$, i.e., $C$ is neither in $\mathcal{C}'_1$ not in $\mathcal{C}'_2$.
        Therefore $\mathcal{C}'_1 \cup \mathcal{C}'_2 \subseteq (\mathcal{C}_1 \cup \mathcal{C}_2 \cup \{C_{\rm in}\}) \setminus \{C_{\rm out}\}$, in particular $C_{\rm out}$ is not in $\mathcal{C}'_2$ and, hence, $S'$ has less components in $\mathcal{C}_2$ that are not marked red, contradicting the choice of $S$.
    \end{claimproof}

    As all components in $\mathcal{C}_1$ are marked in blue and all components in $\mathcal{C}_2$ are marked in red, it follows that for all components $C$ marked in green or unmarked we have $F(S,C)=\natBeh(C)$ and $S$ contains no other edges incident with $C$.

    Let $\widehat{G}_S$ be obtained from $G_S$ by removing the edges of $\natBeh(C)$ and vertices of $V(C)$ for all $C$ in $\mathcal{C}^-$.
    This reduces the total weight by exactly $\sum_{C \in \mathcal{C}^-} \wFn(\natBeh(C))$, hence, as $G_S$ is of weight at most $\budget$, $\widehat{G}_S$ is of weight at most $\widehat{\budget}$.
    It remains to show that $\widehat{G}_S$ is Eulerian.

    For each $I \in \allImp$ there is an even number of components $C$ in $\mathcal{C}^-$ such that $\Imp(\natBeh(C))=I$.
    If $I=(T,D)$, then we have for each such $C$ that the degree of each vertex in $M$ is odd in $(M \cup C, \natBeh(C))$ if and only if it is odd in~$D$.
    Hence the degree of each vertex in $M$ changes by an even number by the removal of edges of $\natBeh(C)$ for all such components.
    Moreover, the degree of each vertex in $V(\widehat{G}_S) \setminus M$ stays unchanged.
    As the degree of each vertex is even in $G_S$, it follows that the degree of each vertex is even in $\widehat{G}_S$.

    Suppose that $\widehat{G}_S$ is disconnected, i.e., there are vertices $u$ and $v$ such that there is no path between $u$ and $v$ in $G'_S$.
    If $u$ is in a component $C$ of $\widehat{G} \setminus M$, then, since $\widehat{G}_S$ contains behavior $F(S,C)$ of $C$, there is a vertex $u' \in M$ in the same connected component of $\widehat{G}_S$ as~$u$.
    Similarly for $v$.
    Hence there are vertices $u',v' \in M$ such that there is no path between $u'$ and $v'$ in $\widehat{G}_S$.
    In $G_S$ there was a walk between $u'$ and $v'$ formed by the segments of $S$.
    If none of the segments is in $\mathcal{F}(S,C)$, for some $C \in \mathcal{C}^-$, then the walk is still present in~$\widehat{G}_S$.
    Let $S_i$ be a segment in $\mathcal{F}(S,C)$, for some $C \in \mathcal{C}^-$, which starts in $x \in M$ and ends in $y \in M$.
    Since $C$ is unmarked, there is at least one component $C_{\rm green}$ marked in green such that $\Imp(\natBeh(C_{\rm green}))=\Imp(\natBeh(C))$.
    Since $x$ and $y$ lie in the same component of $\natBeh(C)$, they also lie in the same component of $\natBeh(C_{\rm green})$.
    Since $V(C_{\rm green}) \subseteq V(\widehat{G}_S)$, there is a walk from $x$ to $y$ in $\widehat{G}_S$ for each such segment.
    Thus we can replace any missing segment by a walk in $\widehat{G}_S$, obtaining a walk between $u'$ and $v'$ in $\widehat{G}_S$---a contradiction.

    Thus $\widehat{G}_S$ is indeed Eulerian, yielding a solution $\widehat{S}$ for $(\widehat{G}, \widehat{w}, \widehat{\budget})$.
    This completes the proof.
\end{proof}

To prove \Cref{thm:WRP:modulatorToConstComponentsKernel} we estimate the number of vertices and edges in the reduced instance; \Cref{thm:WRP:modulatorToConstComponentsKernel} then follows using \Cref{lem:magic_kernel}.

\begin{lemma}\label{lem:TSP:modulatorToConstComponents:kernelBound}
    After \Cref{rul:component_rule} has been applied, the number of components is bounded by $2(|\allImp|^{2}+2|M|)|\allImp|^{2}+\binom{|M|}{2}+2|\allImp| = k^{\O(r)}$.
    Hence, the number of vertices and edges in the reduced graph~$G$ is $k^{\O(r)}$.
\end{lemma}
\begin{proof}
    As \Cref{rul:component_rule} marks some of the components and removes the unmarked ones, we just need to count the total number of marked components.
    The rule marks at most $2|\allImp|^{2}+2|M|$ components in red for every $I, I' \in \allImp$.
    It marks at most one component in blue for every pair of vertices $u,v \in M$.
    Finally, it marks at most two components in green for every $I \in \allImp$.

    Hence, there are at most $(2|\allImp|^2+2|M|)|\allImp|^2$ red, $\binom{|M|}{2}$ blue, and $2|\allImp|$ green components.
    By \Cref{obs:num_impacts}, this is $k^{\O(r)}$ components in total.
\end{proof}

We are now ready to prove \Cref{thm:WRP:modulatorToConstComponentsKernel} which we repeat here for reader's convenience.
\thmWRPmodulatorToConstComponentsKernel*
\begin{proof}
The size of the kernel follows from \Cref{lem:TSP:modulatorToConstComponents:kernelBound}. It remains to show that \Cref{rul:component_rule} can be applied in polynomial time. As $G$ contains at most $n(k+r)$ edges, the connected components of $R$ can be found in $\O(n(k+r))$ time.

Let $C$ be a component of $R$.
We need to find $\natBeh(C)$ and for each $I' \in \allImp$ we need $\min\{\wFn(F) \mid F \in \allBeh(C,I')\}$ in order to determine $P(C,I,I')$ for each $I,I' \in \allImp$.
We obtain both by iterating over all behaviors $F$ of $C$.
Namely, there are at most $(kr+1)^{2r}$ options for the at most $2r$ edges between $C$ and $M$ contained in $F$ and $3^{\binom{r}{2}}$ options for the multiplicity of edges within $C$ contained in $F$. For each combination of these options we check in $\O(r^2)$ time whether the resulting $F$ is a behavior of $C$ and if so, determine its impact. We store for each $I' \in \allImp$ the minimum cost behavior with this impact. As $|\allImp| = k^{\O(r)}$ by \Cref{obs:num_impacts}, $P(C,I,I')$ for each $I,I' \in \allImp$ can be computed in $k^{\O(r)}2^{\O(r^2)}$ time.

As there are at most $n$ components, in
\[
	\O\Big(k^{\O(r)}2^{\O(r^2)} \cdot n+ |\allImp|^2\cdot n \cdot(2|\allImp|^2+2|M|) \Big)
	=
	k^{\O(r)}2^{\O(r^2)} \cdot n
\]
time we can find the components to be marked in red.
As there are at most $kr+r^2$ edges in each $G_C$, for each $u,v \in M$ we can compute the shortest path between $u$ and $v$ in $G_C$ in $\O((kr+r^2)^2)$ time using Dijkstra's algorithm.
Hence, the components to be marked in blue can be found in $\O(nk^2 \cdot (kr+r^2)^2) \subseteq k^{\O(r)}2^{\O(r^2)} \cdot n$ time.
Finally, the components to be marked green can be also found in $k^{\O(r)}2^{\O(r^2)} \cdot n$ time.

After finishing the marking, the unmarked components can be removed in $\O(n)$ time (as $\wFn(\natBeh(C))$ was already computed previously). It remains to use \Cref{lem:magic_kernel} and the theorem follows.
\end{proof}

\subsection{\sTSP and the Distance to Constant Size Paths}
\label{sec:stsp_paths}

We start by applying \Cref{rrule:short_circuit} to all vertices of $R$.
Note that after each application of the reduction rule $R$ remains a disjoint union of paths, because each vertex has degree at most 2 within $R$.
Hence, for the rest of the section we assume that $V(R) \subseteq \WP$.

We reuse the notions of the (natural) behavior, impact, and price (\Cref{def:behavior,def:impact,def:impact_change_price}) from the previous section.

We define \emph{piece} of $F \in B(C)$ as any connected component of $(C \cup M, F) \setminus M$ to which we add all incident edges in $G_S$.
We define \emph{legs} of a piece as a subset of its edges which are incident with vertices of $M$.
As $G \setminus M$ consists of disjoint union of paths we note that each piece of $F$ consists of a path on $C$ and its legs.

\begin{lemma}\label{lem:piece_two_legs}
    Given a \sTSP instance $(G,W,\omega,\budget)$ let $S$ be a nice optimal solution to the instance.
    Let $k$ be the size of the modulator to disjoint union of paths.
    There are at most $k$ components with pieces of $F(S,C)$ with more than $2$ legs in behaviors $\{F(S,C) \mid \Imp(F(S,C))=I\}$ for any fixed $I \in \mathcal{I}$.
\end{lemma}
\begin{proof}
    Let $\mathcal{C}_I = \{C \mid \Imp(F(S,C))=I\}$ for some $I \in \mathcal{I}$.
    Note that a behavior $F(S,C)$ may be composed of one or more pieces even though $D(F(S,C))$ is connected (see \Cref{fig:few_behaviors}).
    We will show if there are many pieces with more than $2$ legs, than there are several cases where some piece can be altered to achieve lower cost, which is in contradiction with the optimality of the solution.

    If $|\mathcal{C}_I| \ge 2$ and there is $C \in \mathcal{C}_I$ with an edge $\{u,v\}$ that is traversed exactly twice by $F(S,C)$ and there are two paths through $F(S,C) \setminus \{u,v\}$: one from $u$ to $M$ and one from $v$ to $M$, then we can remove two traversals of $\{u,v\}$ from $S$.
    The change decreases degree of $u$ and $v$ in the solution by two so it still remains even and the connectedness is preserved because both $u$ and $v$ are still connected to $M$, and by existence of a second component $C'$ in $\mathcal{C}_I$ with the same impact of $F(S,C')$ as $F(S,C)$.
    Removing two traversals of $\{u,v\}$ makes the optimal solution smaller, a contradiction.
    See \Cref{fig:few_behaviors} for an example of removal of two $\{u,v\}$ edge traversals.

    Assume $|\mathcal{C}_I| \ge 1$.
    Let us fix one of these components and name it $C_0^I \in \mathcal{C}_I$.
    For a component $C \in \mathcal{C}_I \setminus \{C_0^I\}$ we fix an ordering of its vertices $v_1,\dots,v_a$ in natural order of the path from one end to the other, so that we can compare vertices with respect to this ordering.
    Let us fix a piece $U$ of $F(S,C)$ in $C$.
    Let $Q$ denote the subset of vertices of $V(U) \cap V(C)$ which have a neighbor in $M$.
    Let $s$ and $t$ denote the first and the last element of $Q$ with respect to the fixed ordering.
    Let $P = (s=p_1,p_2,\dots,p_\ell=t)$ denote the path over $C$, see \Cref{fig:few_behaviors}.
    Expressed in this notation, in the previous paragraph we showed that each edge of $P$ is traversed exactly once.

    \begin{figure}[h]
        \centering
        \includegraphics{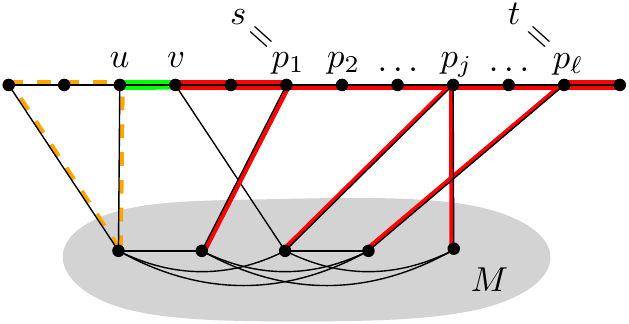}
        \caption{
            A behavior $F(S,C)$ of a solution $S$ over a component $C$.
            The green edge $\{u,v\}$ would be removed from $S$ by the first step of the proof, leaving orange (left, dashed) and red (right, solid) pieces.
            The red piece displays its marked path $P$ by vertices $p_1,\dots,p_n$ between its extremal vertices connected to $M$.
        }%
        \label{fig:few_behaviors}
    \end{figure}

    \newcommand{\moreleggedpieces}{\ensuremath{\operatorname{mlp}}} %

    We shall see that piece $U$ contains a vertex incident with at least two of its legs.
    In the case that $p_1=p_\ell$ we see that this vertex by itself is incident with more than $2$ legs as no other vertex of the piece is incident with legs.
    For $p_1 \ne p_\ell$, as every vertex has even degree in the solution the number of legs incident with the vertices of $P$ needs to be odd for endpoints $p_1$ and $p_\ell$ and it needs to be even (possibly $0$) for all other vertices of the path.
    The only vertices which may have exactly one connection are $p_1$ and $p_\ell$.
    It follows, that there is a vertex $p_j \in P$ which is incident with at least two legs.
    Let us denote the set of these two legs by $\mathcal{L}(U)$.
    In the first case, if the legs in $\mathcal{L}(U)$ are both incident with the same vertex in $M$, then we can remove $\mathcal{L}(U)$ from $S$.
    In that case, we decreased degrees of respective vertices by two, $P$ is still connected to $M$ from $p_1$ and $p_\ell$, and connectedness of $M$ did not change due to the existence of $C_0^I$, a contradiction with minimality of $S$.
    In the other case, let us denote by $L(U)$ the set of two vertices of $M$ which are incident with $\mathcal{L}(U)$, i.e., two of the vertices to which the legs from $p_j$ lead to.

    Let $\moreleggedpieces(F)$ denote the set of pieces in behavior $F$ with more than two legs.
    Now, we build an auxiliary graph $G'$ with vertex set $M$ and initially no edges.
    For each $I \in \mathcal{I}$ and each component $C \in \mathcal{C}_I \setminus \{C_0^I\}$ if $\moreleggedpieces(F(S,C))$ is not empty, then we add an edge between the two vertices of $L(U)$ to $G'$ for every $U \in \moreleggedpieces(F(S,C))$.
    As so, an edge $\{u,v\} \in E(G')$ corresponds to a path from $u$ to $v$ over a middle vertex $x$ in $C$, i.e., $(u,x,v)$ where $\{u,v\} = L(U)$ and $x \in V(C)$.
    Each edge in $G'$ is created by a different piece with more than two legs so all of these edges correspond to paths in $G_S$ which have disjoint middle vertices.
    Therefore, any cycle in $G'$ corresponds to a cycle in $G_S$ consisting of legs.
    If $G'$ has more than $k-1$ edges, then it contains a cycle which corresponds to a cycle in $G_S$.
    Removing such cycle from $G_S$ does not break connectivity due to the existence of $C_0^I$ for every piece $U \in \moreleggedpieces(F(S,C))$ as $C \in \mathcal{C}_I$.
    This is a contradiction with the minimality of $S$.

    We first showed, that for each $I \in \mathcal{I}$ if $|\mathcal{C}_I| \ge 2$, the edges of every $C \in \mathcal{C}_I$ between connections to $M$ are traversed at most once by $S$.
    Then we showed that $\sum_{I \in \mathcal{I}}\sum_{C \in \mathcal{C}_I}|\moreleggedpieces(F(S,C))| \le k-1$.
    Hence, for any fixed $I \in \mathcal{I}$ there may be at most $k$ components $C \in \mathcal{C}_I$ with pieces of $F(S,C)$ with more than two legs within $\mathcal{C}_I$.
\end{proof}

It is the case that every natural behavior has only two-legged pieces.

\begin{observation}\label{obs:nat_2_legged}
 For each $C$ component of~$R$, each piece of $\natBeh(C)$ has two legs.
\end{observation}
\begin{proof}
    We prove this observation by a contradiction.
    Assume that for some $C \in R$ there is a piece $U$ in $\natBeh(C)$ which has more than two legs.
    We now follow similar argument as in the proof of \Cref{lem:piece_two_legs}.

    For the component $C$ we fix an ordering of its vertices $v_1,\dots,v_a$ in natural order of the path from one end to the other, so that we can compare vertices with respect to this ordering.
    Let $Q$ denote the subset of vertices of $V(U) \cap V(C)$ which have a neighbor in $M$.
    Let $s$ and $t$ denote the first and the last element of $Q$ with respect to the fixed ordering.
    Let $P = (s=p_1,p_2,\dots,p_\ell=t)$ denote the path over $C$, see \Cref{fig:few_behaviors}.

    If $p_1=p_\ell$ and the number of incident legs is more than two, then we may simply remove two legs incident with $p_1$ -- a contradiction with minimality of $\natBeh(C)$.
    Assume $p_1 \ne p_\ell$.

    Assume that an edge $\{p_j,p_{j+1}\}$ is twice in $\natBeh(C)$.
    Then, removing the two traversals of $\{p_j,p_{j+1}\}$ from $\natBeh(C)$ still yields a behavior (has even degrees and each part is connected to $M$) and is cheaper -- reaching a contradiction again.

    As every edge $\{p_j,p_{j+1}\}$ is traversed once we have that the only vertices with odd degree in $\natBeh(C) \cap C$ are $p_1$ and $p_\ell$.
    To reach even degree in $\natBeh(C)$ these must have odd number of incident legs and all $p_j$ for $1 < j < \ell$ have even number of incident legs.
    If follows that there must be a vertex $p_j$ for $1 \le j \le \ell$ that is incident with at least two legs.
    Removing these two legs changes the degree of $p_j$ by two and does not disconnect it from $M$ as both $p_1$ and $p_\ell$ are still connected; a final contradiction which proves the observation.
\end{proof}

The most technical part of this section is the following lemma that allows us to, under some conditions, mark a non-terminal in the modulator as a terminal. %
For this to work, we want to do the following.
Take many components which share the impact of the natural behavior which touches the particular non-terminal.
Since there are many, many of them also share the impact of the actual behavior and have all pieces 2-legged.
For each of them we want to find a behavior that is half-way between the actual and the natural behavior (using the following lemma).
This behavior should touch the non-terminal and be of at most the same weight as the actual behavior.
Then we find two components for which the half-way behavior has the same impact and change these, so that we obtain a solution that visits the particular non-terminal.

\begin{lemma}[Blending lemma]\label{lem:path_mid_behavior}
	Let $M' \subseteq M$ (the set of actually visited vertices) and $C$ a component of~$R$ (therefore a path) such that $T(\natBeh(C)) \nsubseteq M'$. Let $v \in T(\natBeh(C)) \setminus M'$ and $A \in B(C)$ (the actual behavior) such that $T(A) \subseteq M'$ and such that each piece has two legs. Then there is a behavior $F \in B(C)$ such that $v \in T(F)$, $T(F) \subseteq T(A) \cup T(\natBeh(C))$, every connected component of $(C \cup M,  F)$ contains a vertex of $M'$, and $\wFn(F) \le \wFn(A)$.
\end{lemma}
\begin{proof}
	We construct a behavior $F$ for every actual behavior $A$ by modifying pieces of $A$ with respect to vertex $v$. Let $u$ be an arbitrary neighbor of $v$ in $(C \cup M, \natBeh(C))$.
	First, if $u$ has a neighbor in $M$ in graph $(C \cup M, A)$, say $w$, then we claim that the behavior $F$ obtained from $A$ by removing one occurrence of edge $\{u,w\}$ and adding one occurrence of edge $\{u,v\}$ is sought behavior: Indeed, assume that $\wFn(\{u,w\}) < \wFn(\{u,v\})$. Then by replacing $\{u,v\}$ with $\{u,w\}$ in the natural behavior $\natBeh(C)$ we obtain a behavior of a lower weight, which contradicts the minimality of $\natBeh(C)$. Hence $\wFn(\{u,v\}) \leq \wFn(\{u,w\})$ and $F$ is indeed the sought behavior.

	Hence, for the rest of the proof we assume that the vertex $u$ has no neighbor in $M$ in graph $(C \cup M, A)$.
	We may assume, due to \Cref{lem:piece_two_legs}, that a piece $Y$ of $A$ containing $u$ has two legs $\{w_\ell,u_\ell\}$ and $\{w_r,u_r\}$, where $w_\ell,w_r\in M'$ and $u_\ell,u_r\in (V(C)\cap V(Y)) \setminus \{u\}$. For the rest of the proof, we assume, without loss of generality, that $\{w_\ell,u_\ell\}$ is the lefter and $\{w_r,u_r\}$ is the righter leg of $Y$.
	See \Cref{fig:subtsp_dcsp_nat_act_behav} for a schematical presentation of the possible cases (pieces $Y$ of the actual behavior $A$ are in \textcolor{red}{red}).

	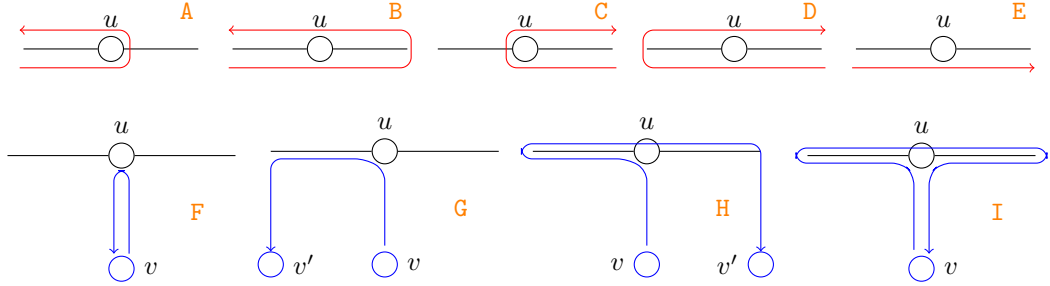
\begin{figure}[h!]
		\tikzstyle{snake} = [decorate,decoration={snake,amplitude=.4mm,segment length=1.5mm}]
		\begin{center}
	        \begin{tikzpicture}
				\node[draw,circle,black,label={$u$}] (v) at (0,0) {};
				\draw (v) edge (-1.15,0) edge (1.15,0);
				\draw[rounded corners,red,->] (-1.2,-0.25) -- (0.25,-0.25) -- (0.25,0.25) -- (-1.2,0.25);
				\node[orange] (A) at (1,0.5) {$\mathtt{A}$};
			\end{tikzpicture}
			\hspace*{0.1cm}
			\begin{tikzpicture}
				\node[draw,circle,black,label={$u$}] (v) at (0,0) {};
				\draw (v) edge (-1.15,0) edge (1.15,0);
				\draw[rounded corners,red,->] (-1.2,-0.25) -- (1.2,-0.25) -- (1.2,0.25) -- (-1.2,0.25);
				\node[orange] (A) at (1,0.5) {$\mathtt{B}$};
			\end{tikzpicture}
			\hspace*{0.1cm}
			\begin{tikzpicture}
				\node[draw,circle,black,label={$u$}] (v) at (0,0) {};
				\draw (v) edge (-1.15,0) edge (1.15,0);
				\draw[rounded corners,red,->] (1.2,-0.25) -- (-0.25,-0.25) -- (-0.25,0.25) -- (1.2,0.25);
				\node[orange] (A) at (1,0.5) {$\mathtt{C}$};
			\end{tikzpicture}
			\hspace*{0.1cm}
			\begin{tikzpicture}
				\node[draw,circle,black,label={$u$}] (v) at (0,0) {};
				\draw (v) edge (1.15,0) edge (-1.15,0);
				\draw[rounded corners,red,->] (1.2,-0.25) -- (-1.2,-0.25) -- (-1.2,0.25) -- (1.2,0.25);
				\node[orange] (A) at (1,0.5) {$\mathtt{D}$};
			\end{tikzpicture}
			\hspace*{0.1cm}
			\begin{tikzpicture}
				\node[draw,circle,black,label={$u$}] (v) at (0,0) {};
				\draw (v) edge (-1.15,0) edge (1.15,0);
				\draw[red,->] (-1.2,-0.25) -- (1.2,-0.25);
				\node[orange] (A) at (1,0.5) {$\mathtt{E}$};
			\end{tikzpicture}

			\vspace*{0.5cm}

			\begin{tikzpicture}
				\node[draw,circle,black,label={$u$}] (v) at (0,0) {};
				\draw (v) edge (-1.5,0) edge (1.5,0);
				\node[draw,circle,blue,label={0:$v$}] (v) at (0,-1.5) {};
				\draw[blue,->,rounded corners] (0.1,-1.3) -- (0.1,-0.2) -- (-0.1,-0.2) -- (-0.1,-1.3);
				\node[orange] (A) at (1,-0.75) {$\mathtt{F}$};
			\end{tikzpicture}
			\hspace*{0.05cm}
			\begin{tikzpicture}
				\node[draw,circle,black,label={$u$}] (v) at (0,0) {};
				\draw (v) edge (-1.5,0) edge (1.5,0);
				\node[draw,circle,blue,label={0:$v$}] (v) at (0,-1.5) {};
				\node[draw,circle,blue,label={0:$v'$}] (v') at (-1.5,-1.5) {};
				\draw[blue,->,rounded corners] (0, -1.25) -- (0,-0.25) -- (-0.25,-0.1) -- (-1.5,-0.1) -- (v');
				\node[orange] (A) at (1,-0.75) {$\mathtt{G}$};
			\end{tikzpicture}
			\hspace*{0.05cm}
			\begin{tikzpicture}
				\node[draw,circle,black,label={$u$}] (v) at (0,0) {};
				\draw (v) edge (-1.5,0) edge (1.5,0);
				\node[draw,circle,blue,label={180:$v$}] (v) at (0,-1.5) {};
				\node[draw,circle,blue,label={180:$v'$}] (v') at (1.5,-1.5) {};
				\draw[blue,->,rounded corners] (0,-1.25) -- (0,-0.25) -- (-0.25,-0.1) -- (-1.65,-0.1) -- (-1.65,0.1) [->]-- (1.5,0.1) -- (v');
				\node[orange] (A) at (1,-0.75) {$\mathtt{H}$};
			\end{tikzpicture}
			\hspace*{0.05cm}
			\begin{tikzpicture}
				\node[draw,circle,black,label={$u$}] (v) at (0,0) {};
				\draw (v) edge (-1.5,0) edge (1.5,0);
				\node[draw,circle,blue,label={0:$v$}] (v) at (0,-1.5) {};
				\draw[blue,->,rounded corners] (-0.1,-1.25) -- (-0.1,-0.25) -- (-0.25,-0.1) -- (-1.65,-0.1) -- (-1.65,0.1) -- (1.65,0.1) -- (1.65,-0.1) -- (0.25,-0.1) -- (0.1,-0.25) -- (0.1,-1.3);
				\node[orange] (A) at (1,-0.75) {$\mathtt{I}$};
			\end{tikzpicture}
		\end{center}
		\caption{List of all actual (\textcolor{red}{red}) and natural behaviors (\textcolor{blue}{blue}) relevant for the proof of~\Cref{lem:path_mid_behavior}.}
		\label{fig:subtsp_dcsp_nat_act_behav}
	\end{figure}

	We first deal with the case that $u$ is not part of the $u_\ell$-$u_r$-path in $C$, that is, both $u_\ell$-$u_r$ are left of $u$ or both are right of $u$ (denoted \texttt{A}-\texttt{D} on \Cref{fig:subtsp_dcsp_nat_act_behav}).

	\begin{claim}\label{lem:path_mid_behavior:X:notE}
		Let $\natBeh(C)$ be any natural behavior containing piece $X$ with at least one leg being $\{u,v\}$ and $A$ be an actual behavior containing one piece $Y$ from pieces marked \texttt{A}-\texttt{D} on \Cref{fig:subtsp_dcsp_nat_act_behav}, i.e., $u$ is not part of the $u_\ell$-$u_r$-path in $C$. Then there is a behavior $F$ such that $v \in T(F)$, $T(F) \subseteq T(A) \cup T(\natBeh(C))$, every connected component of $(C \cup M,  F)$ contains a vertex of~$M'$, and $\wFn(F) \le \wFn(A)$.
	\end{claim}
	\begin{claimproof}
		Let $A$ corresponds to actual behavior \texttt{A} or \texttt{B}. We construct a behavior $F$ by replacing the path $P = (w_r,u_r,\ldots,u)$ in $Y$ with the edge $\{u,v\}$. It is easy to see that $v$ is indeed part of $T(F)$ and the only touched vertices are $v$ and $w_\ell$, and both of them are elements of $T(A)\cup T(\natBeh(C))$. Moreover, $w_\ell\in M'$ and thus we satisfy the condition that every connected component of $(C\cup M, F)$ contains a vertex of $M'$. What remains to prove is that the weight of the behavior $F$ is not greater than the weight of the behavior $A$.

		Assume that $\wFn(\{u,v\})$ is strictly greater than the weight of the replaced path $P$. If we replace the edge $\{u,v\}$ with $P$ in the natural behavior $\natBeh(C)$, then we obtain a behavior with a lower weight, which contradicts the minimality of $\natBeh(C)$. Hence, $\wFn(\{u,v\})$ is at most $\wFn(P)$ and thus $\wFn(F) \leq \wFn(A)$.

		If $A$ corresponds to actual behavior \texttt{C} or \texttt{D}, then we construct a behavior $F$ by replacing the path $P = (w_\ell,u_\ell,\ldots,u)$ in $Y$ with the edge $\{u,v\}$. By the symmetric arguments as in the previous case, we obtain that $F$ satisfies all the conditions and it holds that $\wFn(F) \leq \wFn(A)$. This concludes the proof.
	\end{claimproof}

	For the rest of the proof we assume that $u$ is a part of the $u_\ell$-$u_r$-path in $C$, that is, the piece $Y$ is as denoted \texttt{E} on  \Cref{fig:subtsp_dcsp_nat_act_behav}. In particular, this implies that $u$ has two neighbors in $C$.
	Now we distinguish two cases for the piece $X$ of $\natBeh(C)$ containing $u$ according to whether it contains both the edges of $C$ incident with $u$ (at least once) or not (again \Cref{fig:subtsp_dcsp_nat_act_behav} contains a schematic depiction of some of the possible cases in \textcolor{blue}{blue}).

	\begin{claim}\label{lem:path_mid_behavior:FG:E}
		Let the natural behavior $\natBeh(C)$ contain piece $X$ that contains at most one neighbor of $u$ in $C$ (marked \texttt{F} or \texttt{G} on \Cref{fig:subtsp_dcsp_nat_act_behav}) and $A$ be an actual behavior containing piece $Y$ such that $u$ is a part of the $u_\ell$-$u_r$-path in $C$ (marked \texttt{E} in the same figure). Then there is a behavior $F$ such that $v \in T(F)$, $T(F) \subseteq T(A) \cup T(\natBeh(C))$, every connected component of $(C \cup M,  F)$ contains a vertex of $M'$, and $\wFn(F) \le \wFn(A)$.
	\end{claim}
	\begin{claimproof}
		Let $u'\in V(C)\cap N(u)$ be the neighbor of the vertex $u$ such that $u'$ is not part of the assumed piece of the natural behavior. We construct the desired behavior $F$ by replacing the edge $\{u',u\}$ with the edge $\{u,v\}$ in the actual behavior $A$. This replacement splits the piece into two disjoint pieces and leaves one of them one-legged. Since $u'$ is a terminal, there is a piece in the natural behavior $\natBeh(C)$ visiting $u'$. Thus, we complete $F$ by adding shortest path $P$ from $u'$ to any vertex of the modulator while using only edges that are part of the piece the vertex $u'$ is part of in $\natBeh(C)$.

		The vertex $v$ is clearly part of $T(F)$ and it holds that $T(F)\subseteq T(A)\cup T(\natBeh(C))$. We replaced one piece with two different pieces. This first constructed piece contains $w_\ell$ and the second contains $w_r$. Both of these vertices are part of $M'$ and we do not change remaining pieces of $A$, hence, every connected component of $(C\cup M,F)$ contains a vertex of $M'$.
		As $F$ does not miss any vertex of $C$, it is indeed a behavior.
		It remains to show that the weight of behavior $F$ is not greater than the weight of behavior $A$.

		Assume that $\wFn(\{u',u\}) < \wFn(\{v,u\}) + \wFn(P)$.
		We claim that, if we replace $\{u,v\}$ and $P$ with $\{u',u\}$ in $\natBeh(C)$, then we obtain a behavior with a lower weight, which contradicts the minimality of $\natBeh(C)$.
		Indeed, this edge multiset uses each edge at most twice, each connected component contains a vertex of $M$, whereas there are at most $2r$ edges incident with $M$ in total.
		If $u''$ is the vertex with a leg closest to $u'$ in the same piece of $\natBeh(C)$, then the subpath between $u'$ and $u''$ together with the leg forms $P$. However, this means that all the edges between $u'$ and $u''$ are doubled in $\natBeh(C)$ and remain at least once in the modified edge multiset. Hence the modified set does not miss any terminal and it is indeed a behavior, contradicting the minimality of $\natBeh(C)$.
		Hence, $\wFn(\{u',u\}) \geq \wFn(\{v,u\}) + \wFn(P)$ and thus $\wFn(F) \leq \wFn(A)$.
	\end{claimproof}

	It remains to deal with the case that the piece $X$ of $\natBeh(C)$ contains both neighbors of $u$ in $C$. Note that by \Cref{obs:nat_2_legged} we know that $X$ has two legs. Since the situation is symmetric for the piece $Y$, without loss of generality we can assume that the other leg of $X$ is right of $u$ (situations denoted \texttt{H} and \texttt{I} on \Cref{fig:subtsp_dcsp_nat_act_behav}).

	\begin{claim}\label{lem:path_mid_behavior:HI:E}
		Let the natural behavior $\natBeh(C)$ contain piece $X$ which contains $u$ and both of its neighbors in $C$ and $A$ be an actual behavior containing piece $Y$ such that $u$ is a part of the $u_\ell$-$u_r$-path in $C$ (marked \texttt{E} on \Cref{fig:subtsp_dcsp_nat_act_behav}). Then there is a behavior $F$ such that $v \in T(F)$, $T(F) \subseteq T(A) \cup T(\natBeh(C))$, every connected component of $(C \cup M,  F)$ contains a vertex of $M'$, and $\wFn(F) \le \wFn(A)$.
	\end{claim}
	\begin{claimproof}
	    Without loss of generality we assume that the other leg of $X$ is right of $u$, in particular all the edges to the left of $u$ in the same piece of $\natBeh(C)$ are double edges.
		By $\ell$ we denote the leftmost vertex of $X$. We start with the simpler case that $\operatorname{dist}(\ell,u) \geq \operatorname{dist}(u_\ell,u)$. We construct the behavior $F$ by replacing the edge $\{w_\ell,u_\ell\}$ with the path $P_1 = (v,u,\ldots,u_\ell)\subseteq X$ in the behavior $A$. It is easy to see that the first three conditions are indeed satisfied.

		Assume that $\wFn(\{w_\ell,u_\ell\}) < \wFn(P_1)$. Then we claim that by replacing $P_1$ with $\{w_\ell,u_\ell\}$ in $\natBeh(C)$, we obtain a behavior with a lower weight, which contradicts the minimality of $\natBeh(C)$.
		Indeed, this edge multiset uses each edge at most twice, each connected component contains a vertex of $M$, whereas there are at most $2r$ edges incident with $M$ in total. As all the edges to the left of $u$ were doubled in $\natBeh(C)$, the modified set does not miss any terminal and it is indeed a behavior, contradicting the minimality of $\natBeh(C)$.
		Hence $\wFn(\{w_\ell,u_\ell\}) \geq \wFn(P_1)$ and thus $\wFn(F) \leq \wFn(A)$.

		Now let us turn to the case that $\operatorname{dist}(u_\ell,u) > \operatorname{dist}(\ell,u)$. By $\ell'$ we denote left neighbor of $\ell$ that is not part of $X$, i.e., $\ell'\in (V(C) \cap N(\ell)) \setminus V(X)$. Since $\ell'$ is a terminal and is not visited by $X$, there is a piece $X'$ such that $X'$ visits $\ell'$. By $P_{\ell'}$ we denote a path with smallest weight connecting $\ell'$ with modulator in piece $X'$ and by $P_{\ell}$ we denote the path $(v,u,\ldots,\ell)$. We construct the behavior $F$ by removing the edge $\{\ell',\ell\}$ from $A$ and adding $P_{\ell'}$ and $P_\ell$. It is easy to verify that $v\in T(F)$ and $T(F)\in T(A)\cup T(\natBeh(C))$. It also holds that every connected component of $(C\cup M, F)$ contains vertex of $M'$ as we created two pieces where the first contains $w_\ell$ and the second contains $w_r$ and both vertices were part of $M'$ already in $A$. What remains to show is that $\wFn(F) \leq \wFn(A)$.

		Assume that $\wFn(P_\ell) + \wFn(P_{\ell'}) > \wFn(\{\ell',\ell\})$.
		Then we claim that by exchanging $P_\ell$ and $P_{\ell'}$ for $\{\ell', \ell\}$ in the natural behavior, we obtain a behavior with a lower weight, which is a contradiction with the minimality of $\natBeh(C)$.
		Indeed, this edge multiset uses each edge at most twice, each connected component contains a vertex of $M$, whereas there are at most $2r$ edges incident with $M$ in total.
		If $\ell''$ is the vertex with a leg closest to $\ell'$ in the same piece of $\natBeh(C)$, then the subpath between $\ell'$ and $\ell''$ together with the leg forms $P_{\ell'}$. However, this means that all the edges between $\ell'$ and $\ell''$ are doubled in $\natBeh(C)$ and remain at least once in the modified edge multiset. Hence the modified set does not miss any terminal and it is indeed a behavior, contradicting the minimality of $\natBeh(C)$.
		Therefore $\wFn(P_\ell) + \wFn(P_{\ell'}) \le \wFn(\{\ell',\ell\})$ and we get that $\wFn(F) \leq \wFn(A)$.
	\end{claimproof}

	Using the constructions described in \Cref{lem:path_mid_behavior:FG:E,lem:path_mid_behavior:HI:E} we finished the description of behavior $F$ for all remaining cases. Hence, the proof is complete and the lemma holds.
\end{proof}

Now, we are ready to present the reduction rule used in this case.
The colors have the same meaning as in \Cref{rul:component_rule}. 

\begin{rrule}\label{rul:subset_component_rule}
 For each pair $I, I' \in \allImp$ of impacts if there are at most $2|\allImp|^{2}+2|M|$ components $C$ with finite $P(C,I,I')$, then mark all of them in red.
 Otherwise mark $2|\allImp|^{2}+2|M|$ components $C$ with the least $P(C,I,I')$ in red.

 For each pair of vertices $u,v \in M$, if there is a component $C$ in $R$ such that $G_C$ contains a $u$-$v$-path, then mark the component which contains the shortest such path in blue.

 For each $I \in \allImp$, if there are unmarked components $C$ with $\Imp(\natBeh(C))=I$, then let $I=(T,D)$ and do the following.
 If $T \subseteq \WP$, then if the number of such components is odd, then mark one arbitrary such component in green.
 If the number of such components is even, then mark two arbitrary such components in green.

 If $T \nsubseteq \WP$ and there are at most $\big((r+1)^{4r}\cdot 2^{4r+1} + k\big)\cdot |\allImp|$ unmarked components $C$ with $\Imp(\natBeh(C))=I$, then mark them all in yellow.

 If $T \nsubseteq \WP$ and there are more than $\big((r+1)^{4r}\cdot 2^{4r+1} + k\big)\cdot |\allImp|$ unmarked components $C$ with $\Imp(\natBeh(C))=I$, then add $T$ to $\WP$.

 If $\WP$ was not changed, then for each unmarked component $C$, remove $C$ from $G$ and reduce $\budget$ by $\wFn(\natBeh(C))$.
\end{rrule}
\begin{proof}[Safeness]
    Let $(G,\WP,\wFn,\budget)$ be the original instance and $(\widehat{G},\widehat{\WP},\widehat{\wFn},\widehat{\budget})$ be the new instance resulting from the application of the rule.
    Note that $\widehat{\wFn}$ is just the restriction of $\wFn$ to $\widehat{G}$ which is a subgraph of $G$.
    Let $\mathcal{C}^-$ be the set of components of $C$ removed by the rule.
    Note that for each $I \in \allImp$ there is an even number of components $C \in \mathcal{C}^-$ with $\Imp(\natBeh(C))=I$ and that $\widehat{\budget} = \budget - \sum_{C \in \mathcal{C}^-} \wFn(\natBeh(C))$.
    This safeness proof is similar to the safeness of \Cref{rule:wrp_vc,rul:component_rule}.
    We first show that if the new instance is a \YESi, then so is the original one.

    This direction goes similarly to the safeness of \Cref{rul:component_rule}.
    Let $\widehat{S}$ be a solution walk in the new instance and let $\widehat{G}_S$ be the corresponding multigraph formed by edges of $\widehat{S}$.
    Note that the total weight of $\widehat{G}_S$ is at most $\widehat{\budget}$.
    If $\widehat{\WP} \neq \WP$, then $R^-= \emptyset$, $G=\widehat{G}$, $\widehat{\budget}=\budget$, and $S$ is also a solution for the original instance.
    Otherwise we have for each $C \in \mathcal{C}^-$ that $\Imp(\natBeh(C)) \subseteq \WP$.
    Then we construct a multigraph $G_S$ by adding to $\widehat{G}_S$ for each component $C$ in $\mathcal{C}^-$ the vertex set $V(C)$ together with the edge set $\natBeh(C)$.
    Since each connected component of $(M \cup C, \natBeh(C))$ contains a vertex of $M \cap W$ and $\widehat{G}_S$ is connected, $G_S$ is also connected.
    Furthermore, for each $I \in \allImp$ the number of unmarked components $C$ such that $\Imp(\natBeh(C))=I$ is even.
    Therefore, degrees of vertices in $\widehat{G}_S$ change by an even number and all the added vertices have even degree.
    Hence $G_S$ is Eulerian and contains all vertices of $W$.
    Since the weight of the added edges is exactly $\sum_{C \in \mathcal{C}^-} \wFn(\natBeh(C))$ it follows that $(G,\WP,\wFn,\budget)$ is a \YESi.

    For the other direction, we again follow the safeness of \Cref{rul:component_rule}.
    The main difference is in the case where there are ``many'' natural behaviors that are connected to some vertices of $M \setminus W$.
    To tackle this case, we use the same logic as in the safeness of \Cref{rule:wrp_vc} to see that the solution may be altered slightly (using \Cref{lem:path_mid_behavior}) to contain such vertices.
    Such an alteration of the solution always exists so by adding the non-terminal vertices to $W$ we create an equivalent instance.

    Suppose that the original instance $(G,\WP,\wFn,\budget)$ is a \YESi.

    Suppose first that $\widehat{\WP} \neq \WP$.
    Let $S$ be a solution that traverses all vertices of $\WP$ and the most vertices of $\widehat{\WP} \setminus \WP$ among all solutions.
    Let $G_S$ be the corresponding multigraph formed by edges of $S$.
    If $S$ traverses all the vertices of $\widehat{\WP}$, then it is also a solution for $(\widehat{G},\widehat{\WP},\widehat{\wFn},\widehat{\budget})$.
    Suppose $S$ does not traverse a vertex $v \in \widehat{\WP} \setminus \WP$.
    Since $v$ was added to $\widehat{\WP}$, there is an impact $I=(T,F) \in \allImp$ with $v \in T$ and at least $\big((r+1)^{4r}\cdot 2^{4r+1} + k\big)\cdot |\allImp|$ components $C$ that share impact of their natural behavior $I = \Imp(\natBeh(C))$.
    By the pigeonhole principle there are at least $(r+1)^{4r}\cdot 2^{4r+1} + k$ components $C_i$ that share their actual impact $\Imp(F(S,C_i))$.
    Let $C_I$ be the subset of these components for which all pieces of the actual behavior have two legs.
    By \Cref{lem:piece_two_legs}, at most $k$ pieces of components with the same impact of the actual behavior have more than two legs, hence $|C_I| \ge (r+1)^{4r}\cdot 2^{4r+1}$.
    Now, all components $C_i \in C_I$ have actual behavior $F(S,C_i)$ such that $T(F(S,C_i)) \subseteq W$ and each its piece has two legs, and also a natural behavior $\natBeh(C_i)$ such that $v \in T(\natBeh(C_i))$ while $v \not\in W$.
    Hence, we may apply \Cref{lem:path_mid_behavior} on each $C_i \in C_I$ to obtain a behavior $F_i$.
    We know that $v \in T(F_i)$, $T(F_i) \subseteq T(F(S,C_i)) \cup T(\natBeh(C_i))$, every connected component of $(C_i \cup M, F_i)$ contains a vertex of $M \cap W$, and that $\wFn(F_i) \leq \wFn(F(S,C_i))$.
    For each behavior $F_i$ we have $|T(F_i)| \le 4r$ as $|T(F(S,C_i))| \le 2r$ and $|T(\natBeh(C_i))| \le 2r$.
    Because of our choice of $C_I$ for each $F_i$ impacts $I=\Imp(C_i \cup M, \natBeh(C_i))$ and $\Imp(C_i \cup M, F(S,C_i))$ are fixed.
    The behaviors $F_i$ may have in $(V(C_i) \cup T(F_i), F_i)$ at most $r$ connected components.
    Each vertex in $T(F(S,C_i)) \cup T(\natBeh(C_i))$ may belong to one of these $r$ connected components of $T(F_i)$ or to none.
    Also, it is incident either to odd or even number of edges from $F_i$.
    These two properties fully define possible impacts of the behavior $F_i$.
    Hence, there are at most ${(r+1)}^{4r} \cdot 2^{4r}$ different impacts.
    As we have at least twice as many components in $C_I$ we are guaranteed to have two components $C_a,C_b \in C_I$ which have respective $F_a$ and $F_b$ such that $\Imp(F_a) = \Imp(F_b)$.

    We remove from $G_S$ the behaviors $F(S,C_a)$ and $F(S,C_b)$ (of components $C_a$ and $C_b$) and add vertex $v$ together with the edges of $F_a$ and $F_b$ to obtain $\widehat{G}_S$.
    By \Cref{lem:path_mid_behavior} this does not increase the cost.
    As the two changes to the impacts were identical it follows that the degrees of vertices in $M$ changed by an even amount, the vertices inside of $C_a$ and $C_b$ have even degree in $\widehat{G}_S$ by the construction, and all the other vertex degrees remain unchanged, thus $\widehat{G}_S$ has all degrees even.
    The connectivity did not change because all components of $C_I$ have the same impact as $F(S,C_a)$ and $F(S,C_b)$, while each piece of $F_a$ and $F_b$ is connected to a vertex in $M \cap W$.
    Therefore $\widehat{G}_S$ is Eulerian and the corresponding trail constitutes a solution $\widehat{S}$ that traverses all vertices of $\WP$ and more vertices of $\widehat{\WP} \setminus \WP$ than $S$, contradicting the choice of $S$.
    Hence there is a solution $S$ traversing all vertices of $\widehat{\WP}$ proving the safeness of the rule in the case $\widehat{\WP} \neq \WP$.

    The proof for the case $\widehat{\WP} = \WP$ is exactly the same as the proof of safeness of \Cref{rul:component_rule}, since for each $C$ in $\mathcal{C}^-$ we have $T(\natBeh(C)) \subseteq W$, and, hence we never encounter any non-terminal vertices in course of the proof.
\end{proof}

To prove \Cref{thm:WRP:modulatorToConstPathsKernel} we estimate the number of vertices and edges in the reduced instance; \Cref{thm:WRP:modulatorToConstPathsKernel} then follows using \Cref{lem:magic_kernel}.

\begin{lemma}\label{lem:sTSP:modulatorToConstPaths:kernelBound}
    After \Cref{rul:subset_component_rule} has been applied, the number of components is bounded by $2(|\allImp|^{2}+2|M|)|\allImp|^{2} + \binom{|M|}{2} + 2|\allImp| + \big((r+1)^{4r}\cdot 2^{4r+1} + k\big)\cdot |\allImp|^2 = k^{\O(r)}$.
    Hence, the number of vertices and edges in the reduced graph~$G$ is $k^{\O(r)}$.
\end{lemma}
\begin{proof}
    As \Cref{rul:subset_component_rule} marks some of the components and removes the unmarked ones, we just need to count the total number of marked components.
    The rule marks at most $2|\allImp|^{2}+2|M|$ components in red for every $I, I' \in \allImp$.
    It marks at most one component in blue for every pair of vertices $u,v \in M$.
    It marks at most two components in green for every $I \in \allImp$.
    And finally, it marks at most $\big((r+1)^{4r}\cdot 2^{4r+1} + k\big)\cdot |\allImp|$ components in yellow for every $I \in \allImp$.

    Hence, there are at most $2(|\allImp|^{2}+2|M|)|\allImp|^{2}$ red, $\binom{|M|}{2}$ blue, $2|\allImp|$ green, and $\big((r+1)^{4r}\cdot 2^{4r+1} + k\big)\cdot |\allImp|^2$ yellow components.
    By \Cref{obs:num_impacts}, this is $k^{\O(r)}$ components in total.
\end{proof}

Now, we are ready to prove \Cref{thm:WRP:modulatorToConstPathsKernel} which we repeat for reader's convenience.

\thmWRPmodulatorToConstPathsKernel*
\begin{proof}
    The size of the kernel follows from \Cref{lem:sTSP:modulatorToConstPaths:kernelBound}.
    It remains to show that \Cref{rul:subset_component_rule} can be applied in polynomial time.

    As $G$ contains at most $n(k+r)$ edges, the connected components of $R$ can be found in $\O(n(k+r))$ time.
    Let $C$ be a component of $R$.
    We need to find $\natBeh(C)$ and for each $I' \in \allImp$ we need $\min\{\wFn(F) \mid F \in \allBeh(C,I')\}$ in order to determine $P(C,I,I')$ for each $I,I' \in \allImp$.
    We obtain both by iterating over all behaviors $F$ of $C$.
    Namely, there are at most $(kr+1)^{2r}$ options for the at most $2r$ edges between $C$ and $M$ contained in $F$ and $3^{\binom{r}{2}}$ options for the multiplicity of edges within $C$ contained in $F$.
    For each combination of these options we check in $\O(r^2)$ time whether the resulting $F$ is a behavior of $C$ and if so, determine its impact.
    We store for each $I' \in \allImp$ the minimum cost behavior with this impact.
    As $|\allImp| = k^{\O(r)}$ by \Cref{obs:num_impacts}, $P(C,I,I')$ for each $I,I' \in \allImp$ can be computed in $k^{\O(r)}2^{\O(r^2)}$ time.

    As there are at most $n$ components, in $k^{\O(r)}2^{\O(r^2)} \cdot n$ time we can find the components to be marked in red.
    As there are at most $kr+r^2$ edges in each $G_C$, for each $u,v \in M$ we can compute the shortest path between $u$ and $v$ in $G_C$ in $\O((kr+r^2)^2)$ time using Dijkstra's algorithm.
    Hence, the components to be marked in blue can be found in $\O(nk^2 \cdot (kr+r^2)^2) \subseteq k^{\O(r)}2^{\O(r^2)} \cdot n$ time.
    Then, all the yet unmarked components $C$ are grouped based on their $\natBeh(C)$ which were computed in $k^{\O(r)}2^{\O(r^2)} \cdot n$.
    Each group gets checked if $T(\Imp(\natBeh(C))) \subseteq W$ in $\O(r) \cdot n$ which gives us all the components marked green in $k^{\O(r)}2^{\O(r^2)} \cdot n$ time.
    Each group where $T(\Imp(\natBeh(C))) \not\subseteq W$ gets checked whether it contains
    at most $\big((r+1)^{4r}\cdot 2^{4r+1} + k\big)\cdot |\allImp|$ components.
    If so, then all the components get marked yellow, otherwise, $T(\Imp(\natBeh(C)))$ gets added to $W$.

    After finishing the marking, either $W$ was expanded or the unmarked components are removed in $\O(n)$ time (as $\wFn(\natBeh(C))$ was already computed previously).
    It remains to use \Cref{lem:magic_kernel} and the theorem follows.
\end{proof}

\section{No Polynomial Kernel with Respect to Modulator to Disjoint Cycles}\label{sec:NoTPKwrtDisjCycles}

In this section, we exclude, under reasonable theoretical assumptions, existence of the polynomial (Turing) kernel with respect to the size $k$ of a modulator to disjoint cycles. Modulator to disjoint cycles can be also seen as a modulator to graphs of treewidth two. We provide a polynomial parameter transformation from \textsc{Multicolored Clique} parameterized by $k \log n$, which is known to be \WKc~\cite{HermelinKSWW15}.
To simplify our polynomial parameter transformation, we use a gadget notation. We start with a simple general-purpose gadget for selection.

\subparagraph*{Selection-gadget.} The basic building block of our \emph{selection-gadget} is a \emph{cherry}. A cherry $X$ is a simple graph with $V(X)=\{x^0,x^1,x^*\}$ and $E(X) = \{\{x^0,x^*\},\{x^1,x^*\}\}$. The vertices $x^0$ and $x^1$ are non-terminals and we call them ports. On the contrary, the vertex $x^*$ is a terminal. A selection gadget $\mathbb{S}$ of length $\ell$ consist of $\ell$ cherries $X_1,\ldots,X_\ell$ connected in the following way. For every $i\in\{1,\ldots,\ell-1\}$ we add edges $\{x_i^*,x_{i+1}^0\}$ and $\{x_i^*,x_{i+1}^1\}$. To complete the construction of the selection gadget, we add edges $\{x_\ell^*,x_{1}^0\}$ and $\{x_\ell^*,x_{1}^1\}$. In other words, the selection gadget is a cyclic connection of cherries.
\smallskip

\begin{lemma}\label{lem:selection_gadget}
	Let $\mathbb{S}$ be a selection-gadget of length $\ell$, where $\ell \geq 3$, with cherries $X_1,\ldots,X_\ell$. An optimal solution of the instance of the \sTSP problem on graph $\mathbb{S}$ is of size $2\ell$ and for every $i\in\{1,\ldots,\ell\}$ the solution visits either $x_i^0$ or $x_i^1$ but not both.
\end{lemma}
\begin{proof}
	As each terminal needs to be incident with two traversals of the solution and as no terminals are adjacent it is clear that the optimum walk is of size at least $2\ell$.
	A solution which takes at most $2\ell$ edges has each terminal incident with exactly two solution traversals.
	Assume, for a contradiction, that for some $x_i^*$ both of the solution traversals go to the `same side', e.g., both go without loss of generality (due to symmetry) to $x_i^1$ or $x_i^0$.

	If one edge is traversed twice, e.g., $\{x_i^*,x_i^1\}$, then the solution walk needs to continue from $x_i^1$ further to reach other terminals.
	It needs to traverse the edge $\{x_i^1,x_{i+1}^*\}$ twice, and no other edge traversal may be done incident with $x_{i+1}^*$ as two were already spent.
	In that case, other terminals are never reached, a contradiction.
	We proceed similarly for two traversals of $\{x_i^*,x_i^0\}$.

	If both edges $\{x_i^*,x_i^1\}$ and $\{x_i^*,x_i^0\}$ are used once, then the walk must continue as well, using $\{x_i^1,x_{i+1}^1\}$ once and $\{x_i^0,x_{i+1}^1\}$ once.
	Again, the solution used two edges incident with $x_{i+1}^*$ so there is no hope to reach other terminals, a contradiction.

	Thus, for each terminal, the solution traverses one edge to $x_i^1$ or $x_i^0$ and a second edge to $x_{i-1}^1$ or $x_{i-1}^0$.
	To constitute a walk, these traversals need to join on a single non-terminal between each pair of consecutive terminals.
	This forms a closed walk over the terminals, using $2\ell$ edges and for each $i\in\{1,\ldots,\ell\}$ visiting exactly one of $x_i^1$ or $x_i^0$ but not both.
\end{proof}

\subparagraph*{Cycle-gadget.} A \emph{cycle-gadget} $\mathcal{B}$ of size $\ell$ is a cycle with $3\ell$ vertices $(\tau_{1}^1,\tau_{1}^2,\tau_{1}^3,\tau_{2}^1,\ldots,\tau_{\ell}^3)$. For every $i\in\{1,\ldots,\ell\}$ we call the vertices $\{\tau_{i}^1,\tau_{i}^2,\tau_{i}^3\}$ a \emph{triplet} (denoted $\tau_i$). All vertices of the cycle-gadget are terminals.
Let $\mathcal{B}$ be a cycle-gadget of size at least $1$ and $v$ be a vertex that is not part of $\mathcal{B}$. In our constructions, we are connecting vertices with triplets of cycle-gadgets, each triplet with exactly one vertex. Hence, if we say that we connect triplet $\tau_i$ of $\mathcal{B}$ with vertex $v$, then we add two edges $\{v,\tau_{i}^1\}$ and $\{v,\tau_{i}^2\}$. Note that the vertex $\tau_{i}^3$ has always degree two.

\thmsTSPWKcomplete*
\begin{proof}
	Let $\mathcal{I} = (G,k,c)$ be an instance of the \textsc{Multicolored Clique} parameterized by $k\log n$. At first, we extend the set of vertices of the graph $G$ by isolated vertices so that every color class contains exactly $N=2^{\lceil\log n\rceil}$ vertices. Note that these additional vertices do not change the solution of $\mathcal{I}$. We call graph vertices $v_a^i$, where $1\leq i\leq k$ and $0\leq a < N$. For the rest of the proof, let $G$ be an input graph with the additional vertices.

	Now, we construct an equivalent instance $\mathcal{I}'$ of the \sTSP problem. First, we introduce one selection-gadget of size $k\log N$ with cherries denoted by $X_{i,j}$, where $1\leq i \leq k$ and $1 \leq j \leq \log N$. This gadget introduces $\log N$ cherries for every color class. We use these cherries to select a specific vertex by encoding its number in binary. Next, for every \emph{non-edge} $\{v_a^i,v_{a'}^{i'}\}\not\in E(G)$, $i\neq i'$, we introduce a cycle-gadget $\mathcal{B}(\{v_a^i,v_{a'}^{i'}\})$ of size $2\log N$. Let $\operatorname{bit}(y,\ell)$ be a function returning value of the $\ell$-th bit of the binary representation of the positive integer $y$. For the first $\log N$ triplets, we connect every triplet $\tau_j$, where $j\in\{1,\ldots,\log N\}$, to the port $x_{i,j}^{\neg\operatorname{bit}(a,j)}$ of the selection-gadget. For the remaining triplets, we connect every triplet $\tau_{(\log N) + j}$, where $j \in \{1,\ldots,\log N\}$, to the port $x_{i',j}^{\neg\operatorname{bit}(a',j)}$ of the selection-gadget. To complete the construction, we set the budget value to $\budget = 2k\log N + (3\cdot 2\log N + 1)\cdot\left(\binom{k}{2}N^2 - |E(G)|\right)$.

	\begin{claim}\label{lem:cycle_gadget_works}
		Let $\mathbb{S}$ be a selection-gadget for vertices of $\mathcal{I}$ and $\mathcal{B}(\{v_{a}^i,v_{a'}^{i'}\})$ be a single cycle-gadget of size $2\log N$ connected to $\mathbb{S}$ according to the previous construction. Then any optimal solution $S$ of \sTSP traverses at least $3\cdot2\log N + 1$ edges incident with $\mathcal{B}(\{v_{a}^i,v_{a'}^{i'}\})$. Moreover, if $S$ traverses exactly $3\cdot 2\log N + 1$ edges, then there is a triplet $\tau_i$ and corresponding port connected to $\tau_i$ such that the part of $S$ incident with the cycle-gadget together with edges connecting $\tau_i$ to its corresponding port forms a closed walk.
	\end{claim}
	\begin{claimproof}
		Each terminal of the cycle-gadget must be visited, hence, each vertex in it has at least two incident edges in the solution.
		In total, we have exactly $3\cdot2\log N + L/2$ edges in the cycle-gadget, where $L$ is the number of edges which connect to the vertices of the selection-gadget.
		As the cycle-gadget must be entered and exited from the selection-gadget at least once, we have at least $3\cdot2\log N + 1$ edges in every cycle-gadget.

		If the number of traversed edges is exactly $3\cdot2\log N + 1$, then the cycle-gadget has exactly two outgoing edges $e_1,e_2$.
		The edges $e_1,e_2$ must be incident to the same port of the selection-gadget, otherwise some of the terminals would not be traversed.
	\end{claimproof}

	By \Cref{lem:selection_gadget,lem:cycle_gadget_works} the solution traverses the selection-gadget and visits exactly one vertex from $\{x_{i,j}^0,x_{i,j}^1\}$ for each cherry $X_{i,j}$. If all cycle-gadgets are connected to a $x_{i,j}^1$, then they may be traversed, if not, the solution does not exist.

	Now we prove that the instances are equivalent. First, for a contradiction, assume that $G$ does not contain a multicolored clique of size $k$, but we found a solution for our instance.
	We observe, that for each color $i$ the chosen port of $X_{i,j}$ for $1 \le j \le \log N$ represents one vertex of the original instance.
	As there is no multicolored clique of size $k$ in $G$ the solution must choose some pair of valuations representing vertices $v_{a}^i$ and $v_{a'}^{i'}$ such that there is no edge between them.
	As this is an non-edge there is a cycle-gadget $\mathcal{B}(\{v_{a}^i,v_{a'}^{i'}\})$ in our instance.
	Hence, such solution cannot traverse all terminals with the prescribed budget, which is a contradiction.

	In the opposite direction, let $\mathcal{I}$ be a \YESi with a solution $S=\{v_{a_1}^{1},v_{a_2}^{2},\ldots,v_{a_k}^{k}\}$ of size $k$ in $G$.
	Then, there is a solution $S'$ of $\mathcal{I}'$ such that it visits vertices $K=\{x_{i,j}^{\operatorname{bit}(a_i,j)}\mid j\in\{1,\ldots,\log N\},{i\in\{1,\ldots,k\}}\}$.
	Since $S$ is a complete graph there is at least one triplet for every cycle-gadget connected to at least one vertex of $K$. Hence, a solution of the weight $\budget$ exists.

	We presented a polynomial parameter transformation from \textsc{Multicolored Clique} parameterized by $k \log n$ to \sTSP where the parameter has size $3 k \log n$. Therefore, \sTSPshort is \WKh.
\end{proof}

\section{No Polynomial Kernel with Respect to Fractioning Number and Degree-Treewidth}\label{sec:NoPKwrtFN}
We already claimed the following result excluding existence of the polynomial kernel for the \TSP problem with respect to the number of structural parameters.
\lemTSPnoPKwrtFN*

We prove \Cref{lem:TSP:noPKwrtFN} using the following two auxiliary lemmas.
\begin{lemma}\label{lem:HP_cross_compose_TSP_fn}
	\textsc{Hamiltonian Path} AND-cross-composes into \TSPshort parameterized by the fractioning number of the input graph.
\end{lemma}
\begin{proof}
  Equivalence relation $\mathcal{R}$ considers as equivalent all malformed instances and all the well-formed instances with the same number of vertices. It is easy to see that this is a polynomial equivalence relation and we may assume that all input graphs have $k$ vertices.

  For $t$ equivalent instances $G_1=(V_1,E_1),\ldots,G_t=(V_t,E_t)$ of the \textsc{Hamiltonian Path} problem we construct an instance $\mathcal{I}=(G=(V,E),\wFn,\budget)$ of the \TSP as follows. We create the graph $G$ by taking the disjoint union of graphs $G_1,\ldots,G_t$ and adding a new apex vertex $v$ connected to all other vertices. Our weight function $\wFn$ assigns each edge $e\in E(G)$ the weight $1$ and we set the budget to be $\budget=t\cdot(k-1)+2t = t\cdot (k+1)$. We can clearly construct $(G,\wFn,\budget)$ in polynomial time with respect to the sum of all input instances and since $\fn(G_i) \leq k$ for every $i\in[t]$, it is immediate that $\fn(G)\leq k+1$.

  Assume first that there is a solution $P_i$ for every $G_i$, $i\in[t]$. We construct a solution of $\mathcal{I}$ by ``gluing'' solutions $P_i$ using the apex vertex $v$. The desired walk in $\mathcal{I}$ is a walk $W=v,P_1,v,P_2,\ldots,P_t,v$. Clearly, $W$ visits all vertices of $G$. Moreover, since the vertex $v$ is visited once before and once after each path $P_i$ and the length of each $P_i$ is exactly $k-1$, it follows that the $\sum_{e\in E(W)}e=t(k-1)+2t=t(k+1) = \budget$.

  In the other direction, let $W$ be a solution walk for the constructed \TSPshort instance $\mathcal{I}$. We observe that the solution visits every graph $G_i$ at least once and the cost for every such graph is at least $k - q_i + 2q_i = k+q_i$, where $q_i$ is the number of connected components in the intersection of the solution and~$G_i$.
  Since $q_i \ge 1$ for all $i \in [t]$ and since we have $\sum_{i = 1}^t (k+q_i) \le \budget = t \cdot (k+1)$, we get that $\sum_{i = 1}^t q_i \le k$.
  Therefore, $q_i = 1$ for all $i \in [t]$ and thus the intersection of the solution and~$G_i$ is a Hamiltonian path as desired.
\end{proof}

\begin{lemma}\label{lem:HP_cross_compose_TSP_degtw}
	\textsc{Hamiltonian Path} AND-cross-composes into \TSPshort parameterized by the combined parameter treewidth $\tw(G)$ and maximum degree $\Delta(G)$ of the input graph.
\end{lemma}
\begin{proof}
	To prove the lemma, we use a construction similar to one used in the proof of \Cref{lem:HP_cross_compose_TSP_fn}. The equivalence relation $\mathcal{R}$ remains the same and the instance $\mathcal{I}=(G=(V,E),\wFn,\budget)$ of the \TSP is constructed as follows. We start the construction of $G$ by taking disjoint union of $G_i$ for all $i\in[t]$. Then, we add a new vertex $v_i$ for every $i\in[t]$ and add an edge ${v_i,u}$ for every $u\in V(G_i)\cup V(G_{(i+1)\bmod t})$. The weight function $\wFn$ is constant function assigning each edge a weight $1$, and the budget remains $\budget = t\cdot(k-1)+2t = t\cdot (k+1)$. Note that since $\tw(G_i)\leq k$ and $\Delta(G_i) \leq k-1$ for every $i\in[t]$, it follows that $\tw(G) \leq k + 2$ and $\Delta(G) = 2k$ due to the added vertices $v_i$.

	Assume first that there is a solution $P_i$ for every $G_i$, $i\in[t]$. We construct a solution of $\mathcal{I}$ by ``gluing'' solutions $P_i$ using the added vertices $v_i$, $i\in[t]$. The desired walk in $\mathcal{I}$ is a walk $W=v_t,P_1,v_1,P_2,v_2,P_3,v_3,\ldots,P_t,v_t$. Clearly, $W$ visits all vertices of $G$. Moreover, for every vertex $v_i$, $i\in[t]$ the walk $W$ uses exactly two edges incident with $v_i$. The length of every $P_i$ is exactly $k-1$, hence $\sum_{e\in E(W)}e=t(k-1)+2t=t(k+1) = \budget$.

	In the opposite direction, let $W$ be a solution walk for the constructed \TSPshort instance $\mathcal{I}$. By definition, $W$ is a closed walk visiting each vertex at least once and thus each vertex is incident with even number of edges in $W$. Suppose that there is a vertex $v\in W$ such that $v$ is visited at least twice by $W$. Then the total weight of $W$ is at least $2(tn+t+1)/2 = t(n+1) + 1$, however the budget $\budget$ is $t(n+1)$. Consequently, $W$ visits every vertex exactly once and therefore is a path. It follows that a path $P_i = G_i\cap W$, where $i\in[t]$, is a Hamiltonian path in the original instance.
\end{proof}

\begin{proof}[Proof of \Cref{lem:TSP:noPKwrtFN}]
	\Cref{pro:nph_or_cross_composition} together with \Cref{lem:HP_cross_compose_TSP_fn} implies that \TSP parameterized by the fractioning number of the input graph does not admit a polynomial compression, unless $\NP\subseteq\coNPpoly$. The same holds for \Cref{pro:nph_or_cross_composition} and \Cref{lem:HP_cross_compose_TSP_degtw}, which completes the proof.
\end{proof}

\section{Conclusions}\label{sec:conclusions}

The core focus of this work is kernelization of the \TSP.
To stimulate further research in this area we would like to promote some follow up research directions.
Design of ``local'' rules might be impossible (as was mentioned in the case of feedback edge set number) and therefore we occasionally have to consider generalizations of this problem that give us more power when designing the reductions.
An interesting open problem that remains in this area is whether \TSPshort does admit a polynomial kernel with respect to the feedback vertex number.
Towards this we have provided several steps that suggest it should be possible to design a polynomial kernel.
It should be noted that in order to do so, one could try to ``lift'' our arguments to modulator to trees of any size.
Note that this is not possible for cycles.

Yet another interesting line of work is the one of $q$-path vertex cover ($q$-pvc).
This is a generalization of the vertex cover number---a graph $G$ has $q$-pvc of size $k$ if it has a modulator $M$ (with $|M|=k$) such that in $G \setminus M$ there is no path of length $q$.
It is not hard to see that some of our arguments can be applied for components of $G \setminus M$ that are stars (of arbitrary size). Therefore, it should be possible to give a polynomial kernel for $3$-pvc.
Thus, the question arises: Is there a polynomial kernel with respect to $r$-pvc for constant $r$?
In this case we are more skeptical and believe that the correct answer should be negative.

\bibliography{wayprout}

\end{document}